\documentclass[11pt]{amsart}

\usepackage[babel,threshold=2]{csquotes}
\usepackage{amsfonts}
\usepackage{amssymb}
\usepackage{enumerate}
\usepackage{amsmath}
\usepackage{amsthm}
\usepackage{graphicx}
\usepackage{color,soul}
\usepackage[dvipsnames]{xcolor}
\usepackage[english]{babel}
\usepackage[doi=false,backend=biber,style=numeric,sorting=none]{biblatex}
\usepackage{comment}
\usepackage{tensor}
\usepackage{hyperref}
\usepackage{multirow}
\usepackage{subcaption}
\usepackage{array}
\usepackage{float}
\usepackage{tikz}
\usepackage{bbm}

\usetikzlibrary{positioning,fit,backgrounds,shapes.geometric}

\numberwithin{equation}{section}

\newtheorem{theorem}{Theorem}[section]

\newtheorem{claim}[theorem]{Claim}

\newtheorem{corollary}[theorem]{Corollary}

\newtheorem{fact}[theorem]{Fact}

\newtheorem*{T:equiv}{Theorem~\ref{t:equiv}}
\newtheorem*{C:prec}{Corollary~\ref{c:prec}}

\theoremstyle{definition}
\newtheorem{example}[theorem]{Example}
\newtheorem{remark}[theorem]{Remark}
\newtheorem{definition}[theorem]{Definition}

\DeclareMathOperator{\Tr}{Tr}

\DeclareMathOperator{\PG}{PG}

\newcommand{\RR}{\mathbb{R}}

\newcommand{\CC}{\mathbb{C}}

\newcommand*{\defeq}{\stackrel{\text{def}}{=}}
\newcommand{\RNum}[1]{\uppercase\expandafter{\romannumeral #1\relax}}

\renewcommand{\epsilon}{\varepsilon}

\begin{document}

\title[Beyond Gaussian Entanglement]{Beyond Gaussian Entanglement: Practical Criteria for Polynomially Excited Quantum States}

\author[G. Homa]{G\'abor Homa}
\address{HUN-REN Wigner Research Centre for Physics, P.~O.~Box 49, H-1525 Budapest, Hungary}
\email{homa.gabor@wigner.hun-ren.hu}

\author[R. Balka]{Rich\'ard Balka} 
\address{HUN-REN Alfr\'ed R\'enyi Institute of Mathematics, Re\'altanoda u.~13--15, H-1053 Budapest, Hungary, AND Institute of Mathematics and Informatics, Eszterh\'azy K\'aroly Catholic University, Le\'anyka u.~4, H-3300 Eger, Hungary}
\email{balkaricsi@gmail.com}

\author[N. N\'emet]{Nikolett N\'emet}
\address{HUN-REN Wigner Research Centre for Physics, P.~O.~Box 49, H-1525 Budapest, Hungary}
\email{nemet.nikolett@wigner.hun-ren.hu}

\author[A. Bodor]{Andr\'as Bodor}
\address{HUN-REN Wigner Research Centre for Physics, P.~O.~Box 49, H-1525 Budapest, Hungary, AND Institute of Mathematics and Informatics, Faculty of Science,
University of P\'ecs, H-7624 P\'ecs, Ifj\'us\'ag \'utja 6, Hungary}
\email{bodor.andras@wigner.hun-ren.hu}

\author[T. Kiss]{Tam\'as Kiss}
\address{HUN-REN Wigner Research Centre for Physics, P.~O.~Box 49, H-1525 Budapest, Hungary}
\email{kiss.tamas@wigner.hun-ren.hu}

\author[A. Csord\'as]{Andr\'as Csord\'as}
\address{Department of Physics of Complex Systems, E\"{o}tv\"{o}s Lor\'and University, P\'azm\'any P\'eter s\'et\'any 1/A, H-1117 Budapest, Hungary}
\email{csordas@tristan.elte.hu}


\keywords{  oscillator systems, physical rule, entanglement, trace-class, integral operators, density operators, quantum theory, polynomial excitation, experimental realization}

\begin{abstract}

 Polynomial excitations consisting of creation and annihilation operators are studied. Subsystem entanglement is proven produced and preserved by such manipulations.
NPT entanglement already present between any two subsystems of a general pure or mixed multipartite Gaussian state is preserved under all possible polynomial excitations. Nevertheless these excitations 
can create NPT entanglement between initially separable subsystems. Thus it is implied that  in this case the number of NPT entangled subsystem pairs cannot decrease.

Polynomial excitations of a wide class of initially separable, special, and physically important pure states are considered. For these states, the method introduced in this paper reduces the infinite-dimensional positivity problem to a finite-dimensional one. Specifically, after partial transposition, NPT entanglement is equivalent to a corresponding  bilinear representation matrix of the polynomial not being positive semidefinite.

Furthermore, by using symplectic diagonalization and spectral decomposition, an explicit transformation is constructed between the Fock-space representation  and position representation. It 
is then applied to a highly specific and nontrivial example involving quadratic excitations of a three-mode mixed Gaussian state. In a special case, the behavior of the NPT entanglement is tracked between the subsystems of the chosen system as a function of a continuous parameter $\kappa$, which tunes the non-locality of the excitation. Quantum entanglement negativity is evaluated as a function of the $\kappa$ parameter. Under a specifically chosen quadratic Hamiltonian, the resulting unitary time evolution maintains NPT entanglement for long intervals, interrupted only by short periods of PPT states. Finally, a closed analytical expression is derived for the covariance matrix of a general quadratic-polynomial Gaussian operator for an arbitrary number of degrees of freedom in phase space. 

In summary, our theoretical framework directly connects mathematical foundations, theoretical modeling, and concrete, experimentally and technologically feasible implementations.

\end{abstract}

\maketitle
\begin{center}
   \today 
\end{center} 
\section{Introduction}

Characterizing quantum entanglement in continuous-variable multipartite quantum systems is a highly non-trivial task. While the rigorous foundational framework for the positivity and entanglement of polynomial Gaussian integral operators was recently established in Ref.~\cite{japcsi}, classifying quantum correlations in  complex, multi-mode systems remains a long-standing challenge.  In this paper, we extend our previous approach into an exact, analytical methodology of entanglement detection valid for an arbitrary finite number of degrees of freedom and subsystems. We show that this infinite-dimensional spectral problem can be solved without any approximations in many important cases, providing a complete and rigorous description of continuous-variable quantum correlations.

Our results apply to all physical systems and dynamics that preserve the Gaussian or polynomial Gaussian nature of the density operator.
Such dynamics appear ubiquitously: atomic and nuclear vibrations, phonons in crystal lattices, molecular vibrations, quantized electromagnetic modes of optical cavities, and small fluctuations of Bose--Einstein condensates \cite{sakurai2017modern, walls2008quantum, weinberg1995quantum}. The harmonic approximation holds when deviations from equilibrium are small and nonlinear effects are negligible. Beyond this regime, anharmonicity, interactions, measurement back-action, or intrinsic nonlinearities give rise to non-Gaussian phenomena, including for example quantum chaos and higher-order correlations \cite{haake2010quantum, girardeau1960relationship}. Of fundamental importance, oscillators are also extensively studied within the framework of open quantum systems \cite{vernon,HPZ,H-Y,Bernád2018,HBCSCS,HCSB,japcsi2osc}, a context where they can serve to explore foundational aspects of quantum gravity (see for example, \cite{universe10080306}).

The oscillator paradigm extends to fundamental and effective field theories. In quantum field theory, a free scalar field is equivalent to an infinite set of independent harmonic oscillators, with each mode representing an oscillator degree of freedom \cite{peskin2018introduction, weinberg1995quantum}. Second quantization generalizes this concept to bosonic and fermionic modes, incorporating quantum statistics and linear response \cite{fetter1971quantum}. Low-energy effective theories in condensed matter often rely on oscillator-like collective modes, such as phonons as Goldstone bosons, or other emergent quasiparticles \cite{altland2010condensed, zinn2002quantum}. In experiments, these systems are realized with high precision using trapped ions, optical and magnetic traps, nanomechanical and superconducting resonators, and cavity QED setups \cite{leibfried2003quantum, bloch2008many, aspelmeyer2014cavity, gu2017microwave, haroche2006exploring}. Their well-defined spectra, controllability, and low noise make them ideal platforms for fundamental studies and quantum technologies, where our exact entanglement verification method can be directly utilized \cite{weedbrook2012gaussian, adesso2014continuous, caves1981quantum, braunstein2005quantum, genoni2010quantifying}.

\subsection{Experimental Platforms and Applications}

Gaussian quantum systems form a central framework in continuous-variable quantum physics, quantum optics, and quantum information theory \cite{weedbrook2012gaussian,adesso2014continuous}.  A laser incorporated in many experiments, produces Gaussian quantum states to a good approximation. Most established components of 
a given setup preserve the Gaussianity of the state in phase space (Table~\ref{tab:cv_operations}.):

\begin{enumerate}
\item With a suitable combination of beam splitters,
any orthogonal Bogoliubov transformation of the input modes becomes feasible (see eqs. \eqref{eq:gauss_operator},\eqref{eq:Shat_operator}, and \eqref{eq:tilde_creation_annihilation_def}).
\item  Optical parametric ampliﬁers (like nonlinear crystals) 
can squeeze the input states.
\item  Beamsplitters with a strong coherent field can realize the displacement operator ($\hat{\mathbf{D}}_\mathbf{r}$) \cite{serafini2024quantum}. 

\end{enumerate}

Beam splitter and squeezing operations are presented as symplectic transformations. When a combination of these transformations act non-locally (which is the typical case), they can modify the entanglement between the involved modes. In contrast, the displacement operation, which corresponds to a translation in phase space, does not change the entanglement. However, in experiments,
it is essential for handling coherent states.
Since the displacement operation leaves the entanglement and separability between modes or subsystems invariant, it can be safely omitted from the present analysis without loss of generality.

	\begin{table}[htbp]
		\centering
		\caption{Comparison of some important continuous-variable quantum operations}
		\label{tab:cv_operations}
		\resizebox{\textwidth}{!}{
			\begin{tabular}{|>{\raggedright\arraybackslash}p{3.5cm}|
			                 >{\raggedright\arraybackslash}p{3.7cm}
			                |>{\raggedright\arraybackslash}p{3.7cm}
			                |>{\raggedright\arraybackslash}p{2.2cm}
			                |>{\raggedright\arraybackslash}p{3.0cm}|}
				\hline
				\textbf{Physical operation, device} & \textbf{Geometric effect in phase space} & \textbf{Matrix/Operator type (phase space)} & \textbf{Photon number} & \textbf{Purity of $\hat{\rho}$ } \\ \hline
				\textbf{Displacement} & Displacement, linear shift & Affine transformation & Changes & \textbf{Purity preserved} ($\text{Tr}\hat{\rho}^2 = \text{Tr}\hat{\rho}^{'2}$) \\ \hline
				\textbf{beam splitter} & Rotation between modes & Symplectic (orthogonal) & Conserves & \textbf{Purity preserved} ($\text{Tr}\hat{\rho}^2 = \text{Tr}\hat{\rho}^{'2}$) \\ \hline
				\textbf{Squeezing} & Scaling (squeezing / stretching) & Symplectic  (hyperbolic) & Changes & \textbf{Purity preserved} ($\text{Tr}\hat{\rho}^2 = \text{Tr}\hat{\rho}^{'2}$) \\ \hline
                \textbf{Amplitude damping} &  Contraction towards the origin / scaling of fluctuations & Acts in Gaussian channel  & Decreases & \textbf{Can become mixed} ($\text{Tr}\hat{\rho}'^2 \le 1$) \\ \hline
                \textbf{Phase damping} & Phase randomization / angular blurring & Acts in decoherence channel & Conserves & \textbf{Can become mixed} ($\text{Tr}\hat{\rho}'^2 \le 1$) \\ \hline
               	\textbf{Polynomial excitation} & Non-linear distortion (Wigner function ripples / negative regions) & Non-Gaussian operation & Alters (depends on polynomial) & \textbf{Can become mixed} ($\text{Tr}\hat{\rho}'^2 \le 1$) \\ \hline
				\end{tabular}
		}
	\end{table}

Accounting for losses, each channel can be represented in its own way. Amplitude damping for example can be modeled using beam splitters with one output beam ignored.
In this case, the system is no longer closed and in general evolves into a mixed Gaussian state.
Phase damping describes the loss of quantum coherence due to phase noise. In quantum optics, it can be implemented using phase modulators or by coupling the system to a noisy environment. It is non-unitary and does not necessarily preserve Gaussianity 
\cite{Chen2006PhysRevASimultaneousAmplitudeAndPhaseDamping,Hiroshima2001PhysRevADecoherence,BreuerPetruccione2002OpenQuantumSystems}.

\href{}{}    

The processes listed above (except for phase damping and polynomial excitation) constitute the framework within which starting from the vacuum state of the input modes,
in principle, any Gaussian state (and only that) can be realized. 
A beneficial property of Gaussian states is that they can be completely characterized by their first and second statistical
moments, while all higher-order moments are fixed functions of these quantities \cite{serafini2024quantum,adesso2007entanglement}.
This structural simplicity also underlies well-known limitations of Gaussian resources
for quantum information processing; in particular, non-Gaussian resources are required for universal
continuous-variable quantum computation \cite{Lloyd1999QuantumComputation,Bartlett2002ClassicalSimulation}. This shows up for example in
the impossibility of entanglement distillation using only Gaussian resources and operations \cite{eisert2002distilling,fiurasek2002gaussian,giedke2002characterization}. Nevertheless, creation of generic Gaussian states forms the foundation of non-Gaussian state preparation.

Contrary to the need for non-linearity in discrete variable quantum computation, the need for non-Gaussianity does not come from the impossibility of entanglement in Gaussian states. For example, general squeezed states show entanglement between different optical modes. There is, however, an important distinction between bound and distillable entanglement. In continuous-variable systems, entanglement is often analyzed using the positive partial transpose (PPT) criterion \cite{Perescrit,horodecki1996separability}. In this case, the partial transpose of a bipartite density operator is taken in order to see whether the obtained operator is positive semidefinite or not. A PPT state may be separable or PPT entangled; in the latter case its entanglement is bound and cannot be distilled. A negative eigenvalue of the partial transpose is sufficient to establish that the state is NPT and hence entangled. NPT is a strong signature of operationally useful entanglement: in particular, every bipartite NPT Gaussian state is distillable~\cite{giedke2001distillability}. Beyond the Gaussian setting, however, whether every NPT state is distillable remains an open problem~\cite{Horodecki2009}. For Gaussian states, the PPT criterion provides a partial characterization of bipartite entanglement in several relevant settings \cite{SimonR,werner2001bound}.
In multimode systems, the distribution of PPT and non positive partial transpose (NPT) bipartitions depends sensitively on the structure of the state and on the operations applied \cite{adesso2014continuous, He2026}.

 Preparing and manipulating non-Gaussian states experimentally are much more difficult than
 Gaussian ones. However, there were significant advances
in this area
in the last 25 years. Reviews of this field can be found in
\cite{lvovsky2020productionapplicationsnongaussianquantum}, 
\cite{BIAGI2022100414} and \cite{Walschaers2021PRXQuantum.2.030204}.
Within quantum optics, there are several techniques for creating non-Gaussian states, such as
 warm atomic vapour \cite{Mika2022npjQuantumInfSingleModeQuantumNonGaussianLight},
 cavity QED \cite{Deleglise2008NatureReconstructionOfNonClassicalCavityFieldStates,Nogues2000PhysRevAMeasurementOfANegativeValueForTheWignerFunction},
nonlinear superconducting circuits \cite{Hofheinz2008NatureGenerationOfFockStates,Vlastakis2013ScienceDeterministicallyEncodingQuantumInformation,ChangPhysRevXObservationOfThreePhotonSpontaneousParametricDownConversion},
 optomechanics \cite{Rakhubovsky2024ProgressInQuantumElectronicsQuantumNonGaussianOptomechanics} and within these areas,
 polynomial excitations of phonons, i.e. motional degrees of freedom in trapped ion chain \cite{Meekhof1996,Schneider2012,Fluhmann2019encoding} and  quantum optomechanics \cite{Enzian2021}, magnon excitations in ultracold gases \cite{bloch2008many}  and solid-state magnetism  \cite{Noel2025nonlinear}. 
It is a common experience that experimentally prepared states can be expressed as a finite number
of creation and annihilation operators acting on a Gaussian state \cite{Takase2026,Thapliyal2024,NeergaardNielsen2011b} (more details in Appendix \ref{Appendix_photon_add&subtract}.
In these cases - even though the system is described by an infinite-dimensional Hilbert space -
the tools developed in this article may help determine the presence and type of quantum entanglement
between modes. The results in Section~\ref{sub:GKLC_excitations}
(cf. Eq.~\eqref{eq:quadratic_ploynomial_gaussian}-\eqref{eq:example_Fock_repr})
and Appendix~\ref{sec:appandix_A} help
express a state with photon additions and subtractions as
a Gaussian times a polynomial in $\mathbf x$ and $\mathbf y$
(in position representation).
Then with the help of theorems in Section~\ref{sec:theorems_for_entanglement}
it can be determined whether the state is NPT with respect to a chosen bipartition, allowing an experimenter to adjust the parameters when designing a new setup. (As an example an experimenter may wish
to create a scenario like in Figure~\ref{fig:experiment_scenario} where
a polynomial excitation creates entanglement between previously
separable subsystems.)

This Paper mostly focuses on non-Gaussianity generated by polynomial excitation (Eq.~\eqref{eq:op_symbolically_1} and Eq.~\eqref{eq:op_symbolically}). Polynomial excitations are  operations that generally produce non-Gaussian quantum states \cite{Filippov2019, Walschaers2017,Lee2010,Fiurasek2009,Olivares2006} and may either generate or preserve entanglement \cite{Cao2026,Lee2011}, depending on the form of the polynomial and the input state. Photon addition and subtraction have many useful applications \cite{Dodonov2009,Sharma2026}, therefore they represent a topical field of research.
Analyzing entanglement, a solid mathematical argument is given for choosing the positivity detection of the partial transpose, showing, among other things, that in case of polynomial excitation of separable pure states, this can be reduced to a linear algebraic problem.   
 This is especially beneficial when analyzing or designing a setup, as verifying NPT entanglement in an infinite dimensional Hilbert space is a tedious task.
 This approach simplifies the verification for polynomially excited special separable pure states by making it finite and simple.  To prove that these states are NPT entangled, 
 one only needs to calculate the eigenvalues of a finite matrix. If a negative value appears, it is a direct proof of entanglement. This gives a fast and analytical tool to test new experimental setups. 
 
 Therefore, the Paper responds to such non-trivial, but related questions as:  Are there universal rules governing quantum entanglement between subsystems under polynomial excitations? What happens with the entanglement already present in a Gaussian state as a result of polynomial excitation?  Does the state behave in the same way under local and nonlocal photon addition and subtraction? How can one efficiently detect entanglement in these settings? How can the entanglement structure evolve as a result of unitary time evolution? 

Polynomial excitations are experimentally accessible in several platforms. In quantum optics, they are commonly implemented using conditional schemes, where measurement outcomes herald the application 
of an effective non-Gaussian operation \cite{zavatta2004quantum,parigi2007probing}. In superconducting circuit systems, intrinsic nonlinear elements provide direct access to non-quadratic interactions
\cite{kirchmair2013observation}. Related mechanisms also arise in trapped-ion and optomechanical systems through effective higher-order couplings \cite{rabl2011photon,kienzler2015quantum}.
These realizations demonstrate that polynomial excitations are physically meaningful operations rather than purely theoretical constructs.
From a physical perspective this class also includes  short-time Kerr-type
nonlinear interactions \cite{kirchmair2013observation}, multi-photon processes, and certain forms of measurement-induced nonlinear back-action \cite{filip2005measurement,fiurasek2009engineering}.

\begin{figure}
\includegraphics[page=1, angle=0, width=0.95\textwidth]{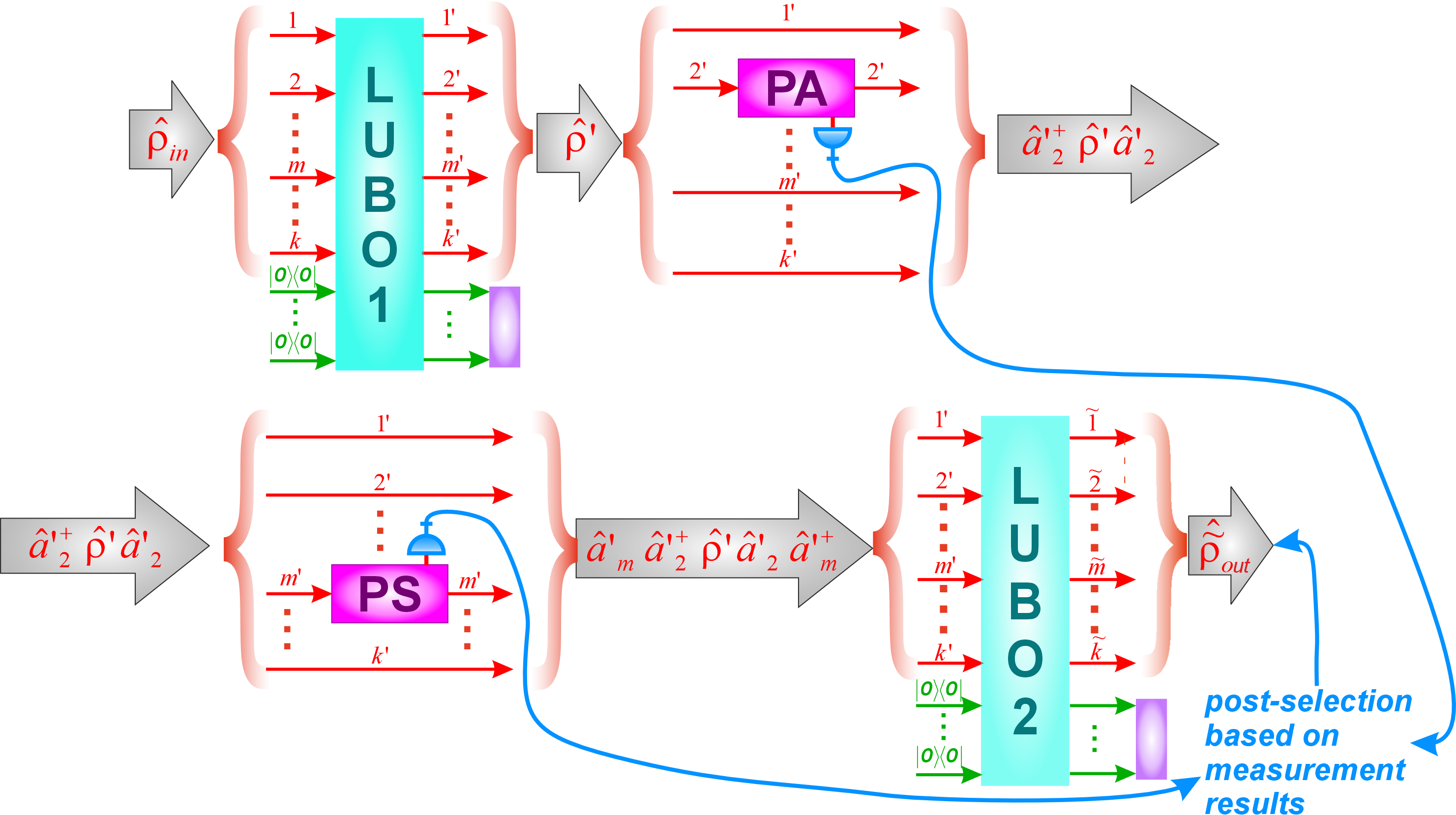}
\caption{
 Schematic block diagram of a possible experimental implementation of polynomial excitations. First, k input  modes  in Gaussian thermal states are subjected to a linear unitary Bogoliubov transformation (LUBO). To account for dissipative processes, optional vacuum modes can be added at the input and traced out at the output (green lines). Photons can then be added or subtracted (PA/PS) from selected modes, conditioned on the corresponding measurement outcomes. Finally, another Bogoliubov transformation can be performed.}
\label{fig:experiment_draft1}
\end{figure}

Figure \ref{fig:experiment_draft1}.~depicts  general schematics
of possible preparation of non-Gaussian states.
From the theoretical point of view it
can be modeled as a sequence of optical manipulations preserving the Gaussian character (linear unitary Bogoliubov transformation (LUBO) (e.g. beam splitters), squeezing transformation (parametric amplifiers), etc.)  \cite{Roh2026} 
and non-Gaussian excitations (usually
a photon addition or subtraction from a single mode \cite{Mogyorosi2019PhysRevA.100.013851,Liu_2015,Xu_2015,Tang_2015}).
With suitable components, a broad class of polynomial Gaussian
state can be realized, even experimental imperfections can be incorporated if they are kept sufficiently low.  See for example \cite{Tatsi2022PhysRevA.105.053701} for an experiment
and \cite{Saidi2026} for a possible metrology application.

As an example, let us review the experiment of Ra et al. \cite{Ra2020natureNonGaussianQuantum}. Instead of spatially separated beams, they manipulated time–frequency Hermite–Gaussian modes \cite{Ansari2018opticaTailoringNonlinearProcesses}. Initially, the states are prepared in a squeezed vacuum in a synchronously pumped optical parametric oscillator. A second beam, with its time–frequency mode engineered by a pulse shaper, is then used as a gate and recombined with the squeezed vacuum in a parametric upconversion process. When an upconverted photon is detected, it heralds the subtraction of one photon from the gate and vacuum beams. By carefully engineering the gate beam, a photon can be subtracted from an arbitrary, experimentally controlled superposition of modes. Finally, the Wigner function of the modes of the resulting state is measured with homodyne detection. The coherent superposition of up to four modes could be controlled; for example, a four-mode linear cluster state and a four-mode square cluster state were created \cite{Briegel2001PhysRevLettClusterStates}.

Building on the strong foundation of some of the author's previous work \cite{japcsi}, this paper analyzes entanglement in experimentally relevant polynomially excited states, including multipartite cases. It demonstrates, in a special but important case, the connection between the evolution of the Gaussian covariance matrix and the general covariance matrix used for entanglement detection. Furthermore, it provides a simple method for detecting entanglement based on linear algebra calculations for a wide class of   polynomial excited separable pure states.

Independently of the dynamics, the Paper identifies a physical invariance rule for polynomial excitations: if two Gaussian subsystems of an arbitrary mixed or pure Gaussian state are initially NPT  entangled, their NPT  entanglement is preserved under polynomial excitations (see Claim~\ref{c:mainfact1} based on \cite{japcsi}).  

A non-trivial example demonstrates that polynomial excitations can generate NPT entanglement between initially separable subsystems. Thus, it is shown that the number of NPT entangled pairs of subsystems cannot decrease under polynomial excitations  of Gaussian systems. This rule can be very important for example from the point of view  of quantum technology, including quantum computing and network setups \cite{weedbrook2012gaussian,braunstein2005quantum}. Every entangled pair of subsystems can serve as a fundamental resource or communication channel \cite{Lloyd1999QuantumComputation}. Since polynomial excitations cannot decrease the number of entangled pairs, they 
preserve all existing channels and can even build new ones between independent modes (see Figure~\ref{fig:experiment_scenario})   \cite[cf.][]{eisert2002distilling,Ourjoumtsev2007}. This property is also useful against environmental noise destroying quantum correlations \cite{serafini2024quantum}, as shown in Table~\ref{tab:cv_operations}. 
It enables robust information processing 
without losing the original quantum correlations.  The  Paper also presents nontrivial examples of this rules 
analyzed in detail.

The Paper aims to provide a complete study of non-Gaussianity, polynomial excitations and entanglement detection with concrete practical proposals, mathematical foundations and calculated specific examples. The remainder of it is organized as follows. Section~\ref{sec:theor_bg}. presents the most important mathematical and theoretical physical foundations for comprehensiveness. 
In Section~\ref{sec:pol_exc_pure}., the concept of bilinear representation matrix is introduced and applied for efficient entanglement detection for a  wide class of density operators. Section ~\ref{sec:theorems_for_entanglement}. discusses our universal results on the entanglement of arbitrary harmonic oscillator systems. Section~\ref{sec:examples}. presents example calculations illuminating the theoretical results. Additionally, we use analytical and numerical methods to explore scenarios beyond the scope of our theoretical framework. Finally, appendices for relevant sections can be found at the end of the Paper.

\section{Theoretical background}
\label{sec:theor_bg}

\subsection{Mathematical background.}

Let $L^2\left(\mathbb{R}^{N}\right)$ denote the Hilbert space of the complex-valued square integrable functions defined on $\RR^N$ with the scalar product $\langle f,g \rangle=\int_{\RR^N} f^{*}(\mathbf x)g(\mathbf x) \, \mathrm{d} \mathbf x$, where $z^{*}$ is the complex conjugate of $z$. A kernel $K \in L^2(\mathbb{R}^{2N})$ defines an integral operator 
\begin{equation*}
\widehat{K}\colon L^2\left(\mathbb{R}^{N}\right)\to L^2\left(\mathbb{R}^{N}\right) \quad 
\textrm{as} \quad\left( \widehat{K} f \right) (\mathbf x) = \int_{\mathbb{R}^{N}}\, K(\mathbf x,\mathbf y) f(\mathbf y)\, \mathrm{d}\mathbf y, 
\end{equation*}
see e.\,g.~\cite{Neumann,Neumann1927,deGosson_book,Stein}. We say that $\widehat{K}$ is \emph{self-adjoint} if $\widehat{K}=\widehat{K}^\dagger$; for continuous kernels $K$ this is equivalent to $K(\mathbf y,\mathbf x)=K^*(\mathbf x,\mathbf y)$ for all $\mathbf x,\mathbf y\in \RR^N$. We call $\widehat{K}$ \emph{positive semidefinite} if $\langle \widehat{K} f, f\rangle \geq 0$ for all $f\in L^2(\RR^N)$. If $\widehat{K}$ is positive semidefinite, it is necessarily self-adjoint. As a \emph{Hilbert--Schmidt} operator $\widehat{K}$ is \emph{compact} \cite[Proposition 9.12]{deGosson_book}, see \cite{deGosson_book} for the definitions. From now on we suppose that $\widehat{K}$ is self-adjoint. Since $\widehat{K}$ is compact and self-adjoint, the spectral theorem \cite[Theorem~3 in Chapter~28]{lax2002functional} yields that it has countably many eigenvalues $\{\lambda_i\}^\infty_{i=0}$, and there exists an orthonormal basis consisting of eigenvectors. Then $\widehat{K}$ satisfies a Fredholm-type integral equation
\begin{equation*}
\int_{\RR^N} \, K(\mathbf x,\mathbf y) \varphi_i(\mathbf y)\,\mathrm{d}\mathbf y=\lambda_i \varphi_i(\mathbf x), \quad \text{where }  \varphi_i\in L^2\left(\mathbb{R}^N\right). 
\end{equation*} 
We call $\widehat{K}$ \emph{trace-class} if $\sum_{i=0}^{\infty} |\lambda_i| < \infty$, and the \emph{trace} of $\widehat{K}$ is defined as
\begin{equation*}
 \Tr (\widehat{K})\defeq \sum_{i=0}^{\infty} \lambda_i. 
\end{equation*}
If $K$ is continuous, then by \cite{Duflo} we obtain the formula 
\begin{equation} \label{eq:trace}
\Tr (\widehat{K})=\int_{\RR^N} K(\mathbf x,\mathbf x) \, \mathrm{d}\mathbf x.
\end{equation}

Clearly, $\widehat{K}$ is positive semidefinite if and only if all of its eigenvalues are non-negative, equivalently
\begin{equation*} 
\int_{\mathbb{R}^{N}} \int_{\mathbb{R}^{N}} K(\mathbf{x},\mathbf{y}) f(\mathbf{x})f^{*}(\mathbf{y})  \, \mathrm{d} \mathbf{x} \, \mathrm{d} \mathbf{y} \geq 0 \quad \text{for all } f \in  L^2(\mathbb{R}^N).
\end{equation*}
The following theorem basically dates back to Mercer \cite{Mercer}.

\begin{theorem}[Mercer] \label{t:Mercer} Assume that $K\in L^2(\RR^{2N})$ is a continuous kernel. Then $\widehat{K}$ is positive semidefinite if and only if for all $\mathbf x_1,\dots, \mathbf x_k \in \RR^N$ and $c_1,\dots, c_k\in \CC$ we have 
\begin{equation} \label{eq:Mercer}  
\sum_{i,j=1}^k c_i c^{*}_j K(\mathbf x_i,\mathbf x_j)\geq 0. 
\end{equation} 
\end{theorem}

We say that $\widehat{\rho}$ is a \emph{density operator} if it is self-adjoint and positive semidefinite with $\Tr \big\{\widehat{\rho}\big\}=1$, that is, we have $\lambda_i \geq 0$ and $\sum_{i=0}^{\infty} \lambda_i=1$.

\begin{definition} \label{d:poly}
A self-adjoint kernel $K_G\in L^2(\RR^{2N})$ is \emph{Gaussian} if
\begin{align}
K_G(\mathbf{x},\mathbf{y})=&\exp \left\{-(\mathbf{x}-\mathbf{y})^{\top} \mathbf{A} (\mathbf{x}-\mathbf{y}) -i(\mathbf{x}-\mathbf{y})^{\top} \mathbf{B} (\mathbf{x}+\mathbf{y})\right.
\nonumber\\
&\quad \quad \left. -(\mathbf{x}+\mathbf{y})^{\top} \mathbf{C}(\mathbf{x}+\mathbf{y})-i\mathbf{D}^{\top}(\mathbf{x}-\mathbf{y})-\mathbf{E}^{\top}(\mathbf{x}+\mathbf{y})+F\right\},
\label{eq:gaussian_kernel}
\end{align} 
where $\mathbf{A},\mathbf{B},\mathbf{C}\in \RR^{N\times N}$ are real  matrices such that $\mathbf{A},\mathbf{C}$ are symmetric and positive definite, $\mathbf{E},\mathbf{D}\in \RR^N$, and $F\in \RR$. 

\end{definition}

\subsection{Quantum entanglement}
In this paper we consider density operators $\widehat{\rho}\colon L^2\left(\RR^N\right) \to L^2\left(\RR^N\right)$. For a summary of separability and entanglement in this model the readers might consult \cite{japcsi}. 
Now we only discuss separability and a reducing process in phase--space representation. We have position vector $\mathbf{x}\in \RR^N$ and momentum vector $\mathbf{p}\in \RR^N$ given as 
\begin{equation}
\mathbf{x}=(x_1,\dots,x_N)^{\top} \quad \textrm{and} 
\quad
\mathbf{p}=(p_1,\dots,p_N)^{\top}.
\end{equation}
Assume that $m$ subsystems are given with $\widehat{\rho}$, that partition the coordinates into $m$ pairwise disjoint sets. Thus we obtain a representation of the $k$th subsystem $(k=1,\dots,m)$ as
\begin{equation*}
\mathbf{x}^{(k)} = \bigl(x^{(k)}_1, x^{(k)}_2, \dots, x^{(k)}_{d_k}\bigr)^\top \in \mathbb{R}^{d_k}, 
\qquad
\mathbf{p}^{(k)} = \bigl(p^{(k)}_1, p^{(k)}_2, \dots, p^{(k)}_{d_k}\bigr)^\top \in \mathbb{R}^{d_k}.
\end{equation*}
If $\widehat{\rho}$ is separable with respect to this partition, then its Wigner function is of the form
\begin{equation}
W(\mathbf{x},\mathbf{p})
=\sum_{j=1}^{\infty} \pi_j
\prod_{k=1}^m
W^{(k)}_j\!\left(\mathbf{x}^{(k)},\mathbf{p}^{(k)}\right)
\label{eq:wigner_multi}
\end{equation}
with some $\pi_j\geq 0$ satisfying $\sum_{j=1}^{\infty} \pi_j=1$, where $W_j^{(k)}$ is a Wigner function of a density operator acting on the $k$th subsystem, see \cite{japcsi}.


We will single out two subsystems, and create a bipartite setup by integrating out the remaining subsystems and thus reducing the dimension of the model. Consider two subsystems with phase--space variables 
$\mathbf{x}^{(1)},\mathbf{p}^{(1)}\in\mathbb{R}^{d_1}$ and
$\mathbf{x}^{(2)},\mathbf{p}^{(2)}\in\mathbb{R}^{d_2}$, where $d_1+d_2\leq N$. Let $\mathbf{x}^{[3]},\mathbf{p}^{[3]}\in \RR^{N-d_1-d_2}$ denote the phase--space variables of the remaining subsystems. We restrict ourself to the first two subsystems and obtain the corresponding new Wigner function $W_r$ by taking the partial trace over the another subsystems, that is, 
\begin{equation}
W_r(\mathbf{x}^{(1)},\mathbf{p}^{(1)};\mathbf{x}^{(2)},\mathbf{p}^{(2)})
=
\int_{\mathbb{R}^{2(N-d_1-d_2)}} W(\mathbf{x},\mathbf{p})\,
\mathrm{d}\mathbf{x}^{[3]}\,\mathrm{d}\mathbf{p}^{[3]}.
\label{eq:partial_trace_phase_space}
\end{equation}

If $\widehat{\rho}$ is separable with respect to the original partition and \eqref{eq:wigner_multi} holds, then our reduced operator $\widehat{\rho}_r$ remains separable, and our new Wigner function satisfies
\begin{equation}
W_r(\mathbf{x}^{(1)},\mathbf{p}^{(1)};\mathbf{x}^{(2)},\mathbf{p}^{(2)})
=
\sum_{j=1}^{\infty} \pi_j
W^{(1)}_j(\mathbf{x}^{(1)},\mathbf{p}^{(1)})\cdot
W^{(2)}_j(\mathbf{x}^{(2)},\mathbf{p}^{(2)}).
\label{eq:red_wign_func}
\end{equation}

The partial transpose with respect to the first subsystem in this reduced model acts on the Wigner function as
\begin{equation}
\widetilde{W}_r\left(\mathbf{x}^{(1)},\mathbf{p}^{(1)};\mathbf{x}^{(2)},\mathbf{p}^{(2)}\right)
\defeq
W_r\left(\mathbf{x}^{(1)},-\mathbf{p}^{(1)};\mathbf{x}^{(2)},\mathbf{p}^{(2)}\right),
\label{eq:partial_transpose_phase_space}
\end{equation}
see e.\,g.~\cite{Perescrit,SimonR}.
Now the following well--known necessary condition connects separability to positive semidefiniteness.

\begin{theorem}[Peres--Horodecki criterion~\cite{Perescrit,horodecki1996separability}] \label{t:PH} If $\widehat{\rho}_r$ is separable, then $\widetilde{W}_{r}(\mathbf{x},\mathbf{p})$ defines a density operator.
\end{theorem}

A bipartite density operator $\widehat{\rho}_r$ is called PPT (\emph{positive partial transpose}) if the Wigner function of its partial transpose $\widetilde{W}_r$ corresponds to a valid density operator. Otherwise, $\widehat{\rho}_r$ is called NPT (\emph{non-positive partial transpose}).

The partial transpose operation in position representation transforms 
the reduced kernel $\rho_r(\mathbf{x},\mathbf{y})=\rho_r\big(\mathbf{x}^{(1)},\mathbf{y}^{(1)};\mathbf{x}^{(2)},\mathbf{y}^{(2)})$ to 
\begin{equation*}
\widetilde{\rho}_r(\mathbf{x},\mathbf{y})\defeq \rho_r\big(\mathbf{y}^{(1)},\mathbf{x}^{(1)};\mathbf{x}^{(2)},\mathbf{y}^{(2)}\big),
\end{equation*}
see the inverse Wigner--Weyl transformation in \cite[Eq.~6]{japcsi}.

\section{Polynomial excitation of initially a specific separable pure state}
\label{sec:pol_exc_pure}

In this section, we focus on some important properties of pure states and the density operators obtained from their polynomial excitations.
Pure states are very important for quantum technology and experimental physics.
They hold the most information and also keep out messy environmental noise.
Ladder operators create polynomial excitations.
Polynomially excited states are not Gaussian and this special feature is crucial.
These non-Gaussian properties let us build universal quantum computers and  also help fix quantum errors \cite{serafini2024quantum}.

Let $P(\mathbf{x},\mathbf{y})$ be a polynomial of $2N$ variables, then $P$ can be written in the form
\begin{equation*}
P(\mathbf{x},\mathbf{y})=\mathbf{v}^\top(\mathbf{x}) \mathbf{M} \mathbf{v}(\mathbf{y}),
\end{equation*}
where $\mathbf{v}(\mathbf{x})$ is a vector of distinct monomials (with coefficients $1$) of the variable $\mathbf{x}$,
and $\mathbf{M}$ is a coefficient matrix. We will call $\mathbf{M}$ a \emph{bilinear representation matrix } of $P$ with respect to $\mathbf{v}$.

Assuming $\mathbf{x}=(\mathbf{x}_1,\mathbf{x}_2)^\top$ and  $\mathbf{y}=(\mathbf{y}_1,\mathbf{y}_2)^\top$, after the partial transpose $\mathbf{x}_1 \leftrightarrow \mathbf{y}_1$ we obtain
\begin{equation*}
\widetilde{P}(\mathbf{x},\mathbf{y})= \mathbf{w}^\top(\mathbf{x})\widetilde{\mathbf{M}} \mathbf{w}(\mathbf{y}),
\end{equation*}
where $\mathbf{w}(\mathbf{x})$ is a vector of monomials (with coefficients $1$) of the variable $\mathbf{x}$,
and $\widetilde{\mathbf{M}}$ is a bilinear representation matrix of $\widetilde{P}$ with respect of $\mathbf{w}$.

\begin{example}
Let $P(\mathbf{x},\mathbf{y})=(x_1+i\kappa x_2)(y_1-i\kappa y_2)$, and
let us choose $\mathbf{v}(\mathbf{x})=(x_1,x_2)^\top$ and 
$\mathbf{M}=\begin{pmatrix}
1 & -i\kappa\\
i\kappa & \kappa^2
\end{pmatrix}$. Then we obtain 
\begin{equation*}
P(\mathbf{x},\mathbf{y})=(x_1+i\kappa x_2)(y_1-i\kappa y_2) = \mathbf{v}^\top(\mathbf{x}) \mathbf{M} \mathbf{v}(\mathbf{y}), 
\end{equation*}
and $\mathbf{M}$ is clearly a positive semidefinite Hermitian matrix. It is easy to see that partial transposition in first mode $x_1\leftrightarrow y_1$ yields
\begin{equation*}
\widetilde{P}(\mathbf{x},\mathbf{y})=(y_1+i\kappa x_2)(x_1-i\kappa y_2)= \mathbf{w}^\top(\mathbf{x})\widetilde{\mathbf{M}} \mathbf{w}(\mathbf{y}),
\end{equation*}
where 
\begin{equation*} \mathbf{w}(\mathbf{x})=(x_1,x_2,x_1x_2,1)^\top \quad \textrm{and} \quad  \widetilde{\mathbf{M}}= \begin{pmatrix}
1 & 0 & 0 & 0 \\
0 & \kappa^2 & 0  & 0  \\
0 & 0 & 0 &  i\kappa  \\
0 & 0 & -i\kappa & 0
\end{pmatrix}. 
\end{equation*}
It is easy to see that $\widetilde{\mathbf{M}}$ is a positive semidefinite Hermitian matrix if and only if $\kappa=0$.
\end{example}

\begin{fact} \label{f:pol}
Let $Q(\mathbf{x})$ be a nonzero polynomial of $n$ variables. Then the set of roots $\{\mathbf{x}: Q(\mathbf{x})=0\}$ has zero $n$-dimensional Lebesgue measure.     
\end{fact}

\begin{proof} 
We prove by induction on the number of variables $n$. If $n=1$ then we have finitely many roots in $\RR$ which has measure zero. Assume that the statement is true for $n-1$ and we prove it for $n$. Let $Z=\{\mathbf{x}: Q(\mathbf{x})=0\}$. We may assume that $x_1$ is in a monomial of $Q$ with nonzero coefficient. We can write 
\begin{equation*} Q(\mathbf{x})=\sum_{i=1}^M x_1^{i}P_i(x_2,\dots,x_n)
\end{equation*} 
with some polynomials $P_i$ of $n-1$ variables. Define 
\begin{equation*} S=\{(x_2,\dots,x_n): P_i(x_2,\dots,x_n)=0 \textrm{ for all } i\}.
\end{equation*}
Clearly, 
\begin{equation*} Z=\{\mathbf{x}\in Z: (x_2,\dots,x_n)\notin S\}\cup (\RR\times S),
\end{equation*}
so it is enough to prove that the sets $\{\mathbf{x}\in Z: (x_2,\dots,x_n)\notin S\}$ and $\RR\times S$ have zero $n$-dimensional Lebesgue measure.  First, if $(x_2,\dots,x_n)\notin S$ is fixed, then $Q(\mathbf{x})$ is a nontrivial polynomial in $x_1$, so it has at most $M$ roots. Therefore, Fubini's theorem implies that $\{\mathbf{x}\in Z: (x_2,\dots,x_n)\notin S\}$ has zero $n$-dimensional Lebesgue measure. By induction, $S$ has zero $(n-1)$-dimensional Lebesgue measure, so Fubini's theorem implies that $\RR\times S$ has zero $n$-dimensional Lebesgue measure.
\end{proof}

\begin{theorem}\label{c:homa_tetel}
Let $P(\mathbf{x},\mathbf{y})$ be a self-adjoint polynomial, and let $\mathbf{M}$ be a corresponding bilinear representation matrix with respect to $\mathbf{v}$, so $P(\mathbf{x},\mathbf{y})=\mathbf{v}^\top(\mathbf{x}) \mathbf{M} \mathbf{v}(\mathbf{y})$. Let $\rho(\mathbf{x},\mathbf{y})$ be the kernel of a pure density operator given in position representation. Assume that $P\rho\in L^2\left(\RR^{2N}\right)$. Then $P(\mathbf{x},\mathbf{y})\rho(\mathbf{x},\mathbf{y})$ is a positive semidefinite kernel if and only if $\mathbf{M}$ is a positive semidefinite matrix.
\end{theorem}
\begin{proof}
First assume that $\mathbf{M}$ is positive semidefinite. As $\hat{\rho}$ is a pure density operator, we can write  
\begin{equation} \rho(\mathbf{x},\mathbf{y})=u^{*}(\mathbf{x})u(\mathbf{y}).
\label{eq:pure_density_matrix}
\end{equation}

Then we obtain 
\begin{equation*}
K(\mathbf{x},\mathbf{y})\defeq P(\mathbf{x},\mathbf{y})\rho(\mathbf{x},\mathbf{y})= 
u^{*}(\mathbf{x}) \mathbf{v}^\top(\mathbf{x}) \mathbf{M} \mathbf{v}(\mathbf{y}) u(\mathbf{y}). 
\end{equation*}
Therefore, for all $f\in L^2(\mathbb{R}^N)$ we have 
\begin{align*}\label{eq:homa_tetel_1}
\langle \hat{K}f,f \rangle&=
 \int_{\mathbb{R}^{2N}} f^*(\mathbf{x})u^{*}(\mathbf{x})\mathbf{v}^\top(\mathbf{x}) \mathbf{M} \mathbf{v}(\mathbf{y}) u(\mathbf{y}) f(\mathbf{y}) \, \mathrm{d}\mathbf{x} \,\mathrm{d}\mathbf{y} \\
 &=\left(\int_{\mathbb{R}^{N}} f(\mathbf{x})u(\mathbf{x})\mathbf{v}(\mathbf{x}) \, \mathrm{d}\mathbf{x}\right)^{\dagger} \mathbf{M} \left(\int_{\mathbb{R}^{N}} f(\mathbf{y})u(\mathbf{y})\mathbf{v}(\mathbf{y}) \,  \mathrm{d} \mathbf{y}\right) \\
 &=\mathbf{w}^{\dagger} \mathbf{M} \mathbf{w} \geqslant 0, 
\end{align*}
where we used the positivity of $\mathbf{M}$ for the vector 
$\mathbf{w} \defeq \int_{\mathbb{R}^{N}} f(\mathbf{x})u(\mathbf{x})\mathbf{v}(\mathbf{x})\,  \mathrm{d}\mathbf{x}$. Thus $\hat{K}$ is positive semidefinite.

For the other direction, assume that $\rho(\mathbf{x},\mathbf{y})=u^{*}(\mathbf{x})u(\mathbf{y})$, and 
\begin{equation} K(\mathbf{x},\mathbf{y})\defeq P(\mathbf{x},\mathbf{y})\rho(\mathbf{x},\mathbf{y}) 
\label{eq:pure_dens_mat_with_polynom}
\end{equation}
is the kernel of the positive semidefinite operator $\hat{K}$. 

First we prove that the set of vectors 
\begin{equation*} \left\{u(\mathbf{x})\mathbf{v}(\mathbf{x}): \mathbf{x} \in \RR^N\right\}
\end{equation*}
spans $\CC^N$.
Assume to the contrary that this is not the case, then there is a nonzero vector $\mathbf{c}\in \CC^N$ orthogonal to all vectors $u(\mathbf{x})\mathbf{v}(\mathbf{x})$, that is,  
\begin{equation*}
\sum_{i=1}^N c_i^{*} u(\mathbf{x}) v_i(\mathbf{x})=u(\mathbf{x}) \sum_{i=1}^N c_i^{*} v_i(\mathbf{x})=0
\end{equation*}
for all $\mathbf{x}\in \RR^N$, where $v_i(\mathbf{x})$ denotes the $i$th coordinate of $\mathbf{v}(\mathbf{x})$. Note that $v_i(\mathbf{x})$ are different monomials of $\mathbf{x}$, so 
\begin{equation*} 
Q(\mathbf{x})\defeq \sum_{i=1}^N c_i^{*}  v_i(\mathbf{x})
\end{equation*}
is a nonzero polynomial. As $u(\mathbf{x})Q(\mathbf{x})=0$ for all $\mathbf{x}\in \RR^N$, we obtain that 
\begin{equation*} 
\{\mathbf{x}: u(\mathbf{x})\neq 0\}\subset \{\mathbf{x}: Q(\mathbf{x})=0\}.
\end{equation*}
By Fact~\ref{f:pol} the set $\{\mathbf{x}: Q(\mathbf{x})=0\}$ has zero $N$-dimensional Lebesgue measure, so $u(\mathbf{x})=0$ for Lebesgue almost every $\mathbf{x}\in \RR^N$. However, this implies that the operator $\hat{\rho}$ is the zero operator, which is a contradiction. 

Now let $\mathbf{z}\in \mathbb{C}^N$ be an arbitrarily fixed vector. In order to prove that $\mathbf{M}$ is positive semidefinite, we need to show that $\mathbf{z}^{\dagger} \mathbf{M} \mathbf{z}\geq 0$. Since the set $\left\{u(\mathbf{x})\mathbf{v}(\mathbf{x}): \mathbf{x} \in \RR^N\right\}$ spans $\CC^N$, there exist $c_1,\dots,c_N\in \CC$ and $\mathbf{x}_1,\dots,\mathbf{x}_N\in \RR^N$ such that $\mathbf{z}=\sum_{i=1}^N c_i u(\mathbf{x}_i) \mathbf{v}(\mathbf{x}_i)$. 
Now we have 
\begin{equation*} \mathbf{z}^{\dagger} \mathbf{M} \mathbf{z}= \sum_{i,j=1}^N c_i^{*} u^{*}(\mathbf{x}_i) \mathbf{v}^{\top} (\mathbf{x_i}) \mathbf{M} \mathbf{v}(\mathbf{x_j}) u(\mathbf{x}_j)
c_j=\sum_{i,j=1}^{N} c_i^{*} c_j K(\mathbf{x}_i,\mathbf{x}_j)\geq 0,
\end{equation*}
where the last inequality follows from Mercer's criterion (Theorem~\ref{t:Mercer}) applied to the positive semidefinite kernel $K$. This shows that $\mathbf{M}$ is positive semidefinite, and the proof is complete. 
\end{proof}

Theorem~\ref{c:homa_tetel} and Fact~\ref{t:PH} imply the following corollary. 
\begin{corollary} 
Let $P(\mathbf{x},\mathbf{y})$ and $\mathbf{M}$ be  as in Theorem \ref{c:homa_tetel} and let $\widehat{\rho}$ be a pure separable bipartite density operator  such that $P(\mathbf{x},\mathbf{y})\rho(\mathbf{x},\mathbf{y})$ is a bipartite density operator.
Let $\widetilde{\mathbf{M}}$ be a bilinear representation matrix after taking the partial transpose. Then $P(\mathbf{x},\mathbf{y})\rho(\mathbf{x},\mathbf{y})$ is NPT  if and only if $\widetilde{\mathbf{M}}$ is not positive semidefinite.
\end{corollary}

\begin{remark} A very important case in the theoretical and experimental physics and the quantum information theory is if $\rho(\mathbf{x},\mathbf{y})=u^{*}(\mathbf{x})u(\mathbf{y})$ is a pure Gaussian kernel of the form 
\begin{equation} \label{rho2}
u(\mathbf{x})=\sqrt{N_{G}}e^{-\mathbf{x^{\top}(A-\mathit{i}B)x+(E-\mathit{i}D)^{\top}x}},
\end{equation} 
where $\mathbf{A}, \mathbf{B}$ are symmetric real matrices such that $\mathbf{A}$ is positive, $\mathbf{D}, \mathbf{E}$ are real vectors, and $N_{G}$ is a normalizing factor.
\end{remark}

\begin{remark}
One direction of Theorem~\ref{c:homa_tetel} is easy to generalize: if $P(\mathbf{x},\mathbf{y})$ is a self adjoint polynomial with a positive semidefinite bilinear representation matrix, and $\rho(\mathbf{x},\mathbf{y})$ is a kernel of a density operator, then $P(\mathbf{x},\mathbf{y})\rho(\mathbf{x},\mathbf{y})$ is the kernel of a positive semidefinite operator. Indeed, by the spectral theorem $\rho(\mathbf{x},\mathbf{y})$ can be written as a positive linear combination of pure density kernels, so applying Theorem~\ref{c:homa_tetel} for the summands yields 
the positivity of the kernel $P(\mathbf{x},\mathbf{y})\rho(\mathbf{x},\mathbf{y})$. 
\end{remark}

The following example shows that the opposite direction of Theorem~\ref{c:homa_tetel} is almost sharp. Namely, even in dimension one, we cannot replace $\rho(x,y)$ by a mixed Gaussian kernel in Theorem~\ref{c:homa_tetel}. 

\begin{example}\label{ex:wrong_with_mixed}
Let 
\begin{equation*} P(x,y)=\frac 23 (x+y)^2-(x-y)^2+1 \quad \textrm{and} \quad   
K(x,y)=\exp\left(-\frac 52(x^2+y^2)+xy\right).
\end{equation*}
Clearly, $K(x,y)$ is the kernel of a mixed, positive semidefinite Gaussian operator. Our polynomial satisfies $P(x,y)=\mathbf{v}^\top(x) \mathbf{M}_P \mathbf{v}(y)$, where 
\begin{equation*} 
\mathbf{v}(x)=(x^2,x,1)^\top \quad \textrm{and} \quad   \mathbf{M}_P=\frac{1}{3}\begin{pmatrix} 0 & 0 & -1 \\ 0 & 10 & 0 \\ -1 & 0 & 3 \end{pmatrix},
\end{equation*} 
and $\mathbf{M}_P$ is not positive semidefinite as $\det(\mathbf{M}_P)=-\frac{10}{27}<0$. 

Finally, we show that $P(x,y)K(x,y)$ is the kernel of a positive semidefinite operator. We refer to \cite{Newton}, where a family of kernels $P_{\gamma} (x,y)K(x,y)$ with unit trace was considered as
\begin{equation*} 
P_{\gamma} (x,y)=(1+\gamma/2)^{-1}\left(\gamma (x+y)^2-(x-y)^2+1\right).
\end{equation*}
We need to prove that for $\gamma=\frac 23$ this is a kernel of a density operator. More general kernels are expanded in \cite[3.3]{Newton} in base $\phi_n$, giving an infinite band matrix with coefficients $\rho_{m,n}$, see formulas $(15)$, $(22)$, and $(23)$. Note that there is an $\frac 1N$ factor missing from $\rho_{m,n}$ in $(22)$, and there is a factor $\epsilon_0$ missing from $a_0$. Fix $\gamma=\frac 23$, in this special case we can check that all non-diagonal elements are zero, that is, $\rho_{m,n}=0$ if $m\neq n$. Then the diagonal elements $\rho_{n,n}$ give the eigenvalues, and we obtain the formula \begin{equation*} 
\rho_{n,n}=\frac{3}{4}(a_0+b_0n)\epsilon^n=\left(8-3\sqrt{6}+\left(12-4\sqrt{6}\right)n \right)\left(\frac{\sqrt{6}-2}{\sqrt{6}+2} \right)^n.
\end{equation*} 
Hence all eigenvalues are positive, showing that indeed $P(x,y)K(x,y)$ is the kernel of a positive semidefinite operator. 
\end{example} 

The above example shows that pureness is an important assumption in  Theorem~\ref{c:homa_tetel}.

\section{Polynomial excitations in harmonic oscillator systems}
\label{sec:theorems_for_entanglement}

In this section the polynomial-Gaussian kernel framework is applied to harmonic oscillator systems. A connection is established between arbitrary mixed or pure polynomial-Gaussian states and the corresponding  covariance matrices. Then the consequent universal physical rules for multipartite NPT  entanglement under polynomial excitation are demonstrated. 

\subsection{Polynomial-Gaussian Wigner functions}
 Let us consider the most general Gaussian multiplied with a general, quadratic polynomial in $\mathbf{r}=(\mathbf{x},\mathbf{p})^\top$:
\[
W_Q(\mathbf{r})=N_p \left(Z_0+ \mathbf{Z}_1^\top \mathbf{r}+\mathbf{r}^\top \mathbf{Z}_2 \mathbf{r} \right)\exp\Bigl(-(\mathbf{r}-\bar{\mathbf{r}})^\top \pmb{\sigma}_G^{-1} (\mathbf{r}-\bar{\mathbf{r}})\Bigr),
\]
where the proper normalization fixes $N_p$ as
\[
1=\int_{\mathbb{R}^{2N}} d\mathbf{r} \,W_Q(\mathbf{r})=N_p \left( Z_0+ \mathbf{Z}_1^\top \bar{\mathbf{r}}+\bar{\mathbf{r}}^\top \mathbf{Z}_2 \bar{\mathbf{r}}+\frac12 \Tr{\left(\mathbf{Z}_2 \pmb{\sigma}_G\right)} \right).
\]
The average $\mathbf{d}$ defined in \eqref{eq:def_average} with $W(\mathbf{r},t) \rightarrow W_Q(\mathbf{r})$ is
\begin{equation}
\mathbf{d}=\bar{\mathbf{r}}+N_p \left(\frac12 \pmb{\sigma}_G\mathbf{Z}_1+\pmb{\sigma}_G \mathbf{Z}_2 \bar{\mathbf{r}} \right) \equiv \bar{\mathbf{r}} +\pmb{\delta},
\end{equation}
and the covariance matrix $\pmb{\sigma}$ defined in \eqref{eq:def_cov_mat} with $W(\mathbf{r},t) \rightarrow W_Q(\mathbf{r})$ can expressed as
\[
\pmb{\sigma}_{i,j} = \left(\pmb{\sigma}_G \right)_{i,j}+N_p \left( \pmb{\sigma}_G\mathbf{Z}_2 \pmb{\sigma}_G \right)_{i,j} -2 \pmb{\delta}_i \pmb{\delta}_j,
\]

\begin{equation}
\pmb{\sigma} = \pmb{\sigma}_G + N_p \left( \pmb{\sigma}_G\mathbf{Z}_2 \pmb{\sigma}_G \right) -2\pmb{\delta} \pmb{\delta}^\top.
\label{eq:sigma_for_polygauss}
\end{equation}
The last term is a dyadic tensor.
These two expressions clearly show that the average and the covariance-matrix get contributions from the polynomial parts as well, and only for Gaussians ($\mathbf{Z}_1=\mathbf{0}$, $\mathbf{Z}_2={0}$) will the equalities  $\mathbf{d}=\bar{\mathbf{r}}$ and $\pmb{\sigma}=\pmb{\sigma}_G$ hold. 
\label{remark_1}

Why are these results relevant for quantum entanglement? In general, the partial transposition is implemented at the covariance matrix level by a sign inversion of the momentum components of the transposed subsystem. In this process the covariance matrix $\pmb{\sigma}$, defined in Eq.~\eqref{eq:def_of_cov_mat} is transformed. The covariance matrix corresponding to polynomial-Gaussian kernel $\pmb{\sigma}$ is different from that of the Gaussian part $\pmb{\sigma}_G$, see for example Eq.~\ref{eq:sigma_for_polygauss}. The partial transposition should be performed for the complete $\pmb{\sigma}$:
\begin{equation}
	\tilde{\pmb{\sigma}}=\pmb{\tau}\pmb{\sigma}\pmb{\tau},
\end{equation}
where $\pmb{\tau}$ is a diagonal matrix containing $+1$ entries for unchanged phase-space coordinates and $-1$ entries at the positions corresponding to the momentum components of the subsystem with respect to which the partial transposition is performed.

In case of a general CV density operator, the symplectic eigenvalues of the partially transposed covariance matrix provide a necessary condition for separability; a separable state must have all symplectic eigenvalues of the partially transposed covariance matrix larger than or equal to one. The covariance matrix of the general CV density operator contains some additional information, i.e., necessary conditions for separability/entanglement can be varified, see e.g. \cite{serafini2024quantum}. 
Taking all this into account, the results \eqref{eq:time_evol_Sigma(t)}, \eqref{eq:spec_symp_mat} and \eqref{eq:sigma_for_polygauss} provide a very useful computational apparatus, as it simplifies calculations of the time-dependent covariance matrix and, consequently, facilitates entanglement detection via the covariance matrix.

In this subsection 
we presented the properties of a polynomial (a quadratic) Gaussian Wigner function. In the next subsection we discuss how one obtains such a system by exciting Gaussian systems.

\subsection{Polynomial excitations.}

They are indispensable tools for manipulating polynomial Gaussian states. We introduce the usual annihilation and creation operators as linear combinations of the canonical position and momentum operators,
\begin{equation}
\label{eq:ladder}
\hat{\mathbf a}
= \frac{1}{\sqrt{2}}
\left(
\hat{\mathbf x}
+ i\hat{\mathbf p}
\right),
\qquad
\hat{\mathbf a}^\dagger
= \frac{1}{\sqrt{2}}
\left(
\hat{\mathbf x}
- i\hat{\mathbf p}
\right).
\end{equation}

The inverse expressions are given by
\begin{equation}
\label{eq:ladder_inv}
\hat{\mathbf x}
= \frac{1}{\sqrt{2}}
\left(\hat{\mathbf a} + \hat{\mathbf a}^\dagger\right),
\qquad
\hat{\mathbf p}
= -\frac{i}{\sqrt{2}}\,
\left(\hat{\mathbf a} - \hat{\mathbf a}^\dagger\right).
\end{equation}

With these definitions, the bosonic commutation relations
\begin{equation}
[\hat a_i,\hat a_j^\dagger] = \delta_{ij},
\qquad
[\hat a_i,\hat a_j] = [\hat a_i^\dagger,\hat a_j^\dagger] = 0
\end{equation}
are satisfied.

In position representation

\begin{equation}
		\hat{a}_j=\frac{1}{\sqrt{2}}\left(x_j+\frac{\partial}{\partial x_j}\right), \qquad
        \hat{a}_j^\dagger=\frac{1}{\sqrt{2}}\left(x_j-\frac{\partial}{\partial x_j}\right).
\end{equation}

Repeatedly applying creation operators to a Gaussian kernel will generate 
polynomials multiplied by the Gaussian kernel (the application is a multiplication
from the left at the quantum state level,
which becomes a sandwich
($\hat{a}_{j_1}^\dagger\cdots\hat{a}_{j_k}^\dagger\hat{\rho}_{G}\hat{a}_{j_k}\cdots\hat{a}_{j_1}$) at the density matrix level). If the original kernel was a
density operator, the generated polynomial Gaussian kernels and all
convex combinations of them (after normalization) are also guaranteed
to be physical with the same Gaussian part. 
All polynomial
Gaussian kernels can be written in the following form:
\begin{equation}
		\hat{K}_{PG}=\sum_{j=1}^{M}\frac{P^\dagger_j(\hat{\mathbf{a}},
		 \hat{\mathbf{a}}^\dagger)\hat{K}_G P_{M-j}(\hat{\mathbf{a}},
	 \hat{\mathbf{a}}^\dagger)}{\Tr\Bigl(\sum_{j=1}^{M}{P^\dagger_j(\hat{\mathbf{a}}, \hat{\mathbf{a}}^\dagger)\hat{K}_G P_{M-j}(\hat{\mathbf{a}}, \hat{\mathbf{a}}^\dagger)\Bigr)}}, \,\,\,\hat{\mathbf{a}}=\left( \hat{a}_1, \ldots, \hat{a}_N \right)^\top.
        \label{eq:op_symbolically_1}
\end{equation}
where $\hat{K}_G$ represents the Gaussian part. As a special case, if $\hat{K}_G=\hat{\rho}_G$ is a density operator, and when  $ P_j = P_{M-j}$  for all $j$, the above expression reduces to the Kraus representation of a density operator, which plays a fundamental role in physics.  The Kraus representation characterizes all physically realizable deterministic quantum operations, i.e., completely positive trace preserving (CPTP) maps or, in case of post-selected (heralded) experiments, completely positive trace decreasing maps. Consequently, it forms one of the cornerstones of quantum information theory \cite{Choi}.

In a simple special case, a quadratic Hamiltonian, diagonal in the normal-mode basis, takes the form
\begin{equation}
\hat H
= \sum_{k=1}^N \hbar\omega_k
\left(
\hat a_k^\dagger \hat a_k + \frac{1}{2}
\right),
\end{equation}
where $\omega_k$ are the normal-mode frequencies.

Let $P(\mathbf{x},\mathbf{p})$ be a polynomial such that
$P(\mathbf{x},\mathbf{p})W^{(G)}(\mathbf{x},\mathbf{p})$ is also a Wigner function. 
Now tracing out subsystems $3,\dots,k$ as in  Eq.~\eqref{eq:red_wign_func} we obtain
\begin{equation}
Q (\mathbf{x}^{(1)},\mathbf{p}^{(1)};
          \mathbf{x}^{(2)},\mathbf{p}^{(2)})
W^{(G)}_2(\mathbf{x}^{(1)},\mathbf{p}^{(1)};
          \mathbf{x}^{(2)},\mathbf{p}^{(2)})
\label{eq:poly_after_tracing3}
\end{equation}
with a new polynomial $Q$.
Note that the Gaussian part remains the same as in \eqref{eq:red_wign_func}, see e.g. \cite{Abramowitz1964}.

\begin{theorem} [Balka--Csord\'as--Homa~\cite{japcsi}]  \label{c:NPT}
If the bipartite Gaussian $\widehat{\rho}_G$ is NPT, then $\widehat{\rho}_{\PG}$ is NPT for any self-adjoint polynomial $P$, too.  
\end{theorem}

Theorem~\ref{c:NPT} implies that if $W^{(G)}_2$ is NPT, then so is \eqref{eq:poly_after_tracing3}.

\begin{claim}\label{c:mainfact1}
For an arbitrary multipartite quantum  system characterized by a Gaussian density operator, any two NPT  entangled subsystems remain NPT  entangled after arbitrary polynomial excitations \eqref{eq:op_symbolically_1}. This physical rule of invariance is independent of dynamics.
\end{claim}

Example~\ref{sub:GKLC_excitations} shows that the converse of Claim~\ref{c:mainfact1}
is not true:

\begin{fact}
There exists a multipartite quantum system characterized by a Gaussian density operator, for which two initially separable subsystems can become NPT  entangled after polynomial excitations \eqref{eq:op_symbolically_1}. 
\end{fact}

Claim~\ref{c:mainfact1} implies the following.

\begin{fact}\label{fact:mainfact2}
In a multipartite quantum  system characterized by a Gaussian density operator, the number of NPT  entangled pairs of subsystems cannot decrease under polynomial excitations \eqref{eq:op_symbolically_1}.
\end{fact}
The total number of bipartite splits in a system of $N$ degrees of freedom (modes), when NPT entanglement can be detected,  including also such splits that can be obtained after tracing out some degrees of freedoms,  is:
\begin{equation*}
T(N) = 
\sum_{k=2}^{N} \binom{N}{k} \left( 2^{k-1} - 1 \right) = \frac{3^N - 2^{N+1} + 1}{2}.
\end{equation*}
As the number of degrees of freedom $N$ increases, the total number of splits $T(N)$ explodes exponentially. Nevertheless the theorems outlined above simultaneously hold for the entire set of subsystems. There is no need for a case-by-case verification of entanglement preservation across these subsystems. Our universal theorems prove that these polynomial excitations robustly preserve the NPT entanglement, maintaining quantum correlations simultaneously across the full complex system.

\section{Examples}
\label{sec:examples}

To complement the next subsections, Appendices \ref{Appendix_time_evolution} and \ref{sec:appendix_A} present the necessary derivations and technical framework for determining the unitary time evolution of the Wigner function and the covariance matrix. Additionally, they include proofs for the conservation of energy and the time-invariance of the symplectic eigenvalues. The comprehensive presentation in these two appendices serves a valuable didactic purpose. 

\subsection{Unitary time evolution of a three-mode mixed Gaussian GHZ like state}
\label{unitary_dyn}
This first example is relevant as it shows how unitary time-evolution can create longer periods of NPT entanglement in a non-trivial setting. According to Claim~\ref{c:mainfact1} and Eq.~\eqref{eq:sol_for_W}, if naturally the initial state is an arbitrary polynomially excited state of an $N$-mode Gaussian density operator, its subsystems will also become NPT entangled during the same time intervals when the unexcited state. In most cases, these time intervals may be extended as a result of excitation, but never shortened.
Firstly, let us observe the time evolution from $t=0$ of a Gaussian state studied by Giedke, Kraus, Lewenstein, and Cirac \cite{Giedke}. This special, three-particle, locally symmetric Gaussian state (called GHZ-like by Giedke et al., but here
referred to as GKLC-state) has the property that in some parameter region, decomposing it into subsystems of $1+2$ particles results in separable subsystems, but the $1+1+1$ decomposition is PPT-entangled. This behaviour may be changed by unitary dynamics causing the state to be NPT entangled at later times.

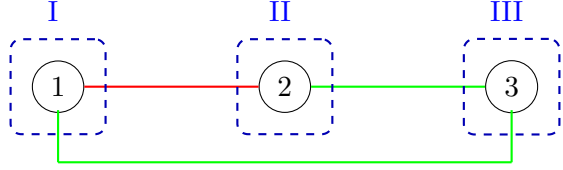
\begin{figure}
\centering
\begin{tikzpicture}[
    node/.style={
        circle,
        fill=black,
        inner sep=2pt
    },
    label/.style={
        font=\footnotesize
    },
    redbox/.style={
        draw=red,
        thick,
        rounded corners
    },
    blueconn/.style={
        draw=red,
        thick,
        dashed
        },
     redconn/.style={
        draw=red,
        thick
    },
greenconn/.style={
        draw=green,
        thick
    },
        bluebox/.style={
        draw=blue!70!black,
        thick,
        dashed,
        rounded corners
    }
]

\node[shape=circle,draw=black] (n1) at (0,0) {1 };
\node[shape=circle,draw=black] (n2) at (3,0) {2 };
\node[shape=circle,draw=black] (n3) at (6,0) {3 };
\node[] (n4) at (0,1) {{\color{blue} \RNum{1} }};
\node[] (n5) at (3,1) {{\color{blue} \RNum{2} }};
\node[] (n6) at (6,1) {{\color{blue} \RNum{3} }};

\draw[greenconn] (n2) -- (n3);


\draw[redconn] (n1) -- (n2);


\draw[greenconn] (0,-0.31)--(0,-1);
\draw[greenconn] (0,-1)--(6,-1);
\draw[greenconn] (6,-1)--(6,-0.31);

\node[bluebox, fit=(n1), inner sep=8pt] {};
\node[bluebox, fit=(n2), inner sep=8pt] {};
\node[bluebox, fit=(n3), inner sep=8pt] {};

\end{tikzpicture}
\caption{A possible scenario where the physical system contains $N=3$ oscillator modes and $3$ subsystems. Subsystem I consists of mode $1$, subsystem II of mode $2$, subsystem III of mode $3$.  In the unexcited state, mode $1$ of the quantum oscillator system is entangled with mode $2$. In the figure these entanglements are represented by solid red lines connecting the corresponding modes. In the unexcited multipartite Gaussian state, subsystems I and II, are NPT-entangled pairwise; subsystems I and III, and II and III are not entangled pairwise, i.e., they are separable. Regardless of excitation in the oscillator system, all pairwise NPT-type entanglements, depicted by solid red lines, are preserved. The separability between subsystems I and III and subsystems II and III  vanish, and these subsystems become NPT  entangled.  Entanglements generated by the excitation are shown as solid green lines. We will not address this scenario in this paper.} 
\label{fig:experiment_scenario}
\end{figure}

Let us introduce the GKLC Gaussian state.
It is a phase-space Gaussian with zero mean $\langle \hat{\mathbf{r}} \rangle=\mathbf{0}$ and with a covariance matrix:
\begin{equation}
\pmb \sigma_{GKLC}':= \begin {pmatrix} a+\alpha&0&c&0&c&0\\ 0
&b+\alpha&0&-c&0&-c\\ c&0&a+\alpha&0&c&0
\\ 0&-c&0&b+\alpha&0&-c\\ c&0&c&0&
a+\alpha&0\\ 0&-c&0&-c&0&b+\alpha\end {pmatrix},
\label{eq:GKLC_covariance_matrix}
\end{equation}
taking the ordering $\mathbf{r}'=(x_1,p_1,x_2,p_2,x_3,p_3)^\top$ and describing a GKLC-state with an added $\alpha \mathbbm{1}_6$ \cite{Giedke}, with the expressions
\begin{equation}
b=\frac14 \left(5a-\sqrt{9a^2-8} \right), \quad c=\frac14 \left(a-\sqrt{9a^2-8} \right).
\end{equation}
The real parameters of the description are $a \geq1$ and $\alpha \geq 0$. As $b=a+c$, a simpler description is available using
\begin{equation}
p=a-c+\alpha, \quad q=a+2c+\alpha.
\end{equation}
In case $\alpha=0$, the product $pq$ is equal to 1. Standard parameters studied in \cite{Giedke} are the following:
\begin{equation}
a=1.2, \quad \alpha=0.3 \quad \Longrightarrow \quad p=p_s= 1.75678, \quad q=q_s = 0.986447.
\label{eq:standard_pars}
\end{equation}
 We note that for $a=1.2$  the state is bound entangled for $0.29756\leq \alpha \leq 0.31355$, so $\alpha=0.3$ is well chosen.

The properly normalized Gaussian Wigner function with zero means and covariance matrix is described by \eqref{eq:GKLC_covariance_matrix}
\begin{equation}
W(\mathbf{r}')=\frac{1}{\pi^3}\frac{1}{\sqrt{\det{\pmb \sigma_{GKLC}'}}}\exp \left(-\mathbf{r}'^{\top} (\pmb \sigma_{GKLC}')^{-1}\mathbf{r}'\right).
\label{eq:}
\end{equation}

The corresponding kernel in position representation is described by reordering the covariance matrix $\mathbf{r}=(x_1,x_2,x_3,p_1,p_2,p_3)^\top$  with the permutation matrix $\mathbf{P}$
\begin{equation}
 \mathbf{P}:= \begin {pmatrix} 1&0&0&0&0&0\\ 0&0&1&0
&0&0\\ 0&0&0&0&1&0\\ 0&1&0&0&0&0
\\ 0&0&0&1&0&0\\ 0&0&0&0&0&1
\end {pmatrix},
 \end{equation}
which acts as $\mathbf{r}=\mathbf{P}\mathbf{r}'$ and $\pmb \sigma_{GKLC}:=\mathbf{P} \pmb \sigma_{GKLC}' \mathbf{P}^{\top}$, with
 \begin{equation}
\pmb \sigma_{GKLC} =
\frac13 \begin{pmatrix} 
2p+q   & -(p-q) & -(p-q) & 0 & 0 & 0 \\
-(p-q) & 2p+q   & -(p-q) & 0 & 0 & 0 \\
-(p-q) & -(p-q) & 2p+q   & 0 & 0 & 0 \\
 0 & 0 & 0 & p+2q & p-q  & p-q \\
 0 & 0 & 0 & p-q  & p+2q & p-q \\
 0 & 0 & 0 & p-q  & p-q  & p+2q 
\end{pmatrix}.
\label{eq:neq_GKLC_cov_mat}
 \end{equation}
The inverse Wigner--Weyl transformation \cite[Eq.~6]{japcsi}  leads to such a Gaussian density matrix where $\mathbf{B}=\mathbf{0}$, $\mathbf{D}=\mathbf{E}=\mathbf{0}$ and the matrix $\mathbf{A}$ and $\mathbf{C}$ can be read from the blocks of the covariance matrix $\pmb \sigma_{GKLC}$:
\begin{equation}
\pmb \sigma_{GKLC}=\begin{pmatrix} \displaystyle\frac14 \mathbf{C}^{-1} & \mathbf{0} \\ \mathbf{0} & 4\mathbf{A} \end{pmatrix}.
\label{eq:GKLC_cov_mat}
\end{equation}
Both $\mathbf{A}$ and $\mathbf{C}$ are locally symmetric, which means that diagonal elements are the same and off-diagonal elements are the same with diagonal  $A_d$, $C_d$ and off-diagonal $A_o$, $C_o$ elements as
\begin{equation}
A_d=\frac{1}{12}(p+2q), \quad C_d=\frac{1}{12pq}(p+2q),\quad A_o= \frac{1}{12}(p-q), \quad C_o=\frac{1}{12pq}(p-q).
\label{eq:GKLC_A_and_C}
\end{equation}
The Gaussian density matrix of the GKLC-state with unit trace is of the form
\begin{equation}
\rho_G^{(GKLC)}(\mathbf{x},\mathbf{y})=\frac{1}{\pi^{3/2}\sqrt{p^2 q}}\exp \Bigl(-(\mathbf{x}-\mathbf{y})^\top \mathbf{A}(\mathbf{x}-\mathbf{y})-(\mathbf{x}+\mathbf{y})^\top \mathbf{C}(\mathbf{x}+\mathbf{y}) \Bigr).
\label{eq:GKLC_G_density_matrix}
\end{equation}

To prescribe the dynamics let us choose a locally symmetric $\mathbf{H}$ matrix in \eqref{eq:quadratic_hamiltonian_matrix}, which governs the dynamics in \eqref{eq:dynamics_of_rho_op} of the density operator as
\begin{equation}
\mathbf H := \begin{pmatrix}
 3&0&0&0&1&1\\ \noalign{\medskip}0&3&0&1
&0&1\\ \noalign{\medskip}0&0&3&1&1&0\\ \noalign{\medskip}0&1&1&2&0&0
\\ \noalign{\medskip}1&0&1&0&2&0\\ \noalign{\medskip}1&1&0&0&0&2
\end{pmatrix}.
 \end{equation}
This choice gives a relatively simple form for the symplectic matrices in \eqref{eq:spec_symp_mat}, all the four main blocks of them are locally symmetric (the diagonal elements are the same and the off diagonals have the same values, too), and furthermore, all the elements are linear combinations of $\cos$ and $\sin$ of $\sqrt{5}t$ and $\sqrt{2}t$. After partially reflecting the covariance matrix ($p_3 \to -p_3$) the symplectic eigenvalues (i.e., the positive eigenvalues of $i \pmb{\Omega}\pmb{\tilde{\sigma}} $) are calculated numerically. These symplectic eigenvalues can be seen in Fig.~\ref{fig:gaussian_entanglement_later}. If one of them goes below the horizontal line equal to one, then the system is NPT  entangled. This happens  close to $t=0$, it remains for a while NPT  entangled, however, its behavior changes back to separable around $t=4.5$ for a short time and then goes back to NPT  entangled states and later times similar changements can be seen back and forth.

 \begin{equation}
\rho(\mathbf x,\mathbf y,0):=\rho_G^{(GKLC)}(\mathbf x,\mathbf y),
\end{equation}
where


 \begin{figure}
\includegraphics[angle=270, width=0.9\textwidth]{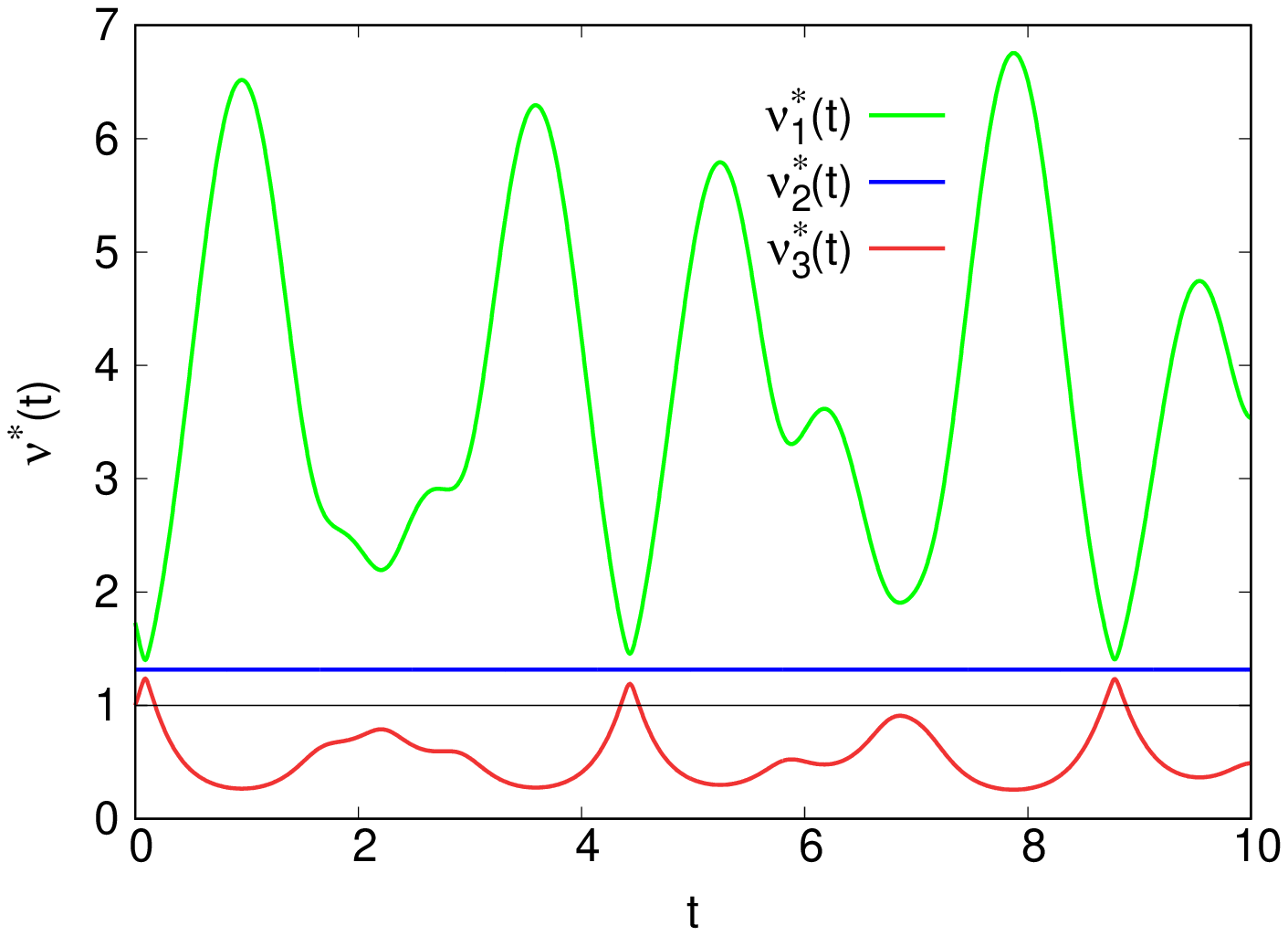}
\caption{Symplectic eigenvalues $\nu_i^*$, $i=1,2,3$ after partial transpose at time $t$ as a function of $t$. Thin black line indicates the value 1.}
\label{fig:gaussian_entanglement_later}
\end{figure}

\subsection{Excitation of a three-mode mixed Gaussian GKLC-state}
\label{sub:GKLC_excitations}
In this example, we consider polynomial excitations of a three-mode mixed Gaussian GKLC-state. The position representation is particularly useful here, since our theoretical results and entanglement criteria are formulated in this representation. At the same time, the Fock representation is of direct experimental relevance, and provides a natural description in terms of creation and annihilation operators. It is shown how the polynomially excited state can be described in both representations and how to relate the corresponding parameters. For the example considered below, the excitation leads to a non-Gaussian state with NPT  entanglement. Let us study the effect of multiplying the previously investigated GKLC-state by some polynomials. 
We will see that, similarly to the dynamics followed in the previous section, the polynomial excitations can also bring the system to an NPT entangled non-Gaussian state. 

To understand better what is going on, let us enumerate a few facts. The covariance matrix \eqref{eq:neq_GKLC_cov_mat} can be diagonalized symplectically \cite{serafini2024quantum}, i.e., 
\begin{equation}
\pmb \sigma_{GKLC} =\mathbf{S}_H \mathrm{diag}\,\left( \mathbf{D}_0, \mathbf{D}_0\right) \mathbf{S}_H^\top,
\end{equation}
where
\[
\mathbf{S}_H=\begin{pmatrix} 
\frac{1}{3w} +\frac{2w}{3} & \frac{1}{3w}-\frac{w}{3} & \frac{1}{3w}-\frac{w}{3} & 0 & 0 & 0 \\
\noalign{\medskip}
\frac{1}{3w}-\frac{w}{3} & \frac{1}{3w} +\frac{2w}{3} & \frac{1}{3w}-\frac{w}{3} & 0 & 0 & 0 \\
\noalign{\medskip}
\frac{1}{3w}-\frac{w}{3} & \frac{1}{3w}-\frac{w}{3} & \frac{1}{3w} +\frac{2w}{3} & 0 & 0 & 0 \\
\noalign{\medskip} 0 & 0 & 0 & \frac{2}{3w} +\frac{w}{3} & -\frac{1}{3w}+\frac{w}{3} & - \frac{1}{3w}+\frac{w}{3} \\
\noalign{\medskip} 0 & 0 & 0 &  -\frac{1}{3w}+\frac{w}{3} & \frac{2}{3w} +\frac{w}{3} & -\frac{1}{3w}+\frac{w}{3} \\
\noalign{\medskip} 0 & 0 & 0 & -\frac{1}{3w}+\frac{w}{3} & -\frac{1}{3w}+\frac{w}{3} & \frac{2}{3w} +\frac{w}{3}
\end{pmatrix}
\]
with
\[
w=\left(\frac{p}{q} \right)^{1/4},
\]
and $\mathbf{D}_0$ is a diagonal matrix which contains the three positive symplectic eigenvalues $\nu_1,\nu_2,\nu_3$ as
\[
\mathbf{D}_0=\mathrm{diag}\, (\nu_1,\nu_2,\nu_3)=\mathrm{diag}\, \left( \sqrt{pq},\sqrt{pq},\sqrt{pq}\right).
\]
In the example we study, all three symplectic eigenvalues are the same and, furthermore, in the special case of $\alpha=0$, they are all equal to one, which means that the Gaussian state is pure.  

The matrix $\mathbf{S}_H$ is symplectic, which means that
\[
\mathbf{S}_H \mathbf \Omega \mathbf{S}_H^\top = \mathbf \Omega,
\]
where $\mathbf{\Omega}$ is the standard symplectic matrix
 \begin{equation}
\mathbf \Omega := \begin{pmatrix}
 0&0&0&1&0&0\\ 0&0&0&0
&1&0\\ 0&0&0&0&0&1\\ -1&0&0&0&0&0
\\ 0&-1&0&0&0&0\\ 0&0&-1&0&0&0
\end{pmatrix}.
 \end{equation}
The symplectic matrix $\mathbf{S}_H$ can be written as
\[
\mathbf{S}_H=\exp \left( \pmb\Omega \mathbf{H}_d\right),
\]
where $\mathbf{H}_d$ is a real, symmetric matrix explicitly given as
\begin{equation}
\mathbf{H}_d=\frac{\log{w}}{3} 
\begin{pmatrix}
 0 & 0 & 0 & 1 & -2 & -2  \\
 0 & 0 & 0 & -2 & 1  & -2  \\
 0 & 0 & 0 & -2 & -2 & 1  \\
 1 & -2 & -2 & 0 & 0 & 0 \\
 -2 & 1 & -2  & 0 & 0 & 0 \\
  -2 & -2 & 1 & 0 & 0 & 0 \\
\end{pmatrix}
\label{eq:hmat}
\end{equation}
Let us define  
\[
\lambda_i=\frac{\nu_i -1}{\nu_i+1}, \quad i=1,2,3.
\]
In our example all the symplectic eigenvalues $\nu_i$ are the same, and bigger than one. 
Consequently, all the $\lambda_i$ are the same and positive
\[
\lambda=\lambda_i=\frac{\sqrt{pq}-1}{\sqrt{pq}+1}, \quad i=1,2,3.
\]
The Gaussian operator for which the covariance matrix is $\pmb \sigma_{GKLC}$, is a density operator.
This density operator can be expressed as \cite{serafini2024quantum}
\begin{equation}
\hat{\rho}_{GKLC}=\hat{S}^\dagger_H \!\!\left(\sum_{m_1,m_2,m_3=0}^\infty |m_1,m_2,m_3 \rangle (1-\lambda)^3 \lambda^{m_1+m_2+m_3}\langle m_1,m_2,m_3 | \right)\!\!\hat{S}_H.
\label{eq:gauss_operator}
\end{equation}

Here the unitary operator $\hat{S}_H$ is defined as
\begin{equation}
\hat{S}_H = \exp \left( \frac{i}{2} \hat{\mathbf{r}}^\top \mathbf{H}_d\hat{\mathbf{r}}\right), \quad \hat{S}_H \hat{\mathbf{r}}\hat{S}^{\dagger}_H =\mathbf{S}_H \hat{\mathbf{r}},
\label{eq:Shat_operator}
\end{equation}
where $\mathbf{H}_d$ is given in \eqref{eq:hmat}. The states appearing in \eqref{eq:gauss_operator} $|m_1,m_2,m_3 \rangle$ are the direct product of the standard harmonic oscillator states with the properties:
\begin{eqnarray}
|m_1,m_2,m_3 \rangle & = & |m_1 \rangle_1 \otimes  |m_2 \rangle _2 \otimes  |m_3 \rangle_3, \\
\hat{a}_j |m \rangle_j &=& \sqrt{m} | m-1 \rangle_j, \quad j=1,2,3, \\
\hat{a}_j^\dagger |m \rangle_j &=& \sqrt{m+1} | m+1 \rangle_j, \quad j=1,2,3,
\end{eqnarray}
It can be checked that the Gaussian density matrix of the GKLC-state in \eqref{eq:GKLC_G_density_matrix} indeed has the form
\[
\rho_G^{(GKLC)}(\mathbf{x},\mathbf{y})=\langle \mathbf{x} | \hat{\rho}_{GKLC} | \mathbf{y}\rangle,
\]
with $\hat{\rho}_{GKLC}$  given in \eqref{eq:gauss_operator}.

We call expression

\begin{equation}
\rho_{PG}(\mathbf x,\mathbf y):=P(\mathbf x,\mathbf y) \rho_G(\mathbf x,\mathbf y),
\end{equation}
polynomial self-adjoint Gaussian kernel, if 
\[
P(\mathbf{x},\mathbf{y})=P(\mathbf{y},\mathbf{x})^*.
\]
Furthermore, if the operator $\hat{\rho}_{PG}$ is positive-semidefinite and properly normalized to  $\mathrm{Tr}\, \hat{\rho}_{PG}=1$ we call it a polynomial Gaussian density operator.

Next, we turn to polynomial Gaussian, where the Gaussian is chosen to be a GKLC Gaussian
and the polynomial is a general, self-adjoint quadratic one:
\begin{eqnarray}
\rho'_{PG}(\mathbf{x},\mathbf{y})&=& \rho_G^{(GKLC)}(\mathbf{x},\mathbf{y})\Bigg[ (\mathbf{x}-\mathbf{y})^\top \mathbf{A}^{(p)}(\mathbf{x}-\mathbf{y})+i(\mathbf{x}-\mathbf{y})^\top \mathbf{B}^{(p)}(\mathbf{x}+\mathbf{y}) \Biggr. \nonumber \\
&&\Biggl. +(\mathbf{x}+\mathbf{y})^\top \mathbf{C}^{(p)}(\mathbf{x}+\mathbf{y})+i\mathbf{D}^{(p)\top} (\mathbf{x}-\mathbf{y})+\mathbf{E}^{(p)\top} (\mathbf{x}+\mathbf{y}) +\mathrm{F}^{(p)}\Biggr].
\label{eq:quadratic_ploynomial_gaussian}
\end{eqnarray}
Here $\mathbf{A}^{(p)}$ and $\mathbf{C}^{(p)}$  are  real symmetric matrices, $\mathbf{B}^{(p)}$ is a real matrix, $\mathbf{E}^{(p)}$ and $\mathbf{D}^{(p)}$ are  real vectors, and $F^{(p)}$  is a real constant. In Ref. \cite{Newton,japcsi} the positivity of the kernel is studied and useful criterions are given to ensure that $\hat{\rho}'_{PG}$ is a density operator. The coefficients  $\mathbf{A}^{(p)}$, $\mathbf{B}^{(p)}$, $\mathbf{C}^{(p)}$, $\mathbf{D}^{(p)}$, $\mathbf{E}^{(p)}$ and $\mathrm{F}^{(p)}$ are called polynomial parameters.

A possible situation for a positive semi-definite kernel is if the polynomial is chosen as $P(\mathbf x,\mathbf y)=p(\mathbf{x})^* p(\mathbf{y})$ (See Ref.~\cite{Newton}).
In our numerics we present results for the polynomial
\begin{eqnarray}
P(\mathbf x,\mathbf y)&:=&\frac{6}{(2p+q)(1+\kappa^2)}\left(x_2+i\kappa x_3\right)\left(y_2-i\kappa y_3\right), \nonumber \\
\rho_{GP}(\mathbf{x},\mathbf{y})&=&P(\mathbf x,\mathbf y) \rho_{G}^{(GKLC)}(\mathbf{x},\mathbf{y})
\label{eq:spec_pol}
\end{eqnarray}
and choosing the Gaussian to $\hat{\rho}_G^{(GKLC)}$.
The prefactor ensures that the trace of $\hat{\rho}_{GP}$ is unity. We show two different choices for Mercer weights $c_k$ and positions $(x_k,y_k,z_k)$ in Fig.~\ref{fig:mercer_wo_reflections}. Both sub-figures show only positive Mercer quantity $M(\kappa)$ (See Eq.~\eqref{eq:Mercer}), and correspondingly  they support that $\rho_{GP}(\mathbf{x},\mathbf{y})$ defined in Eq.~\eqref{eq:spec_pol} is a density matrix. 

\begin{figure}
\includegraphics[angle=270, width=0.45\textwidth]{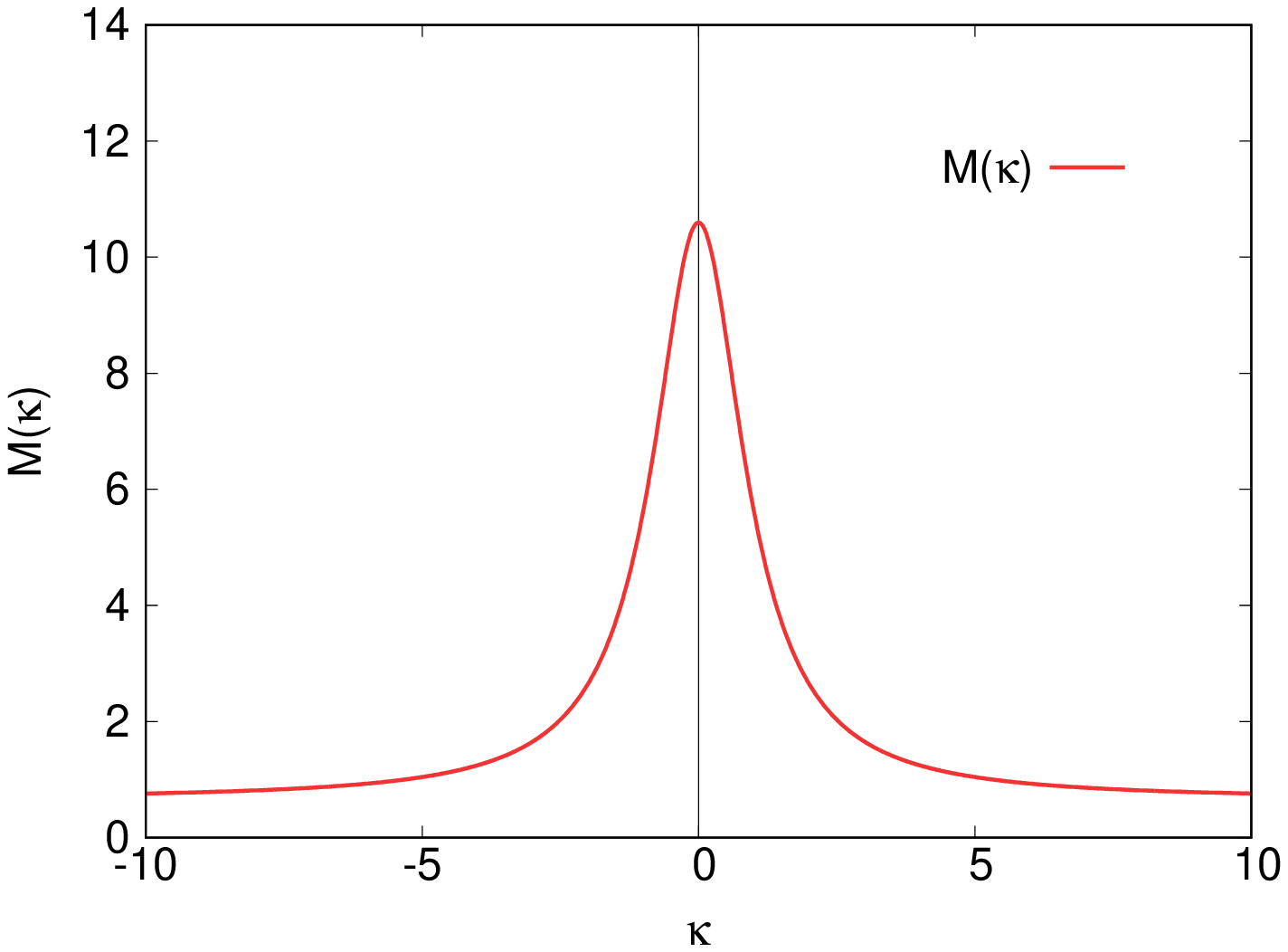} \quad
\includegraphics[angle=270, width=0.45\textwidth]{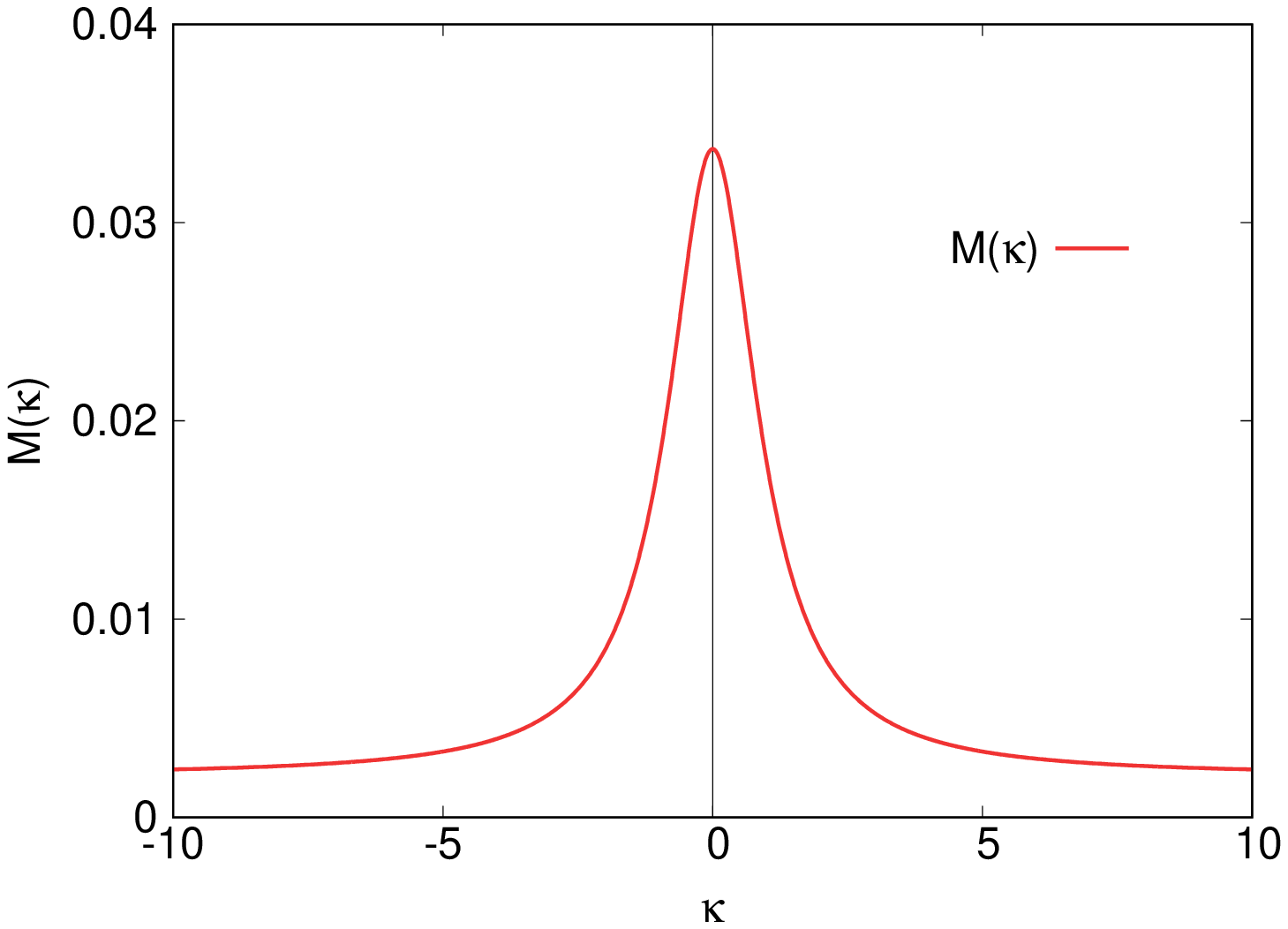}
\caption{The Mercer's  quantity $M$ \eqref{eq:Mercer} as a function of the polynomial parameter $\kappa$ for the polynomial Gaussian density matrix \eqref{eq:spec_pol}. Mercer weights $c_k$ and positions $(x_k,y_k,z_k)$ are:  $c_k= (k-i)^k$, $(x_k,y_k,z_k)=(1/k^k)\cdot(0.1,0.8,-0.2)$  (left panel), and $c_k= (k+i)^k$, $(x_k,y_k,z_k)=(-1/k)^k\cdot(0.1,0.8,-0.2)$ (right panel). $i$ is the imaginary unit and $k=1,\ldots,10$, GKLC parameters are taken as in Eq.\eqref{eq:standard_pars}.}
\label{fig:mercer_wo_reflections}
\end{figure}

We found that the density matrix \eqref{eq:spec_pol} is NPT  entangled for $\kappa \neq 0$ by applying the Peres-Horodecki criterion \eqref{t:PH}. We note that the reflection in $p_3 \to -p_3$ in the Wigner function requires the replacements $(x_3,y_3) \to (y_3,x_3)$ in the density matrix arguments. By choosing the same Mercer parameters as in Fig.~\ref{fig:mercer_wo_reflections} after partial reflection in the third coordinates the Mercer parameter $M(\kappa)$ is negative either for positive or negative $\kappa$ values as depicted in Fig.~\ref{fig:measuring_with_mercer_3}.

\begin{figure}
\includegraphics[angle=270, width=0.9\textwidth]{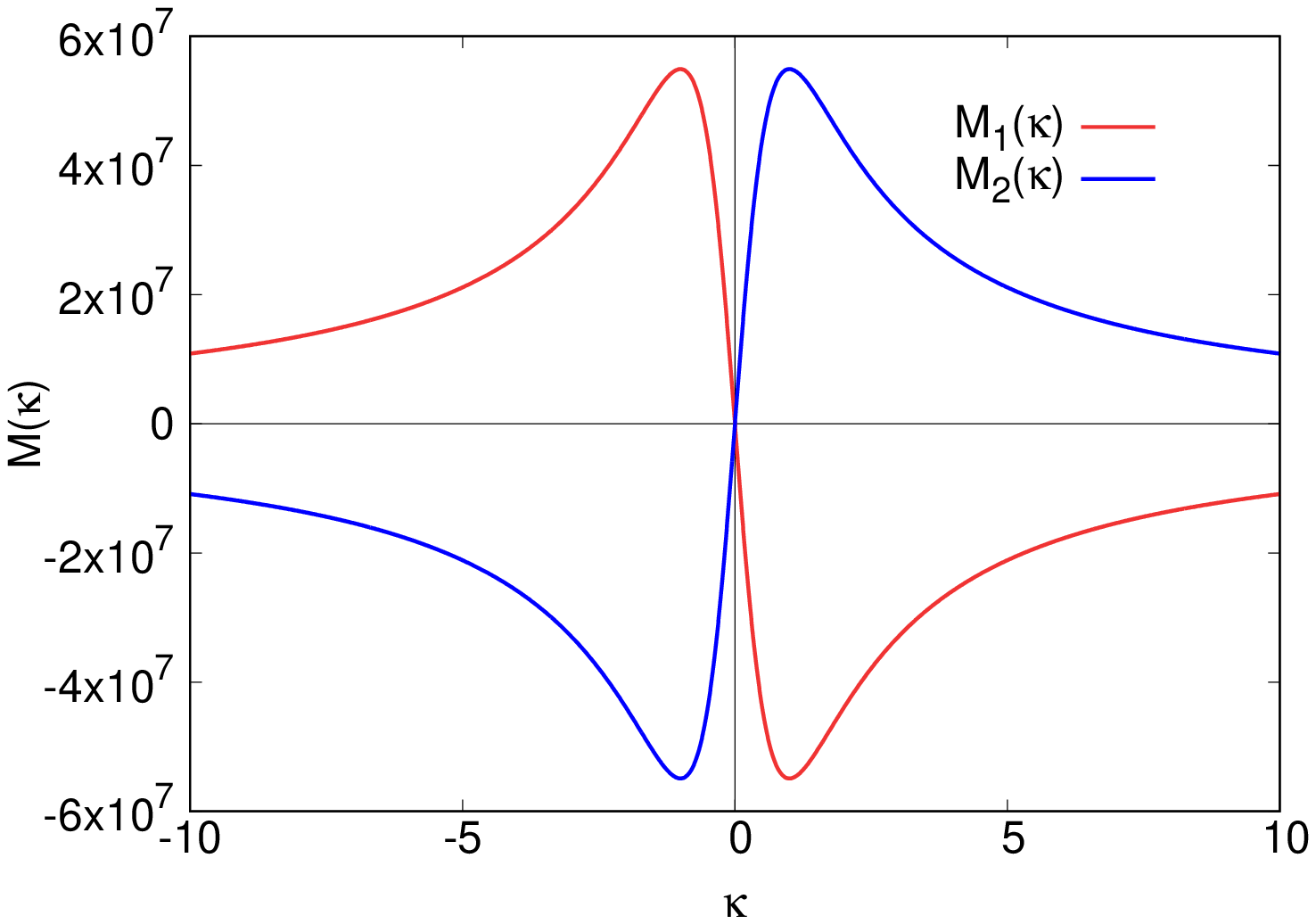}
\caption{The Mercer's quantity $M$ after partial reflection in the third coordinates as a function of the polynomial parameter $\kappa$  for the polynomial Gaussian density matrix \eqref{eq:spec_pol}.
Mercer weights and positions are: $c_k= (k-i)^k$, $(x_k,y_k,z_k)=(1/k^k)(0.1,0.8,-0.2)$  (red line) and  $c_k= (k+i)^k$, $(x_k,y_k,z_k)=(-1/k)^k (0.1,0.8,-0.2)$  (blue line). $i$ is the imaginary unit and $k=1,\ldots,10$, GKLC parameters are taken as in Eq.~\eqref{eq:standard_pars}.}
\label{fig:measuring_with_mercer_3}
\end{figure}

After a straightforward calculation one can find the corresponding operator $\hat{\rho}_{PG}$ to any polynomial Gaussian density matrix $\rho_{PG}(\mathbf{x},\mathbf{y})$:
\[
\rho_{PG}(\mathbf{x},\mathbf{y})=\langle \mathbf{x}|\hat{\rho}_{PG}| \mathbf{y}\rangle
\]
In our examples, we use two types of creation and annihilation operators. One pair of them ($\hat{\tilde{\mathbf{a}}}^\dagger$,$\hat{\tilde{\mathbf{a}}}$) are defined by
\begin{equation}
\label{eq:tilde_creation_annihilation_def}  
\hat{\tilde{a}}_j=
          \hat{\mathbf{S}}_H^{\dagger} \hat{a}_j \hat{\mathbf{S}}_H  
          , \qquad \hat{\tilde{a}}^{\dagger}_j= 
          \hat{\mathbf{S}}_H^\dagger \hat{a}^\dagger_j \hat{\mathbf{S}} _H,
\end{equation}       
in connection with the usual ones ($\hat{\mathbf{a}}^\dagger$,$\hat{\mathbf{a}}$):
\label{eq:gen_creation_annihilation}
\begin{equation}
		\hat{a}_j=\frac{1}{\sqrt{2}}\left(\hat{x}_j+i  \hat{p}_j\right), \qquad
        \hat{a}_j^\dagger=\frac{1}{\sqrt{2}}\left(\hat{x}_j-i  \hat{p}_j\right).
\end{equation}

The polynomial Gaussian kernel form \eqref{eq:op_symbolically_1} of $\hat{\rho}_{PG}$
can be rewritten as
\begin{equation}
		\hat{\rho}_{PG}=\sum_{j=1}^{K}\frac{\tilde{P}^\dagger_j(\hat{\tilde{\mathbf{a}}}, \hat{\tilde{\mathbf{a}}}^\dagger)\hat{\rho}_G \tilde{P}_{K-j}(\hat{\tilde{\mathbf{a}}}, \hat{\tilde{\mathbf{a}}}^\dagger)}{\Tr\Bigl(\sum_{j=1}^{K}{\tilde{P}^\dagger_j(\hat{\tilde{\mathbf{a}}}, \hat{\tilde{\mathbf{a}}}^\dagger)\hat{\rho}_G \tilde{P}_{K-j}(\hat{\tilde{\mathbf{a}}}, \hat{\tilde{\mathbf{a}}}^\dagger)\Bigr)}}, \,\,\,\hat{\tilde{\mathbf{a}}}=\left( \hat{\tilde{a}}_1, \ldots, \hat{\tilde{a}}_N \right)^\top,
        \label{eq:op_symbolically}
\end{equation}
with appropriate new polynomials $\tilde{P}_j(\hat{\tilde{\mathbf{a}}}, \hat{\tilde{\mathbf{a}}}^\dagger)$.

For quadratic polynomial Gaussian \eqref{eq:quadratic_ploynomial_gaussian} the symbolic form \eqref{eq:op_symbolically} can be specified to: 
\begin{eqnarray}
\hat{\rho}'&=&c^{(0)} \hat{\rho}_G+\sum_{j=1}^N \left(\mathbf{c}^{(1)}_{j} \hat{\tilde{a}}^{\dagger}_j\hat{\rho}_G+\mathbf{c}^{(1)*}_{j} \hat{\rho}_G \hat{\tilde{a}}_j  \right) \nonumber \\
&& + \sum_{j,k=1}^N  \left(\mathbf{c}^{(2)}_{j,k} \hat{\tilde{a}}^\dagger_j \hat{\tilde{a}}^\dagger_k\hat{\rho}_G + \mathbf{c}^{(2)*}_{j,k}  \hat{\rho}_G\ \hat{\tilde{a}}_k \hat{\tilde{a}}_j \right)+  \sum_{k,j=1}^N   \mathbf{c}^{(3)}_{j,k} \hat{\tilde{a}}^\dagger_j \hat{\rho}_G\ \hat{\tilde{a}}_k ,
\label{eq:example_Fock_repr}
\end{eqnarray}
where $\mathbf{c}^{(1)}$ is a complex-valued vector,  $\mathbf{c}^{(2)}$ is a symmetric complex-valued matrix, $\mathbf{c}^{(3)}$ is a Hermitian matrix, $c_0$ is a real constant.

Operator parameters  $c_0$, $\mathbf{c}^{(1)}$, $\mathbf{c}^{(2)}$, $\mathbf{c}^{(3)}$ are in one to one correspondance with polynomial parameters $\mathbf{A}^{(p)}$, $\mathbf{B}^{(p)}$, $\mathbf{C}^{(p)}$, $\mathbf{D}^{(p)}$, $\mathbf{E}^{(p)}$ and $\mathrm{F}^{(p)}$.
For GKLC Gaussian we enumerate these relations in Appendix  \ref{sec:appandix_A}.

An alternative way to study entanglement is the direct diagonalization of the operator connected to the partial transpose of the polynomial Gaussian \eqref{eq:spec_pol}. 
Matrix elements can be directly written in terms of the operator parameters, occuring in 
\eqref{eq:example_Fock_repr}.

Recall that the \emph{negativity} of a partially transposed density  operator $\widehat{\tilde{\rho}}$ is the sum of positive eigenvalues of $-\widehat{\tilde{\rho}}$. By this method one can directly obtain the negativity of $\widehat{\tilde{\rho}}$, see Fig.~\ref{fig:negativity}.
\begin{figure}
\includegraphics[angle=270, width=0.9\textwidth]{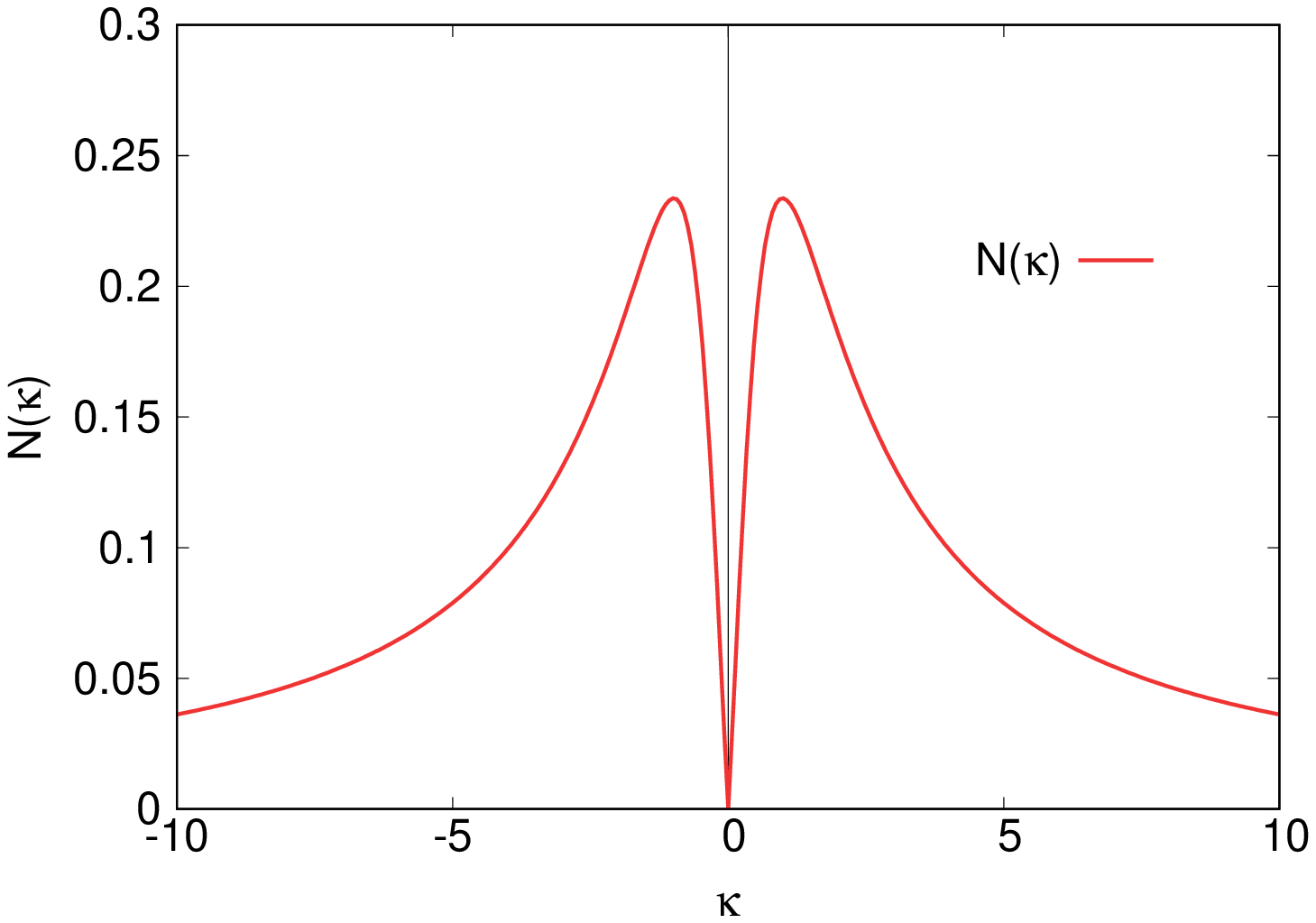}
\caption{Negativity $N$ as a function of the polynomial parameter $\kappa$  for the polynomial Gaussian density matrix \eqref{eq:spec_pol}. Partial reflection was in the third coordinates. GKLC parameters are taken as in Eq.~\eqref{eq:standard_pars}.}
\label{fig:negativity}
\end{figure}

\subsection{Single photon adding transformation}

In this section, the effect of adding a single photon to one mode of a multimode state is studied. It is show that the usual photon-adding transformation, defined by the local creation and annihilation operators, preserves both separability and PPT. Thus, a local single-photon excitation cannot turn a separable state into an entangled PPT state or an NPT state, nor can it turn a PPT state into an NPT state. However, this conclusion does not generally hold for the non-local operators $\hat{\tilde a}_j$, which can act on several modes simultaneously.

The single photon adding transformation $\hat{\rho} \rightarrow \hat{\rho}'$ is abundant in physics,. 
In Kraus representation \cite{Choi} it can be written in the following way:

\begin{equation}
\hat{\rho}'=\frac{\hat{a}_j^\dagger \hat{\rho} \hat{a}_j }{\Tr \left(\hat{a}_j^\dagger \hat{\rho} \hat{a}_j\right) }, 
\label{eq:single_photon_exc}
\end{equation}
which is equivalent to 
\begin{equation}
\rho'(\mathbf{x},\mathbf{y})=\frac{N_p}{2}\left(x_j- \frac{\partial}{\partial x_j}\right) \left(y_j- \frac{\partial}{\partial y_j}\right) \rho(\mathbf{x},\mathbf{y}),
\end{equation}
in position representation, where
\begin{equation}
N_p=\frac{1}{\Tr \left(\hat{a}_j^\dagger \hat{\rho} \hat{a}_j\right) },
\end{equation}
ensures normalization.
If the density operator $\hat{\rho}$ is PPT, then after the partial transpose $x_j \leftrightarrow y_j$ the new operator $\hat{\widetilde{\rho}}$ is also a density operator and the partial transpose of the single photon excited state in position representation is

\begin{equation}
\widetilde{\rho'}(\mathbf{x},\mathbf{y})=\frac{N_p}{2}\left(y_j- \frac{\partial}{\partial y_j}\right) \left(x_j- \frac{\partial}{\partial x_j}\right) \widetilde{\rho}(\mathbf{x},\mathbf{y}).
\end{equation}
 We can see  that  $\hat{\widetilde{\rho}'}$ is also a density operator, therefore this transformation preserves the property of PPT.

If the density operator $\hat{\rho}$ is separable according to a given partition,  then it can be obtained as a convex combination of tensor product density operators, in position representation as
\begin{equation}
\rho(\mathbf{x},\mathbf{y})= \sum_{k=1}^K p_k \prod_{l=1}^L \rho^{(k)}(\mathbf{x}_l,\mathbf{y}_l)  \nonumber
\end{equation}
where $p_k\geq 0$ for all $k \in \mathbb{N}$ and $\sum_k p_k=1$. 
After a single photon excitation it turns into the following in position representation:
\begin{equation}
\rho'(\mathbf{x},\mathbf{y})=\frac{N_p}{2}\left(x_j-\frac{\partial}{\partial x_j}\right) \left(y_j- \frac{\partial}{\partial y_j}\right) \rho(\mathbf{x},\mathbf{y}). \nonumber
\end{equation}

Let us suppose that coordinates $x_j,y_j$ occur for the index $l=m$.
After adding a single photon the excited state is 
\begin{align}
\rho'(\mathbf{x},\mathbf{y})= \frac{N_p}{2}\sum_{k=1}^K p_k  \left(x_j- \frac{\partial}{\partial x_j}\right) \left(y_j- \frac{\partial}{\partial y_j}\right) \prod_{l=1}^L \rho^{(k)}(\mathbf{x}_l,\mathbf{y}_l) \nonumber \\
=\frac{N_p}{2}\sum_{k=1}^K p_k  \left(x_j- \frac{\partial}{\partial x_j}\right) \left(y_j-  \frac{\partial}{\partial y_j}\right) \rho^{(k)}(\mathbf{x}_m,\mathbf{y}_m) \prod_{\substack{l=1 \\ l \neq m}}^L \rho^{(k)}(\mathbf{x}_l,\mathbf{y}_l) \nonumber \\
= \sum_{k=1}^K p_k  \rho'^{(k)}(\mathbf{x}_m,\mathbf{y}_m) \prod_{\substack{l=1 \\ l \neq m}}^L \rho^{(k)}(\mathbf{x}_l,\mathbf{y}_l),
\end{align}
 and  we exploited the linear property of operators acting on the density operator and properties of Kraus representation \cite{Choi}. Consequently, this single photon excitation preserves the property of separability according to a given partition.

It is very important to note that the above statements are generally not valid for the creation and annihilation operators defined in  \eqref{eq:gen_creation_annihilation}. 
In the state
\begin{equation}
\hat{\rho}''=\frac{\hat{\tilde{a}}_j^\dagger \hat{\rho} \hat{\tilde{a}}_j }{\Tr \left(\hat{\tilde{a}}_j^\dagger \hat{\rho} \hat{\tilde{a}}_j\right) }.
\end{equation}
the operators $\hat{\tilde{a}}_j$ will affect several modes simultaneously,
and
consequently, the $\hat{\rho} \rightarrow \hat{\rho}''$ operation no longer necessarily preserves separability and the PPT property.

We have thus proven that the \eqref{eq:single_photon_exc} single-photon entanglement transformation preserves the PPT property and the  separability of an arbitrary density operator. 
In other words, applying local excitations, $ \text{SEP} \rightarrow \text{PPTE} \rightarrow \text{NPT} $ and $\text{PPT} \rightarrow \text{NPT}$ transitions are impossible. Here, PPTE stands for a positive-partial-transpose entangled state, and SEP stands for a separable state.

\section{ Summary, discussion and conclusions}

In this work, we have investigated several physical, and entanglement properties of continuous-variable (CV) multipartite quantum systems.  We provide a unified perspective linking mathematics, theoretical physics, and experimental physics through the preparation of an arbitrary density operator by polynomial excitations and their unitary time evolution.

In Theorem~\ref{c:homa_tetel} we have also considered the experimentally highly relevant case where a pure state density matrix \eqref{eq:pure_density_matrix} is multiplied by a self-adjoint polynomial of the positions \eqref{eq:pure_dens_mat_with_polynom}. For this case, we have shown that if two subsystems of a multipartite  pure quantum system are initially separable, then after multiplication with polynomials, the detection of NPT  entanglement reduces to determining whether the so-called bilinear representation matrix is positive semi-definite. 
Finally, in Example~\ref{ex:wrong_with_mixed} we presented an exact analytical example showing that a polynomial with a non-positive-semidefinite bilinear representation matrix can still produce a positive-semidefinite kernel when multiplied by a kernel of a mixed Gaussian density operator.

Independently of any dynamics, in Claim \ref{c:mainfact1} we have identified a physical invariance rule during polynomial excitations stating that, if two Gaussian subsystems are initially NPT  entangled, then polynomial excitations preserve their NPT  entanglement between these subsystems. In Subsection~\ref{sub:GKLC_excitations} we have demonstrated through a non-trivial example that polynomial excitations can generate NPT  entanglement between two subsystems that are initially separable. As a consequence, we have shown in Fact~\ref{fact:mainfact2} that the number of NPT  entangled pairs of subsystems cannot decrease under polynomial excitations.

Starting from the Fock-representation of Gaussian kernels we have studied excitations of it by applying creation operators from the left and annihilation operators from the right of the Gaussian operator. We have shown that these types of excitations will lead to polynomial Gaussian forms both in position and phase space representations.  Explicit transformation rules are worked out in the Appendix~\ref{sec:appandix_A} for the quadratic excitations of the GKLC Gaussian state. 
Furthermore, we have derived a compact, dynamics-independent expression \eqref{eq:sigma_for_polygauss} for the covariance matrix of quadratic-polynomial Gaussian operators in terms of the underlying Gaussian covariance matrix and the polynomial coefficients.

Our analysis confirms that local single-photon addition preserves both the PPT property and separability of arbitrary multimode quantum states. This shows that local and generalized nonlocal creation operators affect entanglement in fundamentally different ways.

Within a modern formalism, we have solved the von Neumann equation for the time evolution of a system starting from   arbitrary initial conditions under the influence of a time-independent quadratic Hamiltonian. The solution \eqref{eq:sol_for_W} is expressed in a form that is particularly convenient for practical applications. Through a non-trivial Gaussian example  (See Fig.~\ref{fig:gaussian_entanglement_later}), we have demonstrated that unitary time evolution of the GKLC-state one can generate periods of long-lived entanglement and short-lived disentanglement for initially separable subsystems. 

Within this formalism, we have also determined the time evolution of the expectation values \eqref{eq:time_evol_d(t)} of the position and momentum operators, 
the time-dependence of the covariance matrix \eqref{eq:time_evol_Sigma(t)}.
We showed that the expectation value of the energy is constant in time, as it should be, and we found that the symplectic eigenvalues of the time-dependent covariance matrix are also time-independent.

We have connected these results with possible experimental setups and technological applications.

\subsection*{Acknowledgments} We are indebted to Aur\'el G\'abris, \'Akos Gombk\"ot\H{o}, G\'abor Helesfai, Zolt\'an Kaufmann,
 M\'aty\'as Koniorczyk and G\'eza T\'oth for some illuminating conversations. 

\subsection*{Funding} RB was supported by the National Research, Development and Innovation Office -- NKFIH, grant no.~146922, and by the J\'anos Bolyai Research Scholarship of the Hungarian Academy of Sciences.

GH and AB were supported by the 2021 Thematic Excellence Programme of the NKFIH (Project No.~TKP2021-NVA-04).

\appendix
\section{Photon addition and subtraction}
\label{Appendix_photon_add&subtract}
Recently much effort has gone into developing a scheme that exactly performs photon addition and subtraction \cite{Luo2025,Lund2024}. Often the state created in an experiment can be expressed as a finite number
of creation and annihilation operators applied to a Gaussian state.
In these cases - although the system is in an infinite-dimensional Hilbert space -
the tools developed in this article may help determine the presence of quantum entanglement
between modes. The results in section \ref{sub:GKLC_excitations}
(cf. eq \eqref{eq:quadratic_ploynomial_gaussian} -\eqref{eq:example_Fock_repr})
and Appendix~\ref{sec:appandix_A} help
express a state with photon additions and subtractions as
a Gaussian times a polynomial in position representation variables $\mathbf x$ and $\mathbf y$,
then with the help of theorems in Section  \ref{sec:theorems_for_entanglement} it can be determined if the state is NPT
entangled or not, so an experimenter can adjust parameters
when designing a new setup.

For photon substraction a part of an incoming beam (i-th mode) is directed to
the photon number resolving detector. If a single photon is detected in this arm,
then the photon substraction operator was applied on the outgoing beam:
\begin{equation}
\hat\rho_{\mathrm{out}} = 
  \frac{\hat a_i\hat\rho_{\mathrm{in}}\hat a_i^\dagger}
  {\Tr\hat a_i\hat\rho_{\mathrm{in}}\hat a_i^\dagger }.
\end{equation}
If a photon number resolving detector is not available, then only a very small
fraction of the beam must be directed towards the photodetector, so the probability
of detecting more than one photon be negligible.

For photon addition, a more complicated setup is needed. The incoming beam must be sent
to a nondegenerate optical parametric amplifier at the input port for the signal
(and vacuum is entered at the idler port). If a single photon is detected at
the idler (output) mode that means that a photon is added to the signal
mode as well \cite{lvovsky2020productionapplicationsnongaussianquantum}. This implements the photon addition operator:
\begin{equation}
\hat\rho_{\mathrm{out}} = 
  \frac{\hat a_i^\dagger\hat\rho_{\mathrm{in}}\hat a_i}
  {\Tr\hat a_i^\dagger\hat\rho_{\mathrm{in}}\hat a_i}.
\end{equation}
With photon number resolving detectors, multiple photons can be added or subtracted at once
(See e.g. \cite{Jiri2024quantoptExperimentalPreparationmultiphoton}).

Photon addition/subtraction with the combination of other optical elements implementing any symplectic transformation can realize these transformations
(with $\hat{\tilde{a}}$ instead of $\hat a$, see eq. \eqref{eq:tilde_creation_annihilation_def}):
\begin{equation}
\hat\rho^{\mathrm{add}}_{\mathrm{out}} = 
  \frac{\hat{\tilde{a}}_i^\dagger\hat\rho_{\mathrm{in}}\hat{\tilde{a}}_i}
  {\Tr\hat{\tilde{a}}_i^\dagger\hat\rho_{\mathrm{in}}\hat{\tilde{a}}_i},
\quad
\hat\rho^{\mathrm{sub}}_{\mathrm{out}} = 
  \frac{\hat{\tilde{a}}_i\hat\rho_{\mathrm{in}}\hat{\tilde{a}}_i^\dagger}
  {\Tr\hat{\tilde{a}}_i\hat\rho_{\mathrm{in}}\hat{\tilde{a}}_i^\dagger}.
\end{equation}

Repeated application of photon addition/subtraction interspersed
with Gaussian optical elements can realize in principle any physical transformation
in the form of \eqref{eq:op_symbolically_1} (\eqref{eq:op_symbolically_1} also
includes transformations with a non positive semidefinite matrix which are unphysical.). In experiments the requirement of renormalizing
the resulting density operator manifests itself as measurements postselected on the condition
of successful detection of the photons heralding the addition/substraction
operations.

\section{Time evolution of the Wigner function under quadratic Hamiltonians}
\label{Appendix_time_evolution}

\subsection{The Hamiltonian operator}

The general quadratic Hamiltonian for an $N$-mode quantum system can be written as

\begin{equation}
\hat H = \frac{1}{2}
\hat{\mathbf{r}}^\top \mathbf{H} \hat{\mathbf{r}},
\label{eq:quadratic_hamiltonian}
\end{equation}

where $\mathbf{H}$ is a real symmetric positive definite matrix.

\begin{equation}
\mathbf{H}= \mathbf{H}^\top =\begin{pmatrix} \mathbf{H}_{\mathbf{xx}} & \mathbf{H}_{\mathbf{xp}} \\ \mathbf{H}_{\mathbf{xp}}^\top & \mathbf{H}_{\mathbf{pp}} \end{pmatrix} \,\, \text{and} \,\, \hat{\mathbf{r}}=\begin{pmatrix} \hat{\mathbf x} \\ \hat{\mathbf p} \end{pmatrix}.
\label{eq:quadratic_hamiltonian_matrix}
\end{equation}
Here
\begin{equation}
\hat{\mathbf x} = (\hat x_1, \dots, \hat x_N)^\top, \qquad
\hat{\mathbf p} = (\hat p_1, \dots, \hat p_N)^\top,
\end{equation}
with canonical commutation relations
\begin{equation}
[\hat x_i, \hat p_j] = i \hbar \delta_{ij},
\end{equation}
where $[,]$ stands for commutator. 
The blocks $\mathbf{H}_{\mathbf{xx}}$, $\mathbf{H}_{\mathbf{pp}}$ and $\mathbf{H}_{\mathbf{xp}}$ are $N \times N$ matrices.

\subsection{Time evolution of an arbitrary density operator}
The time evolution of the quantum density operator $\hat \rho(t)$ is governed by the von Neumann equation:
\begin{equation}
\frac{d\hat \rho}{dt} = -\frac{i}{\hbar} [\hat H, \hat \rho].
\label{eq:dynamics_of_rho_op}
\end{equation}
With time-independent Hamiltonian  such as \eqref{eq:quadratic_hamiltonian} $\hat{\rho}$ evolves unitarily in time
\[
\hat{\rho}(t) =e^{-i \hat{H}t/\hbar}\hat{\rho}(0)e^{i \hat{H}t/\hbar},
\]
ensuring the property that all the eigenvalues of $\hat{\rho}(t)$ are time-independent.

The associated kernel $K(\mathbf{x},\mathbf{y})=\langle \mathbf{x} | \hat{\rho} | \mathbf{y} \rangle $ is the density matrix.  
The Wigner function \cite[Eq.~5]{japcsi} is
\begin{equation}
W(\mathbf{x},\mathbf{p},t) = \frac{1}{(2\pi \hbar)^N} 
\int_{\mathbb{R}^N} d\mathbf{y}\, e^{-i \mathbf{p}^\top \mathbf{y}/\hbar} 
\left\langle \mathbf{x} + \frac{\mathbf{y}}{2} \Big| \hat \rho(t) \Big| \mathbf{x} - \frac{\mathbf{y}}{2} \right\rangle.
\end{equation}

Taking the time derivative gives
\begin{align}
\frac{\partial W}{\partial t} 
&= \frac{1}{(2\pi \hbar)^N} \int_{\mathbb{R}^N} d\mathbf{y}\, e^{-i \mathbf{p}^\top \mathbf{y}/\hbar} 
\left\langle \mathbf{x} + \frac{\mathbf{y}}{2} \Big| \frac{d \hat \rho}{dt} \Big| \mathbf{x} - \frac{\mathbf{y}}{2} \right\rangle \nonumber \\
&= -\frac{i}{\hbar} \frac{1}{(2\pi \hbar)^N} \int_{\mathbb{R}^N} d\mathbf{y}\, e^{-i \mathbf{p}^\top \mathbf{y}/\hbar} 
\left\langle \mathbf{x} + \frac{\mathbf{y}}{2} \Big| [\hat H, \hat \rho] \Big| \mathbf{x} - \frac{\mathbf{y}}{2} \right\rangle.
\end{align}

Using properties of the commutator and Fourier transforms, this can be written in terms of the Moyal bracket:
\begin{equation}
\frac{\partial W}{\partial t} = \{ H, W \}_\star
= \frac{2}{\hbar} H(\mathbf{x},\mathbf{p}) \,
\sin \Bigg[ \frac{\hbar}{2} 
\big( \overleftarrow{\nabla}_\mathbf{x}^\top \overrightarrow{\nabla}_\mathbf{p} 
- \overleftarrow{\nabla}_\mathbf{p}^\top \overrightarrow{\nabla}_\mathbf{x} \big) 
\Bigg] W(\mathbf{x},\mathbf{p},t).
\end{equation}

Here, $\{H,W\}_\star$ is the Moyal bracket, the phase-space analogue of the quantum Poisson bracket \cite{PAGANI2023169361}. $H(\mathbf{x},\mathbf{p})$ is the classical analogue of the Hamiltonian $\hat H$.

\subsection{Time evolution of an arbitrary Wigner function in case of quadratic Hamiltonian operator}

For the quadratic Hamiltonian \eqref{eq:quadratic_hamiltonian}
\begin{equation}
H=H(\mathbf{x},\mathbf{p}) \equiv H(\mathbf r)
=
\frac{1}{2}\,
\mathbf r^{\top}\mathbf H\,\mathbf r, \quad \text{and} \,\,\,\mathbf r
=
\begin{pmatrix}
\mathbf x \\
\mathbf p
\end{pmatrix}.
\end{equation}
By expanding the sine function as a power series in $\hbar$ it is clear that only the $O(\hbar)$ term remains, $\overleftarrow{\nabla}_\mathbf{x}^\top$ or $\overleftarrow{\nabla}_\mathbf{p}^\top$ cannot act more than twice onto a quadratic $H(\mathbf{x},\mathbf{p})$
:
\begin{equation}
    \frac{\partial W}{\partial t}=\left(\nabla_\mathbf{x} H, \nabla_\mathbf{p} H \right) \, \pmb{\Omega} \,
    \begin{pmatrix} \nabla_\mathbf{x} W \\ \nabla_\mathbf{p} W \end{pmatrix}\equiv \left( \nabla_\mathbf{r} H \right)^\top \pmb{\Omega}  \left(\nabla_\mathbf{r} W \right) ,
\end{equation}
Here
\begin{equation}
\pmb{\Omega}=
\begin{pmatrix}
0 & \mathbbm{1}_N \\
-\mathbbm{1}_N & 0
\end{pmatrix},
\end{equation}
where $\mathbbm{1}_N$ is the $N \times N$ identity matrix. As $\left( \nabla_\mathbf{r} H \right)^\top =\mathbf{r}^\top \mathbf{H} $
\begin{equation}
\frac{\partial W}{\partial t}-\mathbf{r}^\top \mathbf{H}\, \pmb{\Omega}  \left(\nabla_\mathbf{r} W \right)=0.
\label{eq:PDE_for_W}
\end{equation}
Let us suppose that the Wigner-function at time zero is given by a prescribed function
$W_0(\mathbf{r})$, thus the initial condition of \eqref{eq:PDE_for_W} is
\begin{equation}
W(\mathbf{r},t=0)=W_0(\mathbf{r}).
\label{eq:init_cond_for_W}
\end{equation}

The first-order linear PDE \eqref{eq:PDE_for_W} can be solved using the method of characteristics. The left-hand side is a total derivative if one identifies the following equation for the characteristics:
\[
\frac{d \mathbf{r}}{ d t}= \left( -\mathbf{r}^\top \mathbf{H} \pmb{\Omega} \right)^\top=\pmb{\Omega} \mathbf{H} \mathbf{r}.
\]
Its solutions is 
\begin{equation}
\mathbf{r}(s)=\exp\left( s\, \pmb{\Omega} \mathbf{H}\right) \mathbf{r}(0)
\end{equation}
As the total derivative of $W$ with respect to $t$ is zero, $W$ is constant along the characteristics. The next step is to choose such initial conditions where the characteristics $\mathbf{r}(s)$ ends on the phase space point $\mathbf{r}$ at time $t$. Thus, $\mathbf{r}(0)$ should be chosen as $\mathbf{r}(0)=\exp\left( -t \,\pmb{\Omega} \mathbf{H}\right) \mathbf{r}$.

To summarize: the solution of \eqref{eq:PDE_for_W} is
\begin{equation}
W(\mathbf{r},t)=W_0\biggl( \exp\bigl( -t \,\pmb{\Omega} \mathbf{H}\bigr) \mathbf{r}\biggr).
\label{eq:sol_for_W}
\end{equation}
One can verify by direct calculation that \eqref{eq:sol_for_W} satisfies \eqref{eq:PDE_for_W} with initial condition \eqref{eq:init_cond_for_W}. Taking a real symmetric $\mathbf{H}$, $2N \times 2N$ matrix the matrix-exponential in  \eqref{eq:sol_for_W} is a symplectic $2N \times 2N$ real matrix:
\begin{equation}
\mathbf{S}(t)=\exp\bigl( -t \,\pmb{\Omega} \mathbf{H}\bigr), \qquad \mathbf{S}^\top(t) \, \pmb{\Omega} \, \mathbf{S}(t)= \pmb{\Omega}.
\label{eq:simplectic_S(t)}
\end{equation}

Example \ref{unitary_dyn} explicitly shows the unitary time evolution of a multipartite Gaussian together with the entanglements of the subsystems as a function of time. This  example for Eq.~\eqref{eq:sol_for_W} and Claim \ref{c:mainfact1}, 
confirms that the Gaussian part of the polynomial Gaussian Wigner function 
largely determines the entanglement dynamics.

\subsection{The expectation value $\mathbf{d}(t)$ and the covariance matrix $\pmb{\sigma}(t)$}
The time-dependent expectation value of the operator $\hat{\mathbf{r}}$ is defined as
\[
\mathbf{d}(t)=\langle \hat{\mathbf{r}} \rangle_{\hat{\rho}(t)}= \mathrm{Tr}\, \left( \hat{\mathbf{r}} \hat{\rho}(t)\right),
\]
and can be expressed in phase-space as
\begin{equation}
\mathbf{d}(t)=\int_{\mathbb{R}^{2N}} d\mathbf{r} \, \mathbf{r} \,W(\mathbf{r},t). 
\label{eq:def_average}
\end{equation}
We have determined the time evolution of $W(\mathbf{r},t)$ in \eqref{eq:sol_for_W} and showed that the matrix in Eq.~\eqref{eq:simplectic_S(t)} is symplectic with determinant $\det{\mathbf{S}(t)}=1$, thus, by performing a changement of integration variables to $\mathbf{z}=\mathbf{S}(t)\mathbf{r}$ we get the time evolution of $\mathbf{d}(t)$ for the quadratic Hamiltonian as
\begin{equation}
\mathbf{d}(t)=\int_{\mathbb{R}^{2N}} d\mathbf{r} \, \mathbf{r} W_0\Bigl(\mathbf{S}(t)\mathbf{r}\Bigr)=\mathbf{S}^{-1}(t)\int_{\mathbb{R}^{2N}} d\mathbf{z} \, \mathbf{z} W_0(\mathbf{z})=\mathbf{S}^{-1}(t)\mathbf{d}(0),
\label{eq:time_evol_d(t)}
\end{equation}
with $\mathbf{S}^{-1}(t)=\exp\bigl( t \,\pmb{\Omega} \mathbf{H}\bigr)$.

The time-dependent covariance-matrix is given by
\begin{equation}
\pmb{\sigma}_{i,j}(t)=\langle \left\{\hat{\mathbf{r}}_i-\mathbf{d}_i(t),\hat{\mathbf{r}}_j-\mathbf{d}_j(t)\right\} \rangle_{\hat{\rho}(t)}= \mathrm{Tr}\, \left(  \left\{\hat{\mathbf{r}}_i-\mathbf{d}_i(t),\hat{\mathbf{r}}_j-\mathbf{d}_j(t)\right\}\hat{\rho}(t)\right),
\label{eq:def_of_cov_mat}
\end{equation}
which can be converted to phase-space expectations as
\begin{equation}
\frac12 \pmb{\sigma}_{i,j}(t)=\int_{\mathbb{R}^{2N}} d\mathbf{r} \,\Bigl(\mathbf{r}_i-\mathbf{d}_i(t)\Bigr)\Bigl(\mathbf{r}_j-\mathbf{d}_j(t)\Bigr) \, W(\mathbf{r},t).  
\label{eq:def_cov_mat}
\end{equation}
Performing the same trick as for $\mathbf{d}(t)$ and applying  \eqref{eq:time_evol_d(t)} the time evolution of $\pmb{\sigma}(t)$ with quadratic $\hat{H}$ is obtained as
\begin{eqnarray}
\frac12 \pmb{\sigma}_{i,j}(t)&=&\int_{\mathbb{R}^{2N}} d\mathbf{r} \,\Bigl(\mathbf{r}_i-\mathbf{d}_i(t)\Bigr)\Bigl(\mathbf{r}_j-\mathbf{d}_j(t)\Bigr) \, W_0\Bigl(\mathbf{S}(t)\mathbf{r}\Bigr) \nonumber\\
&=& \sum_k \sum_l \mathbf{S}^{-1}_{i,k} \mathbf{S}^{-1}_{j,l} \int_{\mathbb{R}^{2N}} d\mathbf{z} \, \Bigl(\mathbf{z}_k-\mathbf{d}_k(0)\Bigr)\Bigl(\mathbf{z}_l-\mathbf{d}_l(0)\Bigr) \, W_0(\mathbf{z}), \nonumber \\
&=& \sum_k \sum_l \mathbf{S}^{-1}_{i,k} \mathbf{S}^{-1}_{j,l} \pmb{\sigma}_{k,l}(0),
\end{eqnarray}
which  clearly shows that
\begin{equation}
\pmb{\sigma}(t)=\mathbf{S}^{-1}(t) \,\pmb{\sigma}(0) \, \mathbf{S}^{-1\top}(t),
\label{eq:time_evol_Sigma(t)}
\end{equation}
with
\begin{equation}
\mathbf{S}^{-1}(t)=\exp\left( t\pmb{\Omega}\mathbf{H}\right), \qquad \mathbf{S}^{-1\top}(t) =\exp\left(- t\mathbf{H} \pmb{\Omega}\right).
\label{eq:spec_symp_mat}
\end{equation}

In Appendix \ref{sec:appendix_A} we prove that the symplectic eigenvalues of the covariance matrix  \eqref{eq:time_evol_Sigma(t)} are time-independent during the time evolution of a closed system under the influence of a quadratic time-independent Hamiltonian and the system's energy is conserved.

\section{Energy Conservation and Time-Independent Symplectic Eigenvalues}
\label{sec:appendix_A}

\subsection{Calculation of the energy expectation value for quadratic Hamiltonian operator }

With the quadratic Hamiltonian  \eqref{eq:quadratic_hamiltonian} the energy expectation value is
\begin{eqnarray}
\label{eq:energy_exp_val}
\langle\hat H\rangle_{\hat{\rho}(t)}
&=&
\frac12
\int_{\mathbb{R}^{2N}} d\mathbf r \,
\Bigl(\mathbf r^\top \mathbf H \,\mathbf r \Bigr) \, W(\mathbf r,t)=\frac12 \sum_{i,j=1}^N \mathbf{H}_{i,j}\int_{\mathbb{R}^{2N}} d\mathbf r \, \mathbf{r}_i\mathbf{r}_j W(\mathbf r,t) \nonumber \\
&=&\frac12 \sum_{i,j=1}^N \mathbf{H}_{i,j} \left( \pmb{\sigma}_{i,j}(t)+\mathbf{d}_i(t)\mathbf{d}_j(t)\right) \nonumber\\
&=&\frac12\,\mathrm{Tr}\,\left(\mathbf {H}\,\pmb{\sigma}(t)\right)
+
\frac12\,\mathbf d^\top(t) \mathbf H \mathbf d(t),
\label{eq:expect_H(t)}
\end{eqnarray}
where we used the definition of the covariance matrix \eqref{eq:def_cov_mat}

By expanding the exponentials into power-series, it is straightforward to see that 
\[
\mathbf{H}\exp\left( t \pmb{\Omega} \mathbf{H}\right)=\exp\left( t \mathbf{H} \pmb{\Omega} \right) \mathbf{H},
\]
Thus, we have the identity 
\begin{equation}
\exp\left(- t \mathbf{H} \pmb{\Omega} \right) \mathbf{H} \exp\left( t \pmb{\Omega} \mathbf{H}\right) = \mathbf{H}.
\label{eq:identity_2}
\end{equation}
Let us see the first term in Eq.~\eqref{eq:expect_H(t)} and let us use \eqref{eq:time_evol_Sigma(t)} together with \eqref{eq:spec_symp_mat}, and the cyclic property of the trace:
\[
\frac12 \mathrm{Tr}\,\left(\mathbf {H}\,\pmb{\sigma}(t)\right)=\frac12 \mathrm{Tr}\,\Bigl( \exp(-t \mathbf{H}\pmb{\Omega}) \mathbf{H}  \exp(t \pmb{\Omega} \mathbf{H}) \pmb{\sigma}(0)) \Bigr)=\frac12 \mathrm{Tr}\,\left(\mathbf {H}\,\pmb{\sigma}(0)\right),
\]
where in the last step we used the previous identity \eqref{eq:identity_2}. In the second term of \eqref{eq:expect_H(t)} with \eqref{eq:time_evol_d(t)} we have
\[
\frac12 \,\mathbf d^\top(t) \mathbf H \mathbf d(t)= \frac12 \,\mathbf d^\top(0) \exp\left(- t\mathbf{H} \pmb{\Omega}  \right)\mathbf {H} \exp\left(t \pmb{\Omega}\mathbf{H} \right) \mathbf d(0) = \frac12 \,\mathbf d^\top(0) \mathbf H \mathbf d(0).
\]
These expressions show that the expectation value of the energy is conserved.

In the following we apply units of $\hbar=m=1$ with characteristic frequency $\omega=1$ for simplicity.

\subsection{Symplectic eigenvalues of the covariance matrix $\pmb{\sigma}(t)$}

For a general quadratic Hamiltonian \eqref{eq:quadratic_hamiltonian}  the covariance-matrix shows time dependence as the time $t$ evolves. However,the symplectic eigenvalues \cite{serafini2024quantum} of  $\pmb{\sigma}(t)$ are time-independent. To prove this fact let us consider the equation for the characteristic polynomial of the symplectic eigenvalues:
\[
0=\det \left(i \pmb{\Omega} \pmb{\sigma}(t) - \nu \mathbbm{1}_N\right).
\]
The right-hand-side can be transformed by Eq.~\eqref{eq:time_evol_Sigma(t)} as 
\begin{eqnarray}
\det \left(i \pmb{\Omega} \pmb{\sigma}(t) - \nu \mathbbm{1}_N\right)&=&\det \left(i \pmb{\Omega} \mathbf{S}^{-1}(t)\pmb{\sigma}(0)\mathbf{S}^{-1\top}(t) - \nu \mathbbm{1}_N\right) \nonumber \\
&=& \det \Bigl( \mathbf{S}^{\top}(t) \left(i  \mathbf{S}^{-1\top}(t) \pmb{\Omega} \mathbf{S}^{-1}(t) \pmb{\sigma}(0) -\nu \mathbbm{1}_N\right) \mathbf{S}^{-1\top}(t) \Bigr) \nonumber \\
&=& \det\Bigl( \mathbf{S}^{\top}(t) \Bigr) \det \Bigl( i   \pmb{\Omega}\pmb{\sigma}(0) -\nu \mathbbm{1}_N \Bigr) \det\Bigl( \mathbf{S}^{-1\top}(t) \Bigr) \nonumber \\
&=&\det \left(i \pmb{\Omega} \pmb{\sigma}(0) - \nu \mathbbm{1}_N\right),
\label{eq:symplectic_eigenvalues(t)}
\end{eqnarray}
where we used the fact that if $\mathbf{S}(t)$ is symplectic then $\mathbf{S}^{-1\top}(t)$ is also symplectic, and thus, $\mathbf{S}^{-1\top}(t) \pmb{\Omega} \mathbf{S}^{-1}(t)=\pmb{\Omega}$. Eq.~\eqref{eq:symplectic_eigenvalues(t)} clearly shows that the symplectic eigenvalues of $\pmb{\sigma}(t)$ agree with those of $\pmb{\sigma}(0)$, therefore the symplectic eigenvalues of the covariance matrix are time independent. 

\section{Connections between polynomial and operator parameters}
\label{sec:appandix_A}

This appendix is devoted to polynomials of at most second order. As the number of independent coefficients of operator parameters $c_0$, $\mathbf{c}^{(1)}$, $\mathbf{c}^{(2)}$, $\mathbf{c}^{(3)}$ are the same as of the set of the  polynomial parameters $\mathbf{A}^{(p)}$, $\mathbf{B}^{(p)}$, $\mathbf{C}^{(p)}$, $\mathbf{D}^{(p)}$, $\mathbf{E}^{(p)}$ and $\mathrm{F}^{(p)}$ one can expect linear relations between them. These relations, however, are strongly influenced by the parameters of the Gaussian state. These relations can be derived in full generality. For now we focus on GKLC Gaussian parameters, which have the following special properties (See \eqref{eq:GKLC_G_density_matrix}): i) No linear terms in the exponent, ii) The covariance matrix $\pmb{\sigma}_{GKLC}$ is locally symmetric \cite{serafini2024quantum}, which leads to locally symmetric coefficient matrices $\mathbf{A}$ and $\mathbf{C}$, having identical diagonal and identical off-diagonal elements (See \eqref{eq:GKLC_A_and_C}), parametrized by $p$ and $q$. 

Local symmetry in matrices $\mathbf{A}$ and $\mathbf{C}$ induces permutation symmetries in the polynomial and operator parameters: e.g., if one has an expression for $\mathbf{A}_{1,1}$, one can get the corresponding relation for  $\mathbf{A}_{2,2}$ with a permutation of the matrix indices on the right hand side $1 \to 2 \to 3 \to 1$ and applying the symmetry properties of the operator parameters, if necessary, e.g., $\textrm{Re}\, \mathbf{c}_{1,3}^{(2)} \to \textrm{Re} \,\mathbf{c}_{2,1}^{(2)} = \textrm{Re} \,\mathbf{c}_{1,2}^{(2)}$,  or $\textrm{Im}\, \mathbf{c}_{1,3}^{(3)} \to \textrm{Im} \,\mathbf{c}_{2,1}^{(3)} = -\textrm{Im} \,\mathbf{c}_{1,2}^{(3)}$, etc.
Permutation symmetry is also present for the linear and constant coefficients. 

The first set of nontrivial relations is as follows:

\begin{eqnarray}
\mathbf{F}^{(p)}&=&\textbf{c}^{(0)}-\frac{\left(1+\sqrt{pq}\right)}{4 \sqrt{pq}} \textrm{Re}\left[ 2 \left(1+\sqrt{pq}\right) \left(\mathbf{c}^{(2)}_{1,1}+\mathbf{c}^{(2)}_{2,2}+\mathbf{c}^{(2)}_{3,3}\right) \right. \nonumber \\ 
&&\qquad\qquad\qquad \left.+\left(1-\sqrt{pq}\right)\left(\mathbf{c}^{(3)}_{1,1}+\mathbf{c}^{(3)}_{2,2}+\mathbf{c}^{(3)}_{3,3}\right) \right]
\end{eqnarray}

\begin{eqnarray}
\mathbf{E}^{(p)}_{1}&=& \frac{\left(1+\sqrt{pq}\right) \sqrt[4]{p q}}{3\sqrt{2} p q} \textrm{Re}\left[  \left(\sqrt{p}+2 \sqrt{q}\right)\mathbf{c}^{(1)}_{1}+\left(\sqrt{p}-\sqrt{q}\right)\left(\mathbf{c}^{(1)}_{2}+\mathbf{c}^{(1)}_{3}\right) \right], \nonumber\\
\end{eqnarray}

\begin{eqnarray}
\mathbf{D}^{(p)}_{1} &=& \frac{\left(1+ \sqrt{pq}\right)}{3 \sqrt{2}\sqrt[4]{pq}} \textrm{Im}\left[   \left(\sqrt{p}+2 \sqrt{q}\right)\mathbf{c}^{(1)}_{1}+\left(\sqrt{p}-\sqrt{q}\right)\left(\mathbf{c}^{(1)}_{2}+\mathbf{c}^{(1)}_{3}\right) \right],
\end{eqnarray}

\begin{eqnarray}
\mathbf{A}^{(p)}_{1,1}
&=&\frac{\left(1+\sqrt{pq}\right)^2 }{72 \sqrt{p q}}\textrm{Re}\left[
\left(\sqrt{p}+2 \sqrt{q}\right)^2 \left(2 \mathbf{c}^{(2)}_{1,1}-\mathbf{c}^{(3)}_{1,1}\right) \right.\nonumber \\
&& \left.+\left(\sqrt{p}-\sqrt{q}\right) \left(\sqrt{p}+2 \sqrt{q}\right) 
\left(4 \mathbf{c}^{(2)}_{1,2}+4 \mathbf{c}^{(2)}_{1,3}-2 \mathbf{c}^{(3)}_{1,2}-2 \mathbf{c}^{(3)}_{1,3}\right) \right. \nonumber\\ 
&& \left. +\left(\sqrt{p}-\sqrt{q}\right)^2
\left(2 \mathbf{c}^{(2)}_{2,2}+4 \mathbf{c}^{(2)}_{2,3}+2 \mathbf{c}^{(2)}_{3,3}-\mathbf{c}^{(3)}_{2,2}-2 \mathbf{c}^{(3)}_{2,3}-\mathbf{c}^{(3)}_{3,3}\right) \right],
\end{eqnarray}

\begin{eqnarray}
\mathbf{A}^{(p)}_{1,2}&=&
\frac{\left(1+\sqrt{p q}\right)^2}{72 \sqrt{p q}} \textrm{Re}\left[
\left(2 p+5 q+2 \sqrt{p} \sqrt{q}\right) \left(2 \mathbf{c}^{(2)}_{1,2}-\mathbf{c}^{(3)}_{1,2}\right) \right. \nonumber \\
&& \left.
+\left(\sqrt{p}-\sqrt{q}\right)\left(\sqrt{p}+2 \sqrt{q}\right)
\left(2 \mathbf{c}^{(2)}_{1,1}+2 \mathbf{c}^{(2)}_{2,2}-\mathbf{c}^{(3)}_{1,1}-\mathbf{c}^{(3)}_{2,2}\right) \right. \nonumber \\
&& \left.
+\left(\sqrt{p}-\sqrt{q}\right)\left(2 \sqrt{p}+\sqrt{q}\right) \left(2 \mathbf{c}^{(2)}_{1,3}+2 \mathbf{c}^{(2)}_{2,3}-\mathbf{c}^{(3)}_{1,3}-\mathbf{c}^{(3)}_{2,3}\right) \right. \nonumber \\
&& \left.
+\left(\sqrt{p}-\sqrt{q}\right)^2 \left(2 \mathbf{c}^{(2)}_{3,3}-\mathbf{c}^{(3)}_{3,3}\right) 
\right],
\end{eqnarray}

\begin{eqnarray}
\mathbf{B}^{(p)}_{1,1}&=& 
\frac{\left(1+\sqrt{p q}\right)^2}{18 p q} \textrm{Im}\left[
\left(\sqrt{p}+2 \sqrt{q}\right)^2 \mathbf{c}^{(2)}_{1,1} \right.
\nonumber \\
&&\left. +2\left(\sqrt{p}-\sqrt{q}\right)\left(\sqrt{p}+2 \sqrt{q}\right) \left( \mathbf{c}^{(2)}_{1,2}+ \mathbf{c}^{(2)}_{1,3}\right) \right. \nonumber \\
&&\left. +\left(\sqrt{p}-\sqrt{q}\right)^2 \left(\mathbf{c}^{(2)}_{2,2} + 2\mathbf{c}^{(2)}_{2,3}+\mathbf{c}^{(2)}_{3,3}\right) 
\right],
\end{eqnarray}

\begin{eqnarray}
\mathbf{B}^{(p)}_{1,2}
&=& \frac{\left(1+\sqrt{p q}\right)^2}{36 p q} \textrm{Im} \left[ 
2 \left(2 p+5 q+2 \sqrt{p} \sqrt{q}\right) \mathbf{c}^{(2)}_{1,2} \right. \nonumber \\
&&+\left. 2\left(\sqrt{p}-\sqrt{q}\right)\left(\sqrt{p}+2 \sqrt{q}\right) \left(\mathbf{c}^{(2)}_{1,1}+\mathbf{c}^{(2)}_{2,2}\right)  \right.
\nonumber \\
&&\left. +2 \left(\sqrt{p}-\sqrt{q}\right) \left(2 \sqrt{p}+\sqrt{q}\right)\left( \mathbf{c}^{(2)}_{1,3}+ \mathbf{c}^{(2)}_{2,3}\right) \right.
\nonumber \\
&&\left. +2 \left(\sqrt{p}-\sqrt{q}\right)^2\mathbf{c}^{(2)}_{3,3} 
+3 \sqrt{q} \left(2 \sqrt{p}+\sqrt{q}\right)\mathbf{c}^{(3)}_{1,2}  \right.
\nonumber \\
&&\left. +3 \sqrt{q} \left(\sqrt{p}-\sqrt{q}\right)\left( \mathbf{c}^{(3)}_{1,3}- \mathbf{c}^{(3)}_{2,3}\right) \right],
\end{eqnarray}

\begin{eqnarray}
\mathbf{B}^{(p)}_{2,1}
&=& \frac{\left(1+\sqrt{p q}\right)^2}{36 p q} \textrm{Im} \left[ 
2 \left(2 p+5 q+2 \sqrt{p} \sqrt{q}\right) \mathbf{c}^{(2)}_{1,2} \right. \nonumber \\
&&+\left. 2\left(\sqrt{p}-\sqrt{q}\right)\left(\sqrt{p}+2 \sqrt{q}\right) \left(\mathbf{c}^{(2)}_{1,1}+\mathbf{c}^{(2)}_{2,2}\right)  \right.
\nonumber \\
&&\left. +2 \left(\sqrt{p}-\sqrt{q}\right) \left(2 \sqrt{p}+\sqrt{q}\right)\left( \mathbf{c}^{(2)}_{1,3}+ \mathbf{c}^{(2)}_{2,3}\right) \right.
\nonumber \\
&&\left. +2 \left(\sqrt{p}-\sqrt{q}\right)^2\mathbf{c}^{(2)}_{3,3} 
- 3\sqrt{q} \left(2 \sqrt{p}+\sqrt{q}\right)\mathbf{c}^{(3)}_{1,2}  \right.
\nonumber \\
&&\left. -3 \sqrt{q} \left(\sqrt{p}-\sqrt{q}\right)\left( \mathbf{c}^{(3)}_{1,3}- \mathbf{c}^{(3)}_{2,3}\right) \right],
\end{eqnarray}

\begin{eqnarray}
\mathbf{C}^{(p)}_{1,1}
&=&\frac{\left(1+\sqrt{p q}\right)^2}{72 p^{3/2} q^{3/2}}
\textrm{Re}\left[
\left(\sqrt{p}+2 \sqrt{q}\right)^2 \left(2 \mathbf{c}^{(2)}_{1,1}+\mathbf{c}^{(3)}_{1,1}\right) \right.\nonumber \\
&& \left.+\left(\sqrt{p}-\sqrt{q}\right) \left(\sqrt{p}+2 \sqrt{q}\right) 
\left(4 \mathbf{c}^{(2)}_{1,2}+4 \mathbf{c}^{(2)}_{1,3}+2 \mathbf{c}^{(3)}_{1,2}+2 \mathbf{c}^{(3)}_{1,3}\right) \right. \nonumber\\ 
&& \left. +\left(\sqrt{p}-\sqrt{q}\right)^2
\left(2 \mathbf{c}^{(2)}_{2,2}+4 \mathbf{c}^{(2)}_{2,3}+2 \mathbf{c}^{(2)}_{3,3}+\mathbf{c}^{(3)}_{2,2}+2 \mathbf{c}^{(3)}_{2,3}+\mathbf{c}^{(3)}_{3,3}\right) \right],
\end{eqnarray}

\begin{eqnarray}
\mathbf{C}^{(p)}_{1,2}
&=&\frac{\left(1+\sqrt{p q}\right)^2}{72 p^{3/2} q^{3/2}}
\textrm{Re}\left[
\left(2 p+5 q+2 \sqrt{p} \sqrt{q}\right) \left(2 \mathbf{c}^{(2)}_{1,2}+\mathbf{c}^{(3)}_{1,2}\right) \right. \nonumber \\
&& \left.
+\left(\sqrt{p}-\sqrt{q}\right)\left(\sqrt{p}+2 \sqrt{q}\right)
\left(2 \mathbf{c}^{(2)}_{1,1}+2 \mathbf{c}^{(2)}_{2,2}+\mathbf{c}^{(3)}_{1,1}+\mathbf{c}^{(3)}_{2,2}\right) \right. \nonumber \\
&& \left.
+\left(\sqrt{p}-\sqrt{q}\right)\left(2 \sqrt{p}+\sqrt{q}\right) \left(2 \mathbf{c}^{(2)}_{1,3}+2 \mathbf{c}^{(2)}_{2,3}+\mathbf{c}^{(3)}_{1,3}+\mathbf{c}^{(3)}_{2,3}\right) \right. \nonumber \\
&& \left.
+\left(\sqrt{p}-\sqrt{q}\right)^2 \left(2 \mathbf{c}^{(2)}_{3,3}+\mathbf{c}^{(3)}_{3,3}\right) 
\right].
\end{eqnarray}

The inverse relations are as follows:

\begin{eqnarray}
\mathbf{c}^{(0)}&=&\mathrm{F}^{(p)}+\frac{2}{3 \sqrt{pq} \left(1+\sqrt{p q}\right)} \left[
(2 p+q) \left(\mathbf{A}^{(p)}_{1,1}+\mathbf{A}^{(p)}_{2,2}+\mathbf{A}^{(p)}_{3,3}\right) \right. \nonumber \\
&& \left.
-2 (p-q) (\mathbf{A}^{(p)}_{1,2}+\mathbf{A}^{(p)}_{1,3}+\mathbf{A}^{(p)}_{2,3})-2 \sqrt{pq} (p-q) (\mathbf{C}^{(p)}_{1,2}+\mathbf{C}^{(p)}_{1,3}+\mathbf{C}^{(p)}_{2,3})
 \right. \nonumber \\
&&\left. +\sqrt{pq}(2 p+q) (\mathbf{C}^{(p)}_{1,1}+\mathbf{C}^{(p)}_{2,2}+\mathbf{C}^{(p)}_{3,3}) \right],
\end{eqnarray}

\begin{eqnarray}
\textrm{Re}\left(\mathbf{c}^{(1)}_{1}\right)=\frac{\sqrt{2}\sqrt[4]{pq}}{3 \left(1+\sqrt{p q}\right)} \left[\left(2 \sqrt{p}+\sqrt{q}\right) \mathbf{E}^{(p)}_{1}-\left(\sqrt{p}-\sqrt{q}\right) \left(\mathbf{E}^{(p)}_{2}+\mathbf{E}^{(p)}_{3}\right)\right], \nonumber \\
\end{eqnarray}

\begin{eqnarray}
\textrm{Im}\left(\mathbf{c}^{(1)}_1\right)=\frac{\sqrt{2}}{3 \sqrt[4]{p q} \left(1+\sqrt{p q}\right)} \left[\left(2 \sqrt{p}+\sqrt{q}\right) \mathbf{D}^{(p)}_{1}-\left(\sqrt{p}-\sqrt{q}\right) (\mathbf{D}^{(p)}_{2}+\mathbf{D}^{(p)}_{3})\right], \nonumber \\
\end{eqnarray}

\begin{eqnarray}
\textrm{Re}\left(\mathbf{c}^{(2)}_{1,1}\right)&=&\frac{2}{9 \sqrt{pq}\left( 1+ \sqrt{pq}\right)^2 } \left[ 
\left(2 \sqrt{p}+\sqrt{q}\right)^2 \mathbf{A}^{(p)}_{1,1} 
\right.\nonumber \\
&& \left. -2 \left(\sqrt{p}-\sqrt{q}\right) \left(2 \sqrt{p}+\sqrt{q}\right) \left( \mathbf{A}^{(p)}_{1,2}+\mathbf{A}^{(p)}_{1,3}\right) \right. \nonumber \\
&& \left. +\left(\sqrt{p}-\sqrt{q}\right)^2 \left(\mathbf{A}^{(p)}_{2,2}+2 \mathbf{A}^{(p)}_{2,3}+\mathbf{A}^{(p)}_{3,3}\right) \right. \nonumber \\
&& \left. +p q \left(2 \sqrt{p}+\sqrt{q}\right)^2 \mathbf{C}^{(p)}_{1,1}  \right. 
\left. -2 p q\left(\sqrt{p}-\sqrt{q}\right)\left(2 \sqrt{p}+\sqrt{q}\right)  (\mathbf{C}^{(p)}_{1,2}+\mathbf{C}^{(p)}_{1,3}) \right. \nonumber \\
&& \left. +p q \left(\sqrt{p}-\sqrt{q}\right)^2(\mathbf{C}^{(p)}_{2,2}+2 \mathbf{C}^{(p)}_{2,3}+\mathbf{C}^{(p)}_{3,3}) \right],
\end{eqnarray}

\begin{eqnarray}
\textrm{Re}\left(\mathbf{c}^{(2)}_{1,2}\right)
&=&\frac{2}{9 \sqrt{pq} \left(1+ \sqrt{p q}\right)^2 }
\left[
\left(5 p+2 q+2 \sqrt{p} \sqrt{q}\right) \mathbf{A}^{(p)}_{1,2} \right.\nonumber \\
&&
\left. -\left(\sqrt{p}-\sqrt{q}\right) \left(2 \sqrt{p}+\sqrt{q}\right) \left(\mathbf{A}^{(p)}_{1,1}+\mathbf{A}^{(p)}_{2,2}\right)
\right. \nonumber \\
&& \left.
-\left(\sqrt{p}-\sqrt{q}\right)\left(\sqrt{p}+2 \sqrt{q}\right) \left(\mathbf{A}^{(p)}_{1,3}+\mathbf{A}^{(p)}_{2,3}\right) 
+\left(\sqrt{p}-\sqrt{q}\right)^2\mathbf{A}^{(p)}_{3,3} \right. \nonumber \\
&&
\left. +p q \left(5 p+2 q+2 \sqrt{p} \sqrt{q}\right) \mathbf{C}^{(p)}_{1,2} \right. \nonumber \\
&& \left.-pq \left(\sqrt{p}-\sqrt{q}\right) \left(2 \sqrt{p}+\sqrt{q}\right) 
\left(\mathbf{C}^{(p)}_{1,1}+\mathbf{C}^{(p)}_{2,2}\right) \right. \nonumber \\
&& \left.-pq \left(\sqrt{p}-\sqrt{q}\right) \left(\sqrt{p}+2 \sqrt{q}\right)  \left(\mathbf{C}^{(p)}_{1,3}+\mathbf{C}^{(p)}_{2,3}\right) 
+p q \left(\sqrt{p}-\sqrt{q}\right)^2 \mathbf{C}^{(p)}_{3,3} 
\right],\nonumber\\
\end{eqnarray}

\begin{eqnarray}
\textrm{Im}\left(\mathbf{c}^{(2)}_{1,1}\right)
&=&\frac{2}{9 \left(1+\sqrt{pq}\right)^2}
\left[
\left(2 \sqrt{p}+\sqrt{q}\right)^2\mathbf{B}^{(p)}_{1,1}  \right. \nonumber \\
&& \left.
-\left(\sqrt{p}-\sqrt{q}\right) \left(2 \sqrt{p}+\sqrt{q}\right) \left(\mathbf{B}^{(p)}_{1,2}+\mathbf{B}^{(p)}_{1,3}+\mathbf{B}^{(p)}_{2,1}+\mathbf{B}^{(p)}_{3,1}\right) \right. \nonumber \\
&& \left.
+\left(\sqrt{p}-\sqrt{q}\right)^2\left(\mathbf{B}^{(p)}_{2,2}+\mathbf{B}^{(p)}_{2,3}+\mathbf{B}^{(p)}_{3,2}+\mathbf{B}^{(p)}_{3,3}\right) 
\right],
\end{eqnarray}

\begin{eqnarray}
\textrm{Im}\left(\mathbf{c}^{(2)}_{1,2}\right)
&=&
\frac{2}{18 \left(1+\sqrt{p q}\right)^2} \left[
\left(5 p+2 q+2 \sqrt{pq}\right) \left(\mathbf{B}^{(p)}_{1,2}+\mathbf{B}^{(p)}_{2,1}\right) \right. \nonumber \\
&&\left.-2 \left(\sqrt{p}-\sqrt{q}\right)\left(2 \sqrt{p}+\sqrt{q}\right)  (\mathbf{B}^{(p)}_{1,1}+\mathbf{B}^{(p)}_{2,2}) \right. \nonumber \\
&&\left. -\left(\sqrt{p}-\sqrt{q}\right)\left(\sqrt{p}+2 \sqrt{q}\right) \left(\mathbf{B}^{(p)}_{1,3}+\mathbf{B}^{(p)}_{2,3}+\mathbf{B}^{(p)}_{3,1}+\mathbf{B}^{(p)}_{3,2}\right) \right. \nonumber \\
&&  \left. +2 \mathbf{B}^{(p)}_{3,3} \left(\sqrt{p}-\sqrt{q}\right)^2
\right],
\end{eqnarray}

\begin{eqnarray}
\textrm{Re}\left(\mathbf{c}^{(3)}_{1,1}\right)	
&=&\frac{4}{9 \sqrt{pq} \left(1+\sqrt{p} \sqrt{q}\right)^2}
\left[
-\left(2 \sqrt{p}+\sqrt{q}\right)^2\mathbf{A}^{(p)}_{1,1} \right.\nonumber \\
&& \left.+2 \left(\sqrt{p}-\sqrt{q}\right) \left(2 \sqrt{p}+\sqrt{q}\right) \left(\mathbf{A}^{(p)}_{1,2}+\mathbf{A}^{(p)}_{1,3} \right)  \right. \nonumber \\
&& \left.
-\left(\sqrt{p}-\sqrt{q}\right)^2 \left(\mathbf{A}^{(p)}_{2,2}+2 \mathbf{A}^{(p)}_{2,3}+\mathbf{A}^{(p)}_{3,3}\right)  \right. \nonumber \\
&& \left. +p q \left(2 \sqrt{p}+\sqrt{q}\right)^2 \mathbf{C}^{(p)}_{1,1} 
-2 pq  \left(\sqrt{p}-\sqrt{q}\right) \left(2 \sqrt{p}+\sqrt{q}\right)  \left(\mathbf{C}^{(p)}_{1,2}+\mathbf{C}^{(p)}_{1,3} \right) \right. \nonumber \\
&& \left. +p q \left(\sqrt{p}-\sqrt{q}\right)^2
\left(\mathbf{C}^{(p)}_{2,2}+2 \mathbf{C}^{(p)}_{2,3}+\mathbf{C}^{(p)}_{3,3}\right) 
\right],
\end{eqnarray}

\begin{eqnarray}
\textrm{Re}\left(\mathbf{c}^{(3)}_{1,2}\right)
&=&\frac{4}{9 \sqrt{pq} \left(1+\sqrt{p q}\right)^2 } 
\left[
\left(\sqrt{p}-\sqrt{q}\right)\left(2 \sqrt{p}+\sqrt{q}\right) \left(\mathbf{A}^{(p)}_{1,1}+\mathbf{A}^{(p)}_{2,2}\right) \right. \nonumber \\
&&
\left. -\left(5 p+2 q+2 \sqrt{p q}\right) \mathbf{A}^{(p)}_{1,2}
+\left(\sqrt{p}-\sqrt{q}\right)\left(\sqrt{p}+2 \sqrt{q}\right) \left(\mathbf{A}^{(p)}_{1,3}+\mathbf{A}^{(p)}_{2,3}\right)  \right. \nonumber \\
&&
\left.-\left(\sqrt{p}-\sqrt{q}\right)^2 \mathbf{A}^{(p)}_{3,3} 
-pq \left(\sqrt{p}-\sqrt{q}\right)\left(2 \sqrt{p}+\sqrt{q}\right)  \left(\mathbf{C}^{(p)}_{1,1}+\mathbf{C}^{(p)}_{2,2}\right)  \right. \nonumber \\ 
&&
\left. +p q \left(5 p+2 q+2 \sqrt{p q}\right) \mathbf{C}^{(p)}_{1,2}
-pq \left(\sqrt{p}-\sqrt{q}\right)\left(\sqrt{p}+2 \sqrt{q}\right)(\mathbf{C}^{(p)}_{1,3}+\mathbf{C}^{(p)}_{2,3})  \right. \nonumber \\
&&
\left. +p q \left(\sqrt{p}-\sqrt{q}\right)^2\mathbf{C}^{(p)}_{3,3} 
\right],
\end{eqnarray}

\begin{eqnarray}
\textrm{Im}\left(\mathbf{c}^{(3)}_{1,2}\right)
&=&\frac{2\sqrt{p}}{3 \left(1+\sqrt{p q}\right)^2} \left[
\left(\sqrt{p}+2 \sqrt{q}\right)\left(\mathbf{B}^{(p)}_{1,2}-\mathbf{B}^{(p)}_{2,1}\right) \right. \nonumber \\
&& \left. +\left(\sqrt{p}-\sqrt{q}\right) \left(\mathbf{B}^{(p)}_{3,1}-\mathbf{B}^{(p)}_{1,3}+\mathbf{B}^{(p)}_{2,3}-\mathbf{B}^{(p)}_{3,2}\right) 
 \right].
\end{eqnarray}

\printbibliography

@book{Neumann,
author = {Neumann, J.~v.},
publisher = {Springer},
title = {Mathematische Grundlagen der Quantenmechanik},
year = {1932},
}

@article{HPZ,
  title = {Quantum Brownian motion in a general environment: Exact master equation with nonlocal dissipation and colored noise},
  author = {Hu, B. L. and Paz, Juan Pablo and Zhang, Yuhong},
  journal = {Phys. Rev. D},
  volume = {45},
  issue = {8},
  pages = {2843--2861},
  numpages = {0},
  year = {1992},
  publisher = {American Physical Society},
  doi = {10.1103/PhysRevD.45.2843},
 }

@article{H-Y,
  title = {Alternative derivation of the Hu-Paz-Zhang master equation of quantum Brownian motion},
  author = {Halliwell, J. J. and Yu, T.},
  journal = {Phys. Rev. D},
  volume = {53},
  issue = {4},
  pages = {2012--2019},
  numpages = {0},
  year = {1996},
  publisher = {American Physical Society},
  doi = {10.1103/PhysRevD.53.2012},
 }

@article{HBCSCS,
  title = {Range of applicability of the Hu-Paz-Zhang master equation},
  author = {Homa, G. and Csord\'as, A. and Csirik, M.~A. and Bern\'ad, J.~Z.},
  journal = {Phys. Rev. A},
  volume = {102},
  issue = {2},
  pages = {022206},
  numpages = {13},
  year = {2020},
 publisher = {American Physical Society},
  doi = {10.1103/PhysRevA.102.022206},
 }

@article{Newton,
doi = {10.1088/1751-8121/acc147},
year = {2023},
publisher = {IOP Publishing},
volume = {56},
issue = {14},
pages = {145203},
author = {G.~Homa and R.~Balka and J.~Z.~Bern\'ad and M.~K\'aroly and A.~Csord\'as},
title = {Newton’s identities and positivity of trace class integral operators},
journal = {J.~Phys.~A: Math.~Theor.}
}

@article{HCSB,
  title = {Analytical evaluation of the coefficients of the Hu-Paz-Zhang master equation: Ohmic spectral density, zero temperature, and consistency check},
  author = {Homa, G. and Bern\'ad, J. Z. and Csord\'as, A.},
  journal = {Phys.~Rev.~A},
  volume = {108},
  issue = {1},
  pages = {012210},
  numpages = {15},
  year = {2023},
  publisher = {American Physical Society},
  doi = {10.1103/PhysRevA.108.012210},
  }

@article{SimonR ,
  title = {Peres-Horodecki Separability Criterion for Continuous Variable Systems},
  author = {Simon, R.},
  journal = {Phys. Rev. Lett.},
  volume = {84},
  issue = {12},
  pages = {2726--2729},
  numpages = {0},
  year = {2000},
  publisher = {American Physical Society},
  doi = {10.1103/PhysRevLett.84.2726},
  }

@Article{Bernád2018,
author={Bern{\'a}d, J.~Z.
and Homa, G.
and Csirik, M.~A.},
title={An entropy production based method for determining the position diffusion's coefficient of a quantum Brownian motion},
journal={Eur.~Phys.~J.~D},
year={2018},
volume={72},
issue={12},
pages={212},
doi={10.1140/epjd/e2018-90476-0},
}

@book{Stein,
author= "Stein, E.~M. and Shakarchi, R.",
title= "{Real analysis: measure theory, integration, and Hilbert
                       spaces}",
publisher= "Princeton Univ. Press",
series        = "Princeton lectures in analysis",
year          = "2005",
}

@Book{Duflo,
author={Duflo, M.},
title={G{\'e}n{\'e}ralit{\'e}s sur les repr{\'e}sentations induites },
year={1972},
pages = {93-119},
publisher={Groupes de Lie \newline R{\'e}solubles, Monographies de la Soc.~Math.~de France, vol.~4, Dunod, Paris }
}

@book{deGosson_book,
title= {Symplectic Geometry and Quantum Mechanics},
author={de Gosson, M.~A.},
 lccn={2016346359},
  series={Operator Theory:  Advances and Applications},
doi={10.1007/3-7643-7575-2},
   year={ 2006 },
  publisher={Birkh{\"a}user Basel}
}

@article{Perescrit,
  title = {Separability Criterion for Density Matrices},
  author = {Peres, A.},
  journal = {Phys.~Rev. Lett.},
  volume = {77},
  issue = {8},
  pages = {1413--1415},
  numpages = {0},
  year = {1996},
  publisher = {American Physical Society},
  doi = {10.1103/PhysRevLett.77.1413},
 }

@article{Neumann1927,
author = {Neumann, J.~v.},
journal = {Nachrichten von der Gesellschaft der Wissenschaften zu Göttingen, Mathematisch-Physikalische Klasse},
pages = {245-272},
title = {Wahrscheinlichkeitstheoretischer Aufbau der \\ Quantenmechanik},
volume = {1927},
year = {1927},
}

@book{lax2002functional,
  title={Functional Analysis},
  author={Lax, P.~D.},
  lccn={01046547},
  series={Pure and Applied Mathematics: A Wiley Series of Texts, Monographs and Tracts},
  year={2002},
  publisher={Wiley}
}

@article{weedbrook2012gaussian,
  title={Gaussian quantum information},
  author={Weedbrook, Christian and Pirandola, Stefano and Garc{\'\i}a-Patr{\'o}n, Ra{\'u}l and Cerf, Nicolas J. and Ralph, Timothy C. and Shapiro, Jeffrey H. and Lloyd, Seth},
  journal={Reviews of Modern Physics},
  volume={84},
  pages={621},
  year={2012}
}

@review{adesso2014continuous,
  title={Continuous variable quantum information: Gaussian states and beyond},
  author={Adesso, Gerardo and Ragy, Sammy and Lee, Antony R.},
  journal={Open Systems \& Information Dynamics},
  volume={21},
  year={2014}
}

@article{adesso2007entanglement,
  title={Entanglement of Gaussian states},
  author={Adesso, Gerardo and Illuminati, Fabrizio},
  journal={Journal of Physics A},
  volume={40},
  year={2007}
}

@article{eisert2002distilling,
  title={Distilling Gaussian states with Gaussian operations is impossible},
  author={Eisert, Jens and Scheel, Stefan and Plenio, Martin B.},
  journal={Physical Review Letters},
  volume={89},
  year={2002}
}

@article{fiurasek2002gaussian,
  title={Gaussian transformations and distillation of entangled Gaussian states},
  author={Fiur{\'a}{\v{s}}ek, Jarom{\'\i}r},
  journal={Physical Review Letters},
  volume={89},
  year={2002}
}

@article{giedke2002characterization,
  title={Characterization of Gaussian operations and distillation of Gaussian states},
  author={Giedke, G. and Cirac, J. I.},
  journal={Physical Review A},
  volume={66},
  year={2002}
}

@article{braunstein2005quantum,
  title={Quantum information with continuous variables},
  author={Braunstein, Samuel L. and van Loock, Peter},
  journal={Reviews of Modern Physics},
  volume={77},
  year={2005}
}

@article{zavatta2004quantum,
  title={Quantum-to-classical transition with single-photon-added coherent states},
  author={Zavatta, Alessandro and Viciani, Silvia and Bellini, Marco},
  journal={Science},
  volume={306},
  year={2004}
}

@article{parigi2007probing,
  title={Probing quantum commutation rules by addition and subtraction of single photons},
  author={Parigi, Valentina and Zavatta, Alessandro and Kim, M. S. and Bellini, Marco},
  journal={Science},
  volume={317},
  year={2007}
}

@article{kirchmair2013observation,
  title={Observation of quantum state collapse and revival due to the single-photon Kerr effect},
  author={Kirchmair, Gerhard and others},
  journal={Nature},
  volume={495},
  year={2013}
}

@article{filip2005measurement,
  title={Measurement-induced continuous-variable quantum interactions},
  author={Filip, Radim and Marek, Petr},
  journal={Physical Review A},
  volume={71},
  year={2005}
}

@article{fiurasek2009engineering,
  title={Engineering quantum operations on traveling light beams by multiple photon addition and subtraction},
  author={Fiur{\'a}{\v{s}}ek, Jarom{\'\i}r},
  journal={Physical Review A},
  volume={80},
  year={2009}
}

@article{genoni2010quantifying,
  title={Quantifying non-Gaussianity for quantum information},
  author={Genoni, Marco G. and Paris, Matteo G. A. and Banaszek, Konrad},
  journal={Physical Review A},
  volume={82},
  year={2010}
}

@article{horodecki1996separability,
  title={Separability of mixed states: necessary and sufficient conditions},
  author={Horodecki, Micha{\l} and Horodecki, Pawe{\l} and Horodecki, Ryszard},
  journal={Physics Letters A},
  volume={223},
  year={1996}
}

@article{werner2001bound,
  title={Bound entangled Gaussian states},
  author={Werner, Reinhard F. and Wolf, Michael M.},
  journal={Physical Review Letters},
  volume={86},
  year={2001}
}

@article{rabl2011photon,
  title={Photon blockade effect in optomechanical systems},
  author={Rabl, Peter},
  journal={Physical Review Letters},
  volume={107},
  number={6},
  pages={063601},
  year={2011}
}

@article{kienzler2015quantum,
  title={Quantum harmonic oscillator state synthesis by reservoir engineering},
  author={Kienzler, Dominik and Lo, Hoi-Yin and Keitch, Brenden and de Clercq, Leslie and Leupold, Fabian and Lindenfelser, Florian and Marinelli, Mario and Negnevitsky, Vlad and Home, Jonathan P.},
  journal={Science},
  volume={347},
  number={6217},
  pages={53--56},
  year={2015}
}

@book{serafini2024quantum,
  title={Quantum Continuous Variables: A Primer of Theoretical Methods},
  author={Serafini, Alessio},
  publisher={CRC Press},
  year={2024}
}

@book{walls2008quantum,
  title={Quantum Optics},
  author={Walls, Daniel F. and Milburn, Gerard J.},
  publisher={Springer},
  year={2008}
}

@article{japcsi,
    author = {Balka, Richárd and Csordás, András and Homa, Gábor},
    title = {Positivity and Entanglement of Polynomial Gaussian Integral Operators},
    journal = {Progress of Theoretical and Experimental Physics},
    volume = {2024},
    number = {10},
    pages = {103A01},
    year = {2024},
    month = {09},
   issn = {2050-3911}
   }

@article{japcsi2osc,
    author = {Homa, Gábor and Hamar, Dávid and Bernád, József Zsolt and Adam, Peter and Csordás, András},
    title = {An Exact Markovian Decoherence Dynamics of Two Interacting Harmonic Oscillators Coupled to a Bosonic Heat Bath},
    journal = {Progress of Theoretical and Experimental Physics},
    volume = {2025},
    number = {12},
    pages = {123A02},
    year = {2025},
    month = {12},
    issn = {2050-3911}
    }

@book{sakurai2017modern,
  title={Modern Quantum Mechanics},
  author={Sakurai, J. J. and Napolitano, J.},
  year={2017},
  publisher={Cambridge University Press}
}

@book{weinberg1995quantum,
  title={The Quantum Theory of Fields, Vol. I},
  author={Weinberg, S.},
  year={1995},
  publisher={Cambridge University Press}
}

@book{haake2010quantum,
  title={Quantum Signatures of Chaos},
  author={Haake, F.},
  year={2010},
  publisher={Springer}
}

@article{girardeau1960relationship,
  title={Relationship between systems of impenetrable bosons and fermions in one dimension},
  author={Girardeau, M.},
  journal={J. Math. Phys.},
  volume={1},
  pages={516--523},
  year={1960}
}

@book{peskin2018introduction,
  title={An Introduction to Quantum Field Theory},
  author={Peskin, M. E. and Schroeder, D. V.},
  year={2018},
  publisher={CRC Press}
}

@book{fetter1971quantum,
  title={Quantum Theory of Many-Particle Systems},
  author={Fetter, A. L. and Walecka, J. D.},
  year={2003},
  publisher={Dover Publications}
}

@book{altland2010condensed,
  title={Condensed Matter Field Theory},
  author={Altland, A. and Simons, B.},
  year={2010},
  publisher={Cambridge University Press}
}

@book{zinn2002quantum,
  title={Quantum Field Theory and Critical Phenomena},
  author={Zinn-Justin, J.},
  year={2002},
  publisher={Oxford University Press}
}

@article{leibfried2003quantum,
  title={Quantum dynamics of single trapped ions},
  author={Leibfried, D. and Blatt, R. and Monroe, C. and Wineland, D.},
  journal={Rev. Mod. Phys.},
  volume={75},
  pages={281--324},
  year={2003}
}

@article{bloch2008many,
  title={Many-body physics with ultracold gases},
  author={Bloch, I. and Dalibard, J. and Zwerger, W.},
  journal={Rev. Mod. Phys.},
  volume={80},
  pages={885--964},
  year={2008}
}

@article{aspelmeyer2014cavity,
  title={Cavity optomechanics},
  author={Aspelmeyer, M. and Kippenberg, T. J. and Marquardt, F.},
  journal={Rev. Mod. Phys.},
  volume={86},
  pages={1391--1452},
  year={2014}
}

@article{gu2017microwave,
  title={Microwave photonics with superconducting quantum circuits},
  author={Gu, X. and Kockum, A. F. and Miranowicz, A. and Liu, Y. X. and Nori, F.},
  journal={Physics Reports},
  volume={718--719},
  pages={1--102},
  year={2017}
}

@book{haroche2006exploring,
  title={Exploring the Quantum: Atoms, Cavities, and Photons},
  author={Haroche, S. and Raimond, J.-M.},
  year={2006},
  publisher={Oxford University Press}
}

@article{caves1981quantum,
  title={Quantum limits on noise in linear amplifiers},
  author={Caves, C. M.},
  journal={Phys. Rev. D},
  volume={26},
  pages={1817--1839},
  year={1982}
}

@article{Giedke,
 title = {Separability properties of three-mode Gaussian states},
  author = {Giedke, G. and Kraus, B. and Lewenstein, M. and Cirac, J. I.},
  journal = {Phys. Rev. A},
  volume = {64},
  issue = {5},
  pages = {052303},
  numpages = {10},
  year = {2001},
  month = {Oct},
  publisher = {American Physical Society}
}

@book{Abramowitz1964,
  title     = {Handbook of Mathematical Functions with Formulas, Graphs, and Mathematical Tables},
  author    = {Milton Abramowitz and Irene A. Stegun},
  year      = {1964},
  publisher = {Dover Publications},
  address   = {New York},
  series    = {National Bureau of Standards Applied Mathematics Series},
  volume    = {55}

}

@article{PAGANI2023169361,
title = {Time evolution of density matrices as a theory of random surfaces},
journal = {Annals of Physics},
volume = {455},
pages = {169361},
year = {2023},
%issn = {0003-4916},
%doi = {https://doi.org/10.1016/j.aop.2023.169361},
%url = {https://www.sciencedirect.com/science/article/pii/S0003491623001471},
author = {Carlo Pagani and Martin Reuter},
}

@article{Choi,
title = {Choi representation of completely positive maps in brief},
title = {},
author = {Gábor Homa and Adrian Ortega and Mátyás Koniorczyk},
pages = {1123--1133},
volume = {79},
number = {12},
journal = {Zeitschrift für Naturforschung A},
year = {2024},
}

@article{Mercer,
author = {Mercer, J.},
title = {XVI.~Functions of positive and negative type, and their connection with the theory of integral equations},
journal = {Philos.~Trans.~R.~Soc.~A},
volume = {209},
pages = {415-446},
year = {1909},

}

@misc{lvovsky2020productionapplicationsnongaussianquantum,
      title={Production and applications of non-Gaussian quantum states of light}, 
      author={A. I. Lvovsky and Philippe Grangier and Alexei Ourjoumtsev and Valentina Parigi and Masahide Sasaki and Rosa Tualle-Brouri},
      year={2020},
      eprint={2006.16985},
      archivePrefix={arXiv},
      primaryClass={quant-ph},
     
}

@article{Fiurasek2009,
  author  = {Fiur\'a\v{s}ek, Jarom\'ir},
  title   = {Engineering Quantum Operations on Traveling Light Beams by Multiple Photon Addition and Subtraction},
  journal = {Physical Review A},
  volume  = {80},
  pages   = {053822},
  year    = {2009},
  doi     = {10.1103/PhysRevA.80.053822}
}

@article{Ourjoumtsev2007,
  author  = {Ourjoumtsev, A. and Dantan, A. and Tualle-Brouri, R. and Grangier, P.},
  title   = {Increasing Entanglement between Gaussian States by Coherent Photon Subtraction},
  journal = {Physical Review Letters},
  volume  = {98},
  pages   = {030502},
  year    = {2007},
  doi     = {10.1103/PhysRevLett.98.030502}
}

@article{Lund2024,
  author  = {Lund, M. M. and Yang, F. and Christiansen, V. R. and Kornovan, D. and M{\o}lmer, K.},
  title   = {Subtraction and Addition of Propagating Photons by Two-Level Emitters},
  journal = {Physical Review Letters},
  volume  = {133},
  pages   = {103601},
  year    = {2024},
  doi     = {10.1103/PhysRevLett.133.103601}
}

@article{Lee2011,
  author  = {Lee, Su-Yong and Ji, Se-Wan and Kim, Ho-Joon and Nha, Hyunchul},
  title   = {Enhancing Quantum Entanglement for Continuous Variables by a Coherent Superposition of Photon Subtraction and Addition},
  journal = {Physical Review A},
  volume  = {84},
  pages   = {012302},
  year    = {2011},
  doi     = {10.1103/PhysRevA.84.012302},
  eprint  = {1107.2771}
}

@article{Noel2025nonlinear,
  author   = {No{\"e}l, Paul and Schlitz, Richard and Karad{\v{z}}a, Emir and Lambert, Charles-Henri and Nessi, Luca and Binda, Federico and Gambardella, Pietro},
  title    = {Nonlinear longitudinal and transverse magnetoresistances due to current-induced magnon creation-annihilation processes},
  journal  = {Physical Review Letters},
  volume   = {134},
  number   = {14},
  pages    = {146701},
  year     = {2025},
  publisher= {American Physical Society},
  doi      = {10.1103/PhysRevLett.134.146701}
}

@article{Enzian2021,
  title = {Single-Phonon Addition and Subtraction to a Mechanical Thermal State},
  author = {Enzian, G. and Price, J. J. and Freisem, L. and Nunn, J. and Janousek, J. and Buchler, B. C. and Lam, P. K. and Vanner, M. R.},
  journal = {Phys. Rev. Lett.},
  volume = {126},
  issue = {3},
  pages = {033601},
  numpages = {6},
  year = {2021},
  month = {Jan},
  publisher = {American Physical Society},
  doi = {10.1103/PhysRevLett.126.033601},
  url = {https://link.aps.org/doi/10.1103/PhysRevLett.126.033601}
}

@article{Fluhmann2019encoding,
  author   = {Fl{\"u}hmann, C. and Nguyen, T. L. and Marinelli, M. and Negnevitsky, V. and Mehta, K. K. and Home, J. P.},
  title    = {Encoding a qubit in a trapped-ion mechanical oscillator},
  journal  = {Nature},
  volume   = {566},
  number   = {7745},
  pages    = {513--517},
  year     = {2019},
  publisher= {Nature Publishing Group},
  doi      = {10.1038/s41586-019-0960-6}
}

@article{Meekhof1996,
  title = {Generation of Nonclassical Motional States of a Trapped Atom},
  author = {Meekhof, D. M. and Monroe, C. and King, B. E. and Itano, W. M. and Wineland, D. J.},
  journal = {Phys. Rev. Lett.},
  volume = {76},
  issue = {11},
  pages = {1796--1799},
  numpages = {0},
  year = {1996},
  month = {Mar},
  publisher = {American Physical Society},
  doi = {10.1103/PhysRevLett.76.1796},
  url = {https://link.aps.org/doi/10.1103/PhysRevLett.76.1796}
}

@article{Schneider2012,
doi = {10.1088/0034-4885/75/2/024401},
url = {https://doi.org/10.1088/0034-4885/75/2/024401},
year = {2012},
month = {jan},
publisher = {},
volume = {75},
number = {2},
pages = {024401},
author = {Schneider, Christian and Porras, Diego and Schaetz, Tobias},
title = {Experimental quantum simulations of many-body physics with trapped ions},
journal = {Reports on Progress in Physics}
}

@article{Cao2026,
  author       = {Cao, Yun-Long and Xu, Xiao-Ye and Li, Chuan-Feng and Guo, Guang-Can},
  title        = {Gaussian Resource Based Heralded Entangled State Generation Enhanced by Photon Addition and Subtraction},
  journal      = {Advanced Quantum Technologies},
  year         = {2026},
  note         = {In press},
  eprint       = {2602.24077}
}

@article{Dodonov2009,
  author  = {Dodonov, A. V. and Mizrahi, S. S.},
  title   = {Smooth Quantum-Classical Transition in Photon Subtraction and Addition Processes},
  journal = {Physical Review A},
  volume  = {79},
  pages   = {023821},
  year    = {2009},
  doi     = {10.1103/PhysRevA.79.023821},
  eprint  = {0901.4477}
}

@article{Luo2025,
  author       = {Luo, Haoyuan and Shah, Parth S. and Yang, Frank and Mirhosseini, Mohammad and Mahmoodian, Sahand},
  title        = {Dynamic Stimulated Emission for Deterministic Addition and Subtraction of Propagating Photons},
  journal      = {arXiv preprint},
  eprint       = {2512.09711},
  archivePrefix= {arXiv},
  year         = {2025}
}

@article{Lee2010,
  author  = {Lee, Su-Yong and Nha, Hyunchul},
  title   = {Quantum State Engineering by a Coherent Superposition of Photon Subtraction and Addition},
  journal = {Physical Review A},
  volume  = {82},
  pages   = {053812},
  year    = {2010},
  doi     = {10.1103/PhysRevA.82.053812},
  eprint  = {1009.2390}
}

@article{Walschaers2017,
  author  = {Walschaers, Mattia and Fabre, Claude and Parigi, Valentina and Treps, Nicolas},
  title   = {Statistical Signatures of Multimode Single-Photon Added and Subtracted States of Light},
  journal = {Physical Review A},
  volume  = {96},
  pages   = {053835},
  year    = {2017},
  doi     = {10.1103/PhysRevA.96.053835},
  eprint  = {1708.08412}
}

@article{Thapliyal2024,
  author  = {Thapliyal, Kishore and Pe\v{r}ina, Jan and Haderka, Ond\v{r}ej and Mich\'alek, V\'aclav and Machulka, Radek},
  title   = {Experimental Photon Addition and Subtraction in Multi-Mode and Entangled Optical Fields},
  journal = {Optics Letters},
  volume  = {49},
  pages   = {4521--4524},
  year    = {2024},
  doi     = {10.1364/OL.532242},
  eprint  = {2407.09269}
}

@article{Filippov2019,
  author  = {Filippov, S. N.},
  title   = {On Quantum Operations of Photon Subtraction and Photon Addition},
  journal = {Lobachevskii Journal of Mathematics},
  volume  = {40},
  pages   = {1470--1478},
  year    = {2019},
  doi     = {10.1134/S199508021910010X},
  eprint  = {1908.02207}
}

@article{Sharma2026,
  author       = {Sharma, Priyanka and Paris, Matteo G. A. and Olivares, Stefano},
  title        = {In Situ Characterization of a Photon-Subtraction Device via Heralding Counts and Homodyne Detection},
  journal      = {arXiv preprint},
  eprint       = {2609.10100},
  archivePrefix= {arXiv},
  year         = {2026}
}

@article{Takase2026,
  author       = {Takase, Kan and Endo, Mamoru and Hanamura, Fumiya and Hirota, Kazuki and Yabuno, Masahiro and Terai, Hirotaka and Miki, Shigehito and Kashiwazaki, Takahiro and Inoue, Asuka and Umeki, Takeshi and Marek, Petr and Filip, Radim and Asavanant, Warit and Furusawa, Akira},
  title        = {Tuning Wave-Particle Duality of Quantum Light by Generalized Photon Subtraction},
  journal      = {arXiv preprint},
  eprint       = {2602.21629},
  archivePrefix= {arXiv},
  year         = {2026}
}

@article{NeergaardNielsen2011b,
  author  = {Neergaard-Nielsen, Jonas S. and Takeuchi, Makoto and Wakui, Kentaro and Takahashi, Hiroki and Hayasaka, Kazuhiro and Takeoka, Masahiro and Sasaki, Masahide},
  title   = {Photon Subtraction from Traveling Fields---Recent Experimental Demonstrations},
  journal = {Progress in Informatics},
  volume  = {8},
  pages   = {5--18},
  year    = {2011},
  eprint  = {1104.4189}
}

@article{Olivares2006,
  author  = {Olivares, Stefano and Paris, Matteo G. A.},
  title   = {De-{G}aussification by Inconclusive Photon Subtraction},
  journal = {Laser Physics},
  volume  = {16},
  pages   = {1533--1550},
  year    = {2006},
  doi     = {10.1134/S1054660X06110077},
  eprint  = {quant-ph/0610160}
}

@article{Walschaers2021PRXQuantum.2.030204,
  title = {Non-Gaussian Quantum States and Where to Find Them},
  author = {Walschaers, Mattia},
  journal = {PRX Quantum},
  volume = {2},
  issue = {3},
  pages = {030204},
  numpages = {68},
  year = {2021},
  month = {Sep},
  publisher = {American Physical Society},
  
}

@article{Ra2020natureNonGaussianQuantum,
	author = {Young-Sik Ra and Adrien Dufour and Mattia Walschaers and Clément Jacquard and Thibault Michel and Claude Fabre and Nicolas Treps },
	title = {{N}on-{G}aussian quantum states of a multimode light field},
	journal = {{N}ature {P}hysics},
	volume = {16},
    issue = {3},
    pages = {144–147},
    year = {2020},
	note = {[Accessed 08-07-2026]},
}

@article{Jiri2024quantoptExperimentalPreparationmultiphoton,
	author = {Jiří Fadrný and Michal Neset and Martin Bielak and Miroslav Ježek and Jan Bílek and Jaromír Fiurášek},
	title = {Experimental preparation of multiphoton-added coherent states of light},
	journal = { npj Quantum Inf },
	  pages = {89},
    year = {2024},
	note = {[Accessed 13-07-2026]},
}

@article{Ansari2018opticaTailoringNonlinearProcesses,
author = {Vahid Ansari and John M. Donohue and Benjamin Brecht and Christine Silberhorn},
journal = {Optica},
number = {5},
pages = {534--550},
publisher = {Optica Publishing Group},
title = {Tailoring nonlinear processes for quantum optics with pulsed temporal-mode encodings},
volume = {5},
month = {May},
year = {2018},

}

@article{Briegel2001PhysRevLettClusterStates,
  title = {Persistent Entanglement in Arrays of Interacting Particles},
  author = {Briegel, Hans J. and Raussendorf, Robert},
  journal = {Phys. Rev. Lett.},
  volume = {86},
  issue = {5},
  pages = {910--913},
  numpages = {0},
  year = {2001},
  month = {Jan},
  publisher = {American Physical Society},
}

@article{Mika2022npjQuantumInfSingleModeQuantumNonGaussianLight,
  title = {Single-mode quantum non-Gaussian light from warm atoms},
  author = {Jaromír Mika and Lukáš Lachman and Tomáš Lamich and Radim Filip and Lukáš Slodička},
  journal = {npj Quantum Inf},
  volume = {8},
  pages = {123},
  year = {2022},
 }

@article{Deleglise2008NatureReconstructionOfNonClassicalCavityFieldStates,
  title = {Reconstruction of non-classical cavity field states with snapshots of their decoherence},
  author = {Samuel Deléglise and Igor Dotsenko and Clément Sayrin and Julien Bernu and Michel Brune and Jean-Michel Raimond and Serge Haroche },
  journal = {Nature},
  volume = {455},
  pages = {510-514},
  year = {2008},
 
}

@article{Nogues2000PhysRevAMeasurementOfANegativeValueForTheWignerFunction,
  title = {Measurement of a negative value for the Wigner function of radiation},
  author = {Nogues, G. and Rauschenbeutel, A. and Osnaghi, S. and Bertet, P. and Brune, M. and Raimond, J. M. and Haroche, S. and Lutterbach, L. G. and Davidovich, L.},
  journal = {Phys. Rev. A},
  volume = {62},
  issue = {5},
  pages = {054101},
  numpages = {4},
  year = {2000},
  month = {Oct},
  publisher = {American Physical Society},
 }

@article{Hofheinz2008NatureGenerationOfFockStates,
  title = {Generation of Fock states in a superconducting quantum circuit},
  author = {Max Hofheinz and E. M. Weig and M. Ansmann and Radoslaw C. Bialczak and Erik Lucero and M. Neeley and A. D. O’Connell and H. Wang and John M. Martinis and A. N. Cleland},
  journal = {Nature},
  volume = {454},
  pages = {310–314},
  year = {2008},
 
}

@article{
Vlastakis2013ScienceDeterministicallyEncodingQuantumInformation,
author = {Brian Vlastakis  and Gerhard Kirchmair  and Zaki Leghtas  and Simon E. Nigg  and Luigi Frunzio  and S. M. Girvin  and Mazyar Mirrahimi  and M. H. Devoret  and R. J. Schoelkopf },
title = {Deterministically Encoding Quantum Information Using 100-Photon Schrödinger Cat States},
journal = {Science},
volume = {342},
number = {6158},
pages = {607-610},
year = {2013},
}

@article{ChangPhysRevXObservationOfThreePhotonSpontaneousParametricDownConversion,
  title = {Observation of Three-Photon Spontaneous Parametric Down-Conversion in a Superconducting Parametric Cavity},
  author = {Chang, C. W. Sandbo and Sab\'{\i}n, Carlos and Forn-D\'{\i}az, P. and Quijandr\'{\i}a, Fernando and Vadiraj, A. M. and Nsanzineza, I. and Johansson, G. and Wilson, C. M.},
  journal = {Phys. Rev. X},
  volume = {10},
  issue = {1},
  pages = {011011},
  numpages = {16},
  year = {2020},
  month = {Jan},
  publisher = {American Physical Society},
  }

@article{Rakhubovsky2024ProgressInQuantumElectronicsQuantumNonGaussianOptomechanics,
author = {Andrey A. Rakhubovsky and Darren W. Moore and Radim Filip},
title = {Quantum non-Gaussian optomechanics and electromechanics},
journal = {Progress in Quantum Electronics},
volume = {93},
pages = {100495},
year = {2024},
issn = {0079-6727},

}

@article{Chen2006PhysRevASimultaneousAmplitudeAndPhaseDamping,
  title = {Simultaneous amplitude and phase damping of $x\text{\ensuremath{-}}p$ symmetric Gaussian states and their separability},
  author = {Chen, Xiao-yu},
  journal = {Phys. Rev. A},
  volume = {73},
  issue = {2},
  pages = {022307},
  numpages = {5},
  year = {2006},
  month = {Feb},
  publisher = {American Physical Society},
  doi = {10.1103/PhysRevA.73.022307},
  url = {https://link.aps.org/doi/10.1103/PhysRevA.73.022307}
}

@article{Hiroshima2001PhysRevADecoherence,
  title = {Decoherence and entanglement in two-mode squeezed vacuum states},
  author = {Hiroshima, Tohya},
  journal = {Phys. Rev. A},
  volume = {63},
  issue = {2},
  pages = {022305},
  numpages = {4},
  year = {2001},
  month = {Jan},
  publisher = {American Physical Society},
  doi = {10.1103/PhysRevA.63.022305},
  url = {https://link.aps.org/doi/10.1103/PhysRevA.63.022305}
}

@book{BreuerPetruccione2002OpenQuantumSystems,
  address = {Great Clarendon Street},
  author = {Breuer, H. P. and Petruccione, F.},
  nota = {GK: excellent textbook!},
  publisher = {Oxford University Press},
  title = {The theory of open quantum systems},
  year = 2002
}

@article{Bartlett2002ClassicalSimulation,
  title = {Efficient Classical Simulation of Continuous Variable Quantum Information Processes},
  author = {Bartlett, Stephen D. and Sanders, Barry C. and Braunstein, Samuel L. and Nemoto, Kae},
  journal = {Phys. Rev. Lett.},
  volume = {88},
  issue = {9},
  pages = {097904},
  numpages = {4},
  year = {2002},
  month = {Feb},
  publisher = {American Physical Society},
  doi = {10.1103/PhysRevLett.88.097904},
  url = {https://link.aps.org/doi/10.1103/PhysRevLett.88.097904}
}

@article{Lloyd1999QuantumComputation,
  title = {Quantum Computation over Continuous Variables},
  author = {Lloyd, Seth and Braunstein, Samuel L.},
  journal = {Phys. Rev. Lett.},
  volume = {82},
  issue = {8},
  pages = {1784--1787},
  numpages = {0},
  year = {1999},
  month = {Feb},
  publisher = {American Physical Society},
  doi = {10.1103/PhysRevLett.82.1784},
  url = {https://link.aps.org/doi/10.1103/PhysRevLett.82.1784}
}

@article{Mogyorosi2019PhysRevA.100.013851,
  title = {Single-step quantum state engineering in traveling optical fields},
  author = {Mogyorósi, Gábor and Adam, Peter and Molnar, Emese and Mechler, Matyas},
  journal = {Phys. Rev. A},
  volume = {100},
  issue = {1},
  pages = {013851},
  numpages = {9},
  year = {2019},
  month = {Jul},
  publisher = {American Physical Society},
  doi = {10.1103/PhysRevA.100.013851}
}

@article{BIAGI2022100414,
title = {Photon-by-photon quantum light state engineering},
journal = {Progress in Quantum Electronics},
volume = {84},
pages = {100414},
year = {2022},
issn = {0079-6727},
doi = {https://doi.org/10.1016/j.pquantelec.2022.100414},
author = {Nicola Biagi and Saverio Francesconi and Alessandro Zavatta and Marco Bellini}
}

@article{Liu_2015,
doi = {10.1088/1612-2011/12/4/045201},
year = {2015},
month = {feb},
publisher = {IOP Publishing},
volume = {12},
number = {4},
pages = {045201},
author = {Liu, Shi-You and Li, Ya-Zhou and Hu, Li-Yun and Huang, Jie-Hui and Xu, Xue-Xiang and Tao, Xiang-Yang},
title = {Nonclassical properties of Hermite polynomial excitation on squeezed vacuum and its decoherence in phase-sensitive reservoirs},
journal = {Laser Physics Letters},
}

@article{Xu_2015,
   title={Generating Hermite polynomial excited squeezed states by means of conditional measurements on a beam splitter},
   volume={32},
   ISSN={1520-8540},
   DOI={10.1364/josab.32.001146},
   number={6},
   journal={Journal of the Optical Society of America B},
   publisher={Optica Publishing Group},
   author={Xu, Xue-xiang and Yuan, Hong-chun and Fan, Hong-yi},
   year={2015},
   month=May, pages={1146}
}

@article{Tang_2015,
doi = {10.1088/1674-1056/24/3/034208},
year = {2015},
month = {mar},
publisher = {IOP Publishing},
volume = {24},
number = {3},
pages = {034208},
author = {Tang, Xu-Bing and Gao, Fang and Wang, Yao-Xiong and Kuang, Sen and Shuang, Feng},
title = {Non-Gaussian quantum states generation and robust quantum non-Gaussianity via squeezing field*},
journal = {Chinese Physics B}
}

@article{Saidi2026,
author = {Saidi, H. and Slaoui, A. and El Hadfi, H. and Ahl Laamara, R.},
title = {Multiple phase estimation with photon-added multi-mode coherent states of GHZ-type},
journal = {International Journal of Modern Physics A},
volume = {41},
number = {08},
pages = {2650036},
year = {2026}
}

@article{He2026,
author = {He, Ze-Shan and Zhao, Yukuan and Tian, Hao-Shu and Sun, Kai and Xu, Xiao-Ye and Li, Chuan-Feng and Guo, Guang-Can},
year = {2026},
month = {04},
pages = {},
title = {Experimental detection of entanglement in multimode Gaussian states from high-order intensity correlation moments},
doi = {10.48550/arXiv.2604.27567}
}

@article{Tatsi2022PhysRevA.105.053701,
  title = {Manipulating thermal light via displaced-photon subtraction},
  author = {Tatsi, G. and Canning, D. W. and Zanforlin, U. and Mazzarella, L. and Jeffers, J. and Buller, G. S.},
  journal = {Phys. Rev. A},
  volume = {105},
  issue = {5},
  pages = {053701},
  numpages = {10},
  year = {2022},
  month = {May},
  publisher = {American Physical Society},
  doi = {10.1103/PhysRevA.105.053701}
}

@article{Roh2026,
      title={Experimental Quantum Tomography of Multimode Gaussian States}, 
      author={Chan Roh and Geunhee Gwak and Young-Do Yoon and Young-Sik Ra},
      year={2026},
      eprint={2603.21380},
      archivePrefix={arXiv},
      primaryClass={quant-ph},
      url={https://arxiv.org/abs/2603.21380},
      doi = {10.48550/arXiv.2603.21380}
}

@Article{universe10080306,
AUTHOR = {Hsiang, Jen-Tsung and Cho, Hing-Tong and Hu, Bei-Lok},
TITLE = {Graviton Physics: A Concise Tutorial on the Quantum Field Theory of Gravitons, Graviton Noise, and Gravitational Decoherence},
JOURNAL = {Universe},
VOLUME = {10},
YEAR = {2024},
NUMBER = {8},
ARTICLE-NUMBER = {306},
ISSN = {2218-1997}

}

@article{giedke2001distillability,
  title={Distillability Criterion for all Bipartite Gaussian States},
  author={Giedke, G. and Duan, L.-M. and Cirac, J. I. and Zoller, P.},
  journal={Quantum Information \& Computation},
  volume={1},
  number={3},
  pages={79--86},
  year={2001},
  publisher={Rinton Press},
  eprint={quant-ph/0104072},
  archivePrefix={arXiv}
}

@article{Horodecki2009,
  title = {Quantum entanglement},
  author = {Horodecki, Ryszard and Horodecki, Pawe{\l} and
            Horodecki, Micha{\l} and Horodecki, Karol},
  journal = {Rev. Mod. Phys.},
  volume = {81},
  pages = {865--942},
  year = {2009},
  doi = {10.1103/RevModPhys.81.865}
}

@article{vernon,
    author = {Richard P Feynman and Frank L Vernon},
    title = {The theory of a general quantum system interacting with a linear dissipative system},
    journal = {Annals of Physics},
volume={24},
    year = {1963},
pages={118--173},
}

\end{document}